\documentclass{amsproc}

\usepackage[T2A]{fontenc}
\usepackage[cp1251]{inputenc}
\usepackage[english]{babel}

\usepackage{amsfonts}
\usepackage{amsthm}
\usepackage{amssymb}
\usepackage{amsmath}
\usepackage{amsopn}
\usepackage{bm}
\usepackage{array}
\usepackage{color}
\usepackage{mathdots}
\usepackage{graphicx}
\usepackage{longtable}

\newtheorem{prop}{Proposition}
\newtheorem{theo}{Theorem}
\newtheorem{lemma}{Lemma}
\newtheorem{cor}{Corollary}

\theoremstyle{definition}

\newtheorem{rem}{Remark}

\usepackage{cite}
\allowdisplaybreaks[4]

\newcommand{\Lmatr}{\mathsf{L}}

\DeclareMathOperator{\ad}{ad}
\DeclareMathOperator{\Ad}{Ad}
\DeclareMathOperator{\spanOp}{span}
\DeclareMathOperator{\tr}{tr}

\DeclareMathOperator{\adj}{adj}

\newcommand{\x}{\mathrm{x}}
\newcommand{\rmx}{\mathrm{x}}
\newcommand{\rmt}{\mathrm{t}}
\newcommand{\rmr}{\mathrm{r}}
\newcommand{\rmh}{\mathrm{h}}

\newcommand{\sfb}{\mathsf{b}}
\newcommand{\sfx}{\mathsf{x}}
\newcommand{\sft}{\mathsf{t}}
\newcommand{\sfr}{\mathsf{r}}

\newcommand{\rmd}{\mathrm{d}}
\DeclareMathOperator{\modR}{mod}

\newcommand{\M}{\mathcal{M}}

\newcommand{\rmb}{\mathrm{b}}
\newcommand{\rma}{\mathrm{a}}

\newcommand{\A}{\mathrm{A}}
\newcommand{\Pp}{\textsf{P}}
\newcommand{\Qp}{\textsf{Q}}
\newcommand{\Rp}{\textsf{R}}
\newcommand{\Cp}{\textsf{C}}

\newcommand{\Dp}{\textsf{D}}
\newcommand{\Np}{\textsf{N}}

\renewcommand{\H}{\textsf{H}}
\newcommand{\X}{\textsf{X}}
\newcommand{\Y}{\textsf{Y}}
\newcommand{\Z}{\textsf{Z}}

\newcommand{\Complex}{\mathbb{C}}
\newcommand{\Real}{\mathbb{R}}
\newcommand{\Natural}{\mathbb{N}}
\newcommand{\Integer}{\mathbb{Z}}

\newcommand{\Jac}{\mathrm{Jac}}

\DeclareMathOperator*{\res}{res}
\DeclareMathOperator{\const}{const}

\newcommand{\w}{\textsf{w}}

\newcommand{\JFr}{\mathfrak{J}}

\newcommand{\mFr}{\mathfrak{m}}

\newcommand{\kFr}{\mathfrak{k}}
\newcommand{\I}{\mathcal{I}}
\newcommand{\J}{\mathcal{J}}

\DeclareMathOperator{\ReN}{\mathrm{Re}}
\DeclareMathOperator{\ImN}{\mathrm{Im}}

\allowdisplaybreaks

\title[Exact quasi-periodic solutions to the sine(sinh)-Gordon equations]
{Exact quasi-periodic solutions to the sine(sinh)-Gordon equations: 
The method for computation and analysis}
\author{Julia Bernatska}
\address{University of Connecticut, Department of Mathematics}
\email{julia.bernatska@uconn.edu}

\begin{document}

\begin{abstract}
The sine(sinh)-Gordon hierarchy of integrable Hamiltonian systems
is described in detail, and all 
dynamic variables are expressed in terms of the $\wp$-functions 
that uniformize the associated spectral curve.
Quasi-periodic solutions to the sine(sinh)-Gordon equations are obtained
 in terms of the function $\wp_{1,2g-1}$, reality conditions are revised,
 and a method of computation and analysis 
 is presented. The proposed method is designed to analyze solutions 
 by means of the Hamiltonian technique, which is illustrated in genera one and two.
\end{abstract}




\maketitle

\section{Introduction}

The present paper  focuses on the sine-Gordon equation 
\begin{equation}\label{SGEq}
\phi_{\rmt \x} = 4 \sin \phi,
\end{equation} 
and also the sinh-Gordon equation
\begin{equation}\label{ShGEq}
\phi_{\rmt \x} = 4 \sinh \phi.
\end{equation} 
These equations are known to be completely integrable \cite{AKNS1973}
and to possess soliton solutions. The sine-Gordon equation describes
various  physical phenomena (see, for example,
\cite{SGAppl1936}) and has thus been widely studied in the literature.

The simplest soliton solutions, such as kinks and breathers expressed
in terms of elementary functions, were obtained in \cite{AKNS1973} using the inverse scattering method.  
In \cite{Hirota1972}, multi-soliton solutions to the sine-Gordon equation were obtained
via Hirota's method, and the interaction of two solitons was analyzed in detail.
Three types of traveling-wave solutions  
in terms of elliptic functions were found in \cite{CMPSS1978}
by using the Lamb ansatz $\tan \big(\tfrac{1}{4}\phi(\rmx,\rmt)\big)\,{=}\,F(\rmx) G(\rmt)$;
these solutions are interpreted 
as different forms of oscillatory behavior of a one-dimensional  Josephson transmission line.

A generic quasi-periodic solution to the sine-Gordon equation was suggested in
\cite{KK1976}. This solution is expressed in terms of the theta function and has the form, 
see \cite[Eq.(4.2.25)]{bbeim1994},
\begin{equation}\label{ThetaSol}
\phi(\x, \rmt) =  2 \imath \log \frac{\theta( \imath (\bm{U}\x + \bm{W} \rmt) + \bm{D} + \imath \pi \Delta)}
{\theta( \imath(\bm{U}\x + \bm{W} \rmt) + \bm{D} )},
\end{equation}
where $\bm{U}$, $\bm{W}$, $\bm{D}$, $\Delta$ are vectors defined through 
periods on the corresponding spectral curve. The curve is hyperelliptic,
of an arbitrary genus.

In \cite{ForMcL1982}, solutions to the sine-Gordon equation with one period along the $\rmx$-axis are analyzed
based on the spectral theory of linear operators arising in the inverse scattering method.
Traveling-wave solutions and separable solutions (of the form of  the Lamb ansatz) 
are classified up to genus two. 
Reality conditions are also addressed  in \cite{ForMcL1982},
where the following statement is presented: for real-valued solutions to the sine-Gordon equation,
the spectral curve is required to have branch points split into pairs $e_{2i}$, $e_{2i-1}$
which are either real and negative, $e_{2i}\,{<}\,e_{2i-1}\,{<}\,0$, or occur in complex conjugate pairs
$\bar{e}_{2i}\,{=}\, e_{2i-1}$.

In \cite{DN1982}, the analysis of real ovals produced by the anti-holomorphic involution on the 
spectral curve is used to derive reality conditions for genus two.  
This analysis yields a relation between $\bm{D}$, $\Delta$, and a half-period---varying according to the number of
real ovals---which ensures that \eqref{ThetaSol} is real-valued. 
In the genus two case, this relation is interpreted via the characteristics of the theta function;
specifically, in the presence of real branch points, quarter-integer characteristics are required

In \cite{Date1982},  fixed-point-free involutions of the spectral curve 
are used to single out real-valued solutions; reality condition are described through involutions 
of the canonical homology basis. In this way, the existence of real-valued solutions to the sine-Gordon equation
in arbitrary genus was proven.

The problem of choosing tan argument that ensures the solution is real-valued
is addressed in \cite{EF1985}.
It was shown that each pair of real (and negative) branch points 
causes addition of a quarter in the corresponding component of $\bm{D}$.

In \cite{bb1982}, solution \eqref{ThetaSol} in genus two is used to classify two-phase waves
that take the form of the Lamb ansatz for specific  values of parameters.
These waves are expressed via  theta functions with half-integer characteristics
and represent physically meaningful solutions.

Recently, some improvements in computing 
multi-soliton solutions using the inverse scattering method were suggested in \cite{ADM2010};
a review of  solutions expressed in elementary and elliptic functions can be found therein.
In addition to what mentioned in \cite{ADM2010},
two- and three-gap elliptic solutions  were constructed 
on the spectral curves that are two- and three-sheeted coverings 
of the elliptic curve (see \cite{Smi1991,Smi1997}).

Although  the sine(sinh)-Gordon equations are broadly illuminated in the literature,
quasi-periodic wave solutions have remained out of consideration.
Solutions expressed through elliptic functions are periodic, 
and represent a special case of the generic  solution \eqref{ThetaSol} in genus two.
In fact, the only analysis of quasi-periodic solutions is the interpretation of
 two-phase waves in \cite{bb1982}.

The purpose of the present paper is to provide an effective tool for working 
with quasi-periodic solutions. The proposed method is rooted in the Lie-algebraic approach to 
constructing the sine(sinh)-Gordon hierarchy, which is briefly presented in Sections 3 and 4.
Advantages of the method are  a simple and easy-derived 
connection between a finite-gap Hamiltonian system and its spectral curve 
(the latter serves as the Hamiltonian function in separated variables), and a straightforward
algebraic integration in terms of the $\wp$-functions associated with the spectral curve, see Section 5.
The  $\wp$-functions parametrize the phase space of such a Hamiltonian system and provide an elegant form 
for solutions to the integrable hierarchy. Computational techniques for working with these functions
have been recently developed, see  \cite{BerCompWP2024}.

Finite-gap solutions to the sine(sinh)-Gordon equations are expressed in terms of the function $\wp_{1,2g-1}$
for any  hyperelliptic genus $g$, see Theorems\;\ref{T:SineGSol} and \ref{T:SinhGSol} below.
The proposed solution  is equivalent to \eqref{ThetaSol}, see \cite[\S\,216]{bakerAF}. 
In Section 6, the addition law is used to find all real-valued solutions with half-period shifts.
This covers all solutions to the sinh-Gordon equation,  along with the most interesting 
(and physically meaningful, as noted in \cite{bb1982}) solutions to the sine-Gordon equation.

The proposed method is designed to analyze solutions using the Hamiltonian technique.
Parameters of the spectral curve serve as integrals of motion---specifically, constraints or Hamiltonians. 
Depending on the choice of these parameters, different wave propagation regimes arise,
which can be easily investigated, as illustrated in Section 7. 
The genus of the spectral curve is limited only by computational capacity.

While the method works with quasi-periodic wave solutions, 
soliton, traveling-wave, and separated solutions can be obtained from 
the generic finite-gap solution by degenerating the spectral curve, 
which is a subject for further investigation.

\section{Preliminaries}
\subsection{Hyperelliptic curves}
We work with the canonical form of a  hyperelliptic curve $\mathcal{V}$ of genus $g$, namely 
\begin{equation}\label{V22g1Eq}
f(x,y) \equiv -y^2 + x^{2 g+1} + \sum_{i=0}^{2g} \lambda_{2i+2} x^{2g-i} = 0.
\end{equation}
In terms of  branch points $(e_i,0)$, or simply $e_i$, $i=1$, \ldots, $2g+1$, $f$ has the form
\begin{equation}\label{V22g1EqBP}
f(x,y) = -y^2 + \prod_{i=1}^{2g+1} (x- e_i).
\end{equation}
One more branch point $e_0$  is located at infinity, and serves as the basepoint.
We assume that all branch points are distinct, and so the curve is not degenerate, that is, the genus equals $g$.
Finite branch points are enumerated in the ascending order of real and imaginary parts.
The Riemann surface of $\mathcal{V}$ is constructed using a
 monodromy path that passes through all branch points in order.
 This path starts  and ends at infinity (see the orange line in fig.~\ref{cyclesOddCC}). 
For further details on constructing  Riemann surfaces of hyperelliptic curves 
suitable for computation, see \cite[Sect.\,3]{BerCompWP2024}.

Following H.\,Baker \cite[p.\,297]{bakerAF}, a homology basis is defined as follows.
Cuts are made between points $e_{2k-1}$ and $e_{2k}$ for $k\,{=}\,1$, \ldots, $g$, with
an additional cut extending from  $e_{2g+1}$ to infinity.
Canonical homology cycles are then defined:
each $\mathfrak{a}_k$-cycle ($k=1$, \ldots $g$) encircles the cut $(e_{2k-1},e_{2k})$ counter-clockwise,
while each $\mathfrak{b}_k$-cycle emerges from the cut $(e_{2g+1},\infty)$ 
and enters the cut $(e_{2k-1},e_{2k})$, see  fig.\;\ref{cyclesOddCC}. 

\begin{figure}[h]
\includegraphics[width=0.67\textwidth]{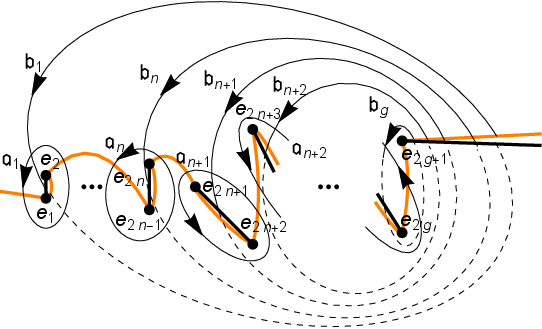} 
\caption{Cuts and cycles on a hyperelliptic curve.} \label{cyclesOddCC}
\end{figure}

Fig.~\ref{cyclesOddCC} represents the spectral curve of 
 a $g$-gap system in the sine-Gordon hierarchy, see Section~\ref{s:RealCond} for details.
Parameters $\lambda_k$  of the curve $\mathcal{V}$ are real, and all but one of the finite branch points  are
complex conjugate pairs.

Let first kind differentials $\rmd u = (\rmd u_1,\rmd u_3, \dots, \rmd u_{2g-1})^t$ 
and second kind differentials $\rmd r = (\rmd r_1,\rmd r_3, \dots, \rmd r_{2g-1})^t$ form a system of
 associated differentials, see \cite[\S\,138]{bakerAF}, and be
defined as in \cite[\textit{Ex.\,i}, p.\,195]{bakerAF}.
Actually, 
\begin{subequations}\label{K1K2DifsGen}
\begin{align}
& \rmd u_{2n-1} =  \frac{x^{g-n} \rmd x}{\partial_y f(x,y)},\quad n=1,\dots, g,\label{K1DifsGen} \\
& \rmd r_{2n-1} =  \frac{ \rmd x}{\partial_y f(x,y)} 
 \sum_{j=1}^{2n-1} (2n-j) \lambda_{2j-2} x^{g+n-j},\  \lambda_0=1,\quad n=1,\dots, g. \label{K2DifsGen}
\end{align}
\end{subequations}
The indices of $\rmd u_{2n-1}$ and $\rmd r_{2n-1}$ 
represent the orders of the zeros and poles at infinity, respectively.

The first kind differentials are not normalized.
The corresponding periods along the canonical cycles $\mathfrak{a}_k$,  $\mathfrak{b}_k$, $k=1$, \ldots, $g$, 
 are defined as follows:
\begin{gather}\label{FKD}
  \omega_k = \oint_{\mathfrak{a}_k} \rmd u,\qquad\qquad
  \omega'_k = \oint_{\mathfrak{b}_k} \rmd u.
\end{gather}
The vectors $\omega_k$, $\omega'_k$ 
form  first kind period matrices $\omega$, $\omega'$, respectively.
Similarly, second kind period matrices $\eta$, $\eta'$ are composed of columns
\begin{gather}\label{SKD}
  \eta_k = \oint_{\mathfrak{a}_k} \rmd r,\qquad\qquad
  \eta'_k = \oint_{\mathfrak{b}_k} \rmd r.
\end{gather}

The normalized first kind period matrices are $1_g$, $\tau$,
where $1_g$ denotes the identity matrix of order $g$, 
and $\tau = \omega^{-1}\omega'$.  Matrix $\tau$ is symmetric with a positive imaginary part: 
$\tau^t=\tau$, $\ImN \tau >0$,
that is, $\tau$ belongs to the Siegel upper half-space. 
The normalized holomorphic differentials are denoted by
\begin{gather*}
 \rmd v = \omega^{-1} \rmd u.
\end{gather*}

\subsection{Abel's map}
The vectors $\omega_k$, $\omega'_k$  form a period lattice $\{\omega, \omega'\}$, and
 $\Jac(\mathcal{V}) \,{=}\, \Complex^g/\{\omega, \omega'\}$ is the Jacobian variety of $\mathcal{V}$.
Let $u=(u_1$, $u_3$, \ldots, $u_{2g-1})^t$ be not normalized coordinates of $\Jac(\mathcal{V})$.

The Abel map is defined by
\begin{gather*}
 \mathcal{A}(P) = \int_{\infty}^P \rmd u,\qquad P=(x,y)\in \mathcal{V},
\end{gather*}
and on  a positive divisor $D = \sum_{i =1}^n (x_i,y_i)$  by $\mathcal{A}(D) =  \sum_{i =1}^n  \mathcal{A}(P_i)$.
The Abel map is one-to-one on the $g$-th symmetric power of the curve.

\subsection{Theta function}
The Riemann theta function  is defined by 
\begin{gather}\label{ThetaDef}
 \theta(v;\tau) = \sum_{n\in \Integer^g} \exp \big(\imath \pi n^t \tau n + 2\imath \pi n^t v\big).
\end{gather}
where $v = \omega^{-1}u$ are normalized coordinates.
Theta function with characteristic $[\varepsilon]$ is defined by
\begin{equation*}
 \theta[\varepsilon](v;\tau) = \exp\big(\imath \pi (\varepsilon'{}^t/2) \tau (\varepsilon'/2)
 + 2\imath \pi  (v+\varepsilon/2)^t \varepsilon'/2\big) \theta(v+\varepsilon/2 + \tau \varepsilon'/2;\tau),
\end{equation*}
where $[\varepsilon]= (\varepsilon', \varepsilon)^t$
 is a $2\times g$ matrix, all components of $\varepsilon$ and $\varepsilon'$
are real values within the interval $[0,2)$. Modulo $2$ addition
is defined on characteristics. 

Every point $u$ within the fundamental domain of  $\Jac(\mathcal{V})$ 
can be represented by its characteristic $[\varepsilon]$ as follows
\begin{equation*}
u =  \tfrac{1}{2}  \omega \varepsilon +  \tfrac{1}{2}  \omega' \varepsilon'.
\end{equation*}
Characteristics with values $0$ and $1$ correspond to
half-periods, which are Abel images of divisors composed of branch points.
Such a characteristic $[\varepsilon]$ is odd whenever $\varepsilon^t \varepsilon'  \,{=}\, 1$ ($\modR 2$)
and even whenever $\varepsilon^t \varepsilon' \,{=}\, 0$ ($\modR 2$). Theta function with characteristic
has the same parity as its characteristic.

\subsection{Sigma function  and $\wp$-functions}\label{ss:}
The modular invariant entire function on $\Complex^g \,{\supset}\, \Jac(\mathcal{V})$  is called the sigma function, 
which we define after \cite[Eq.(2.3)]{belHKF}:
\begin{equation}\label{SigmaThetaRel}
\sigma(u) = C \exp\big({-}\tfrac{1}{2} u^t \varkappa u\big) \theta[K](\omega^{-1} u;  \omega^{-1} \omega'),
\end{equation}
where $[K]$ is the characteristic of the vector $K$ of Riemann constants, and 
$\varkappa = \eta \omega^{-1}$ is a symmetric matrix.

In what follows we use multiply periodic $\wp$-functions 
\begin{gather*}
\wp_{i,j}(u) = -\frac{\partial^2 \log \sigma(u) }{\partial u_i \partial u_j },\qquad
\wp_{i,j,k}(u) = -\frac{\partial^3 \log \sigma(u) }{\partial u_i \partial u_j \partial u_k},
\end{gather*}
which are defined on $\Jac(\mathcal{V}) \backslash \Sigma$, 
where $\Sigma = \{u \mid \sigma(u)=0\}$ denotes the theta divisor 
(as introduced in \cite[p.\,38]{DN1982})
 in not normalized coordinates.

\subsection{Jacobi inversion problem}
A solution to the Jacobi inversion problem on a hyperelliptic curve is proposed in  \cite[Art.\;216]{bakerAF},
see also \cite[Theorem 2.2]{belHKF}.
Let $u = \mathcal{A}(D )$ be the Abel image of  a non-special positive divisor  
$D \in \mathcal{V}^g$. Then $D$ is uniquely defined by the system of equations 
\begin{subequations}\label{EnC22g1}
\begin{align}
&\mathcal{R}_{2g}(x;u) \equiv x^{g} -  \sum_{i=1}^{g} x^{g-i}  \wp_{1,2i-1}(u) = 0, \label{R2g}\\ 
&\mathcal{R}_{2g+1}(x,y;u) \equiv 2 y + \sum_{i=1}^{g} x^{g-i}  \wp_{1,1,2i-1}(u) = 0. \label{R2g1}
\end{align}
\end{subequations}

\subsection{Characteristics and partitions}\label{ss:CharPart}
Let $S = \{0,1,2,\dots, 2g+1\}$ be the set of indices of all branch points of a 
hyperellipitic curve of genus $g$, and $0$ stands for the branch point at infinity.
According to \cite[\S\,202]{bakerAF}, 
all characteristics of half-periods are represented by 
partitions of $S$ of the form $\I_\mFr\cup \J_\mFr$ with $\I_\mFr = \{i_1,\,\dots,\, i_{g+1-2\mFr}\}$
and $\J_\mFr = \{j_1,\,\dots,\, j_{g+1+2\mFr}\}$, where $\mFr$ runs from $0$ to $[(g+1)/2]$, 
and $[\cdot]$ denotes the integer part. 
The index $0$ is usually omitted when listing sets, as well as when calculating the cardinality of sets.

Denote by $[\varepsilon(\I)] = \sum_{i\in\I} [\varepsilon_i]$ $(\modR 2)$ the characteristic of
\begin{gather*}
 \mathcal{A} (\I) = \sum_{i\in\I} \mathcal{A}(e_i) = 
 \omega \Big(\tfrac{1}{2} \varepsilon(\I)  + \tfrac{1}{2} \tau \varepsilon'(\I) \Big).
\end{gather*}
Characteristics corresponding to $2g+1$ branch points serve as a basis for constructing
all $2^{2g}$ half-period characteristics.
Below, a partition is  referred to by the part of less cardinality, denoted by $\I$.
Let
$$[\I] = [\varepsilon(\I)] + [K] \ (\modR 2),$$ 
and $[K]$  equals 
the sum of $g$ odd characteristics of branch points.
In the basis of canonical cycles introduced by fig.~\ref{cyclesOddCC} we have
\begin{gather}\label{Kchar}
 [K] = \sum_{k=1}^g [\varepsilon_{2k}].
\end{gather}

Let $\I_\mFr\cup \J_\mFr$ be a partition introduced above; then $[\I_\mFr]=[\J_\mFr]$.
Characteristics $[\I_\mFr]$ are even when $\mFr$ is even  and odd when $\mFr$ is odd.
According to the Riemann vanishing theorem, $\theta(v+\mathcal{A} (\I_\mFr)+K; \tau)$ 
vanishes to order $\mFr$ at $v\,{=}\,0$.
The integer $\mFr$ is called the \emph{multiplicity}.
Characteristics of multiplicity $0$ are called \emph{non-singular even characteristics},
while those of  multiplicity $1$ are called \emph{non-singular odd characteristics}.
All other characteristics are called \emph{singular}.

\section{Integrable systems on coadjoint orbits of a loop group}\label{s:KdVHier}
The sine-Gordon equation arises within the hierarchy of integrable Hamiltonian systems 
on coadjoint orbits iof the loop algebra $\mathfrak{su}(2)$. 
We briefly recall this construction as presented in \cite{Holod84}
and revised in \cite{BerHol07}.

While similarities to  \cite[Section\;3,\,4]{BerKdV2024}  may be observed,  
 the sine(sinh)-Gordon hierarchy arises on different coadjoint orbits; 
therefore, the present exposition remains essentially  distinct.

\subsection{Affine Lie algebra}
Let $\widetilde{\mathfrak{g}} = \mathfrak{g} \otimes
\mathfrak{L}(z,z^{-1})$, where $\mathfrak{L}(z,z^{-1})$ denotes the algebra of Laurent series in~$z$,
and $\mathfrak{g}=\mathfrak{sl}(2,\Complex)$ has the standard basis
\begin{gather*}
\H = \frac{1}{2} \begin{pmatrix} 1 & 0 \\ 0 & -1  \end{pmatrix},\quad
\X = \begin{pmatrix} 0 & 1 \\ 0 & 0  \end{pmatrix},\quad
\Y = \begin{pmatrix} 0 & 0 \\ 1 & 0  \end{pmatrix}.
\end{gather*}
The algebra $\widetilde{\mathfrak{g}}$ is \emph{principally graded} by the grading operator 
\begin{gather}\label{PrGrad}
 \mathfrak{d} = 2 z \frac{\rmd }{\rmd z}  + \ad_\H,
 \end{gather}
where $\ad$ denotes the adjoint operator in $\mathfrak{g}$, and  $\forall\, \Z \in \mathfrak{g}$\ \   $\ad_\H \Z = [\H,\Z] $. 

Let $\{\X_{2m-1} \,{=}\, z^{m-1} \X,\, \Y_{2m-1} \,{=}\,  z^{m} \Y,\, \H_{2m} \,{=}\, z^m \H \mid m \in \Integer\}$
form a basis of $\widetilde{\mathfrak{g}}$. 
Let $\mathfrak{g}_{\ell}$ be the eigenspace of $\mathfrak{d}$ of degree $\ell$.
Actually, $\mathfrak{g}_{2m-1} = \spanOp \{\X_{2m-1},\, \Y_{2m-1}\}$, and $\mathfrak{g}_{2m} = \spanOp \{\H_{2m}\}$.
We also denote the basis elements in $\widetilde{\mathfrak{g}}$  by $\Z_{a;\ell}$, $a=1$, $2$, $3$, $\ell \in \Integer$,
where $\mathfrak{d} \Z_{a;\ell} = \ell \Z_{a;\ell}$, and  $\Z_{1;\ell}= \H_\ell$, $\Z_{2;\ell}=\Y_\ell$, $\Z_{3;\ell}=\X_\ell$.

With a fixed positive integer $N$, let the bilinear form be defined by
\begin{gather}\label{BiLinFSinG}
\forall A(z), B(z) \in \widetilde{\mathfrak{g}} \qquad 
\langle A(z), B(z) \rangle = \res_{z=0} z^{-N} \tr A(z) B(z).
\end{gather}
By means of $\langle \Z_{a;\ell}, \Z^{\ast}_{b;\ell'} \rangle =  \delta_{a,b} \delta_{\ell,\ell'}$
the basis $\{\Z^{\ast}_{a;\ell}  \mid a=1,2,3,\, \ell\in \Integer\}$ of the dual algebra $\widetilde{\mathfrak{g}}^\ast$ 
is introduced. In more detail,
\begin{gather*}
\Z_{1;2m}^\ast = 2 \H_{2N-2m-2},\quad
\Z_{2;2m-1}^\ast = \X_{2N-2m-1},\quad
\Z_{3;2m-1}^\ast =  \Y_{2N-2m-1}.
\end{gather*}

According to the Adler---Kostant---Symes scheme (see \cite{AdlMoer80}),  $\widetilde{\mathfrak{g}}$
falls into two subalgebars
\begin{gather*}
\widetilde{\mathfrak{g}}_+ = \oplus \sum_{\ell \geqslant 0} \mathfrak{g}_{\ell},\qquad
\widetilde{\mathfrak{g}}_- = \oplus \sum_{\ell \leqslant -1} \mathfrak{g}_{\ell}.
\end{gather*}
The  dual spaces $\widetilde{\mathfrak{g}}_+^\ast$ and $\widetilde{\mathfrak{g}}_-^\ast$ 
with respect to the bilinear form \eqref{BiLinFSinG} are
\begin{gather*}
\widetilde{\mathfrak{g}}_+^\ast = \oplus \sum_{\ell \leqslant 2N-2} \mathfrak{g}_{\ell},\qquad
\widetilde{\mathfrak{g}}_-^\ast =  \oplus \sum_{\ell > 2N-2} \mathfrak{g}_{\ell}.
\end{gather*}
Note, that $\mathfrak{g}_{N-1}^\ast = \mathfrak{g}_{N-1}$, 
and $\mathfrak{g}_{\ell}^\ast$ is dual to $\mathfrak{g}_{2N-\ell-2}$.

\subsection{Phase space of sine-Gordon hierarchy}
The sine-Gordon hierarchy of Hamiltonian systems arises
on coadjoint orbits of the group $\widetilde{G}_+ \,{=}\, \exp(\widetilde{\mathfrak{g}}_+)$.
The subspace 
$\M_N \,{=}\, \widetilde{\mathfrak{g}}^\ast_+ / \big( \sum_{\ell < -2}  \mathfrak{g}_\ell \big)$
is $\ad^\ast$-invariant under the coadjoint action of $\widetilde{\mathfrak{g}}_+$.
Actually,
\begin{gather*}
\M_N = \bigg\{\Lmatr = \sum_{\ell=-2}^{2N-2}  \sum_{a=1,2,3}  L_{a;\ell} \Z_{a;2N-\ell-2}^\ast \bigg\},
\end{gather*}
where $L_{a;\ell}$ are coordinates on $\M_N$:
$$L_{a;\ell} = \langle \Lmatr, \Z_{a;2N-\ell-2}\rangle,$$
and also serve as dynamic variables 
for the $N$-gap system within the hierarchy,
Let $L_{1;\ell}=\alpha_{\ell}$, 
$L_{2;\ell} = \beta_{\ell}$, $L_{3;\ell} = \gamma_{\ell}$; then
every element $\Lmatr \in \M_N$ has the form
\begin{subequations}\label{SinGPhSp}
\begin{gather}
\Lmatr(z) =
\begin{pmatrix} \alpha(z) & \beta(z) \\ \gamma(z) & -\alpha(z)
\end{pmatrix},\\
\begin{split}
\alpha(z) = \sum_{m=0}^{N} &\alpha_{2m} z^m,\quad
\beta(z) = \sum_{m=0}^{N} \beta_{2m-1} z^{m-1},\quad
\gamma(z) = \sum_{m=0}^{N} \gamma_{2m-1} z^{m},\\
&\alpha_{2N} = 0,\qquad \beta_{2N-1} = \gamma_{2N-1} = \sfb = \const.
\end{split}
\end{gather}
\end{subequations}
The $3N$ coordinates of $\M_N$ are ordered as follows: 
\begin{equation}\label{SinGDynVar}
\{\beta_{2m -1}, \gamma_{2m - 1}, \alpha_{2m}\}_{m=0}^{N-1}.
\end{equation}

The symplectic manifold $\M_N$ is equipped with the Lie-Poisson bracket
\begin{subequations}\label{LiePoiBraSinG}
\begin{gather}
\forall \mathcal{F}, \mathcal{H}\in \mathcal{C}^1(\M_N) \qquad
\{\mathcal{F},\mathcal{H}\} = \sum_{i,j =-2}^{2N-2}  \sum_{a, b=1,2,3} 
W_{i,j}^{a,b} \frac{\partial \mathcal{F}}{\partial L_{a;i}} \frac{\partial \mathcal{H}}{\partial L_{b;j}},
\\ W_{i,j}^{a,b}  = \langle \Lmatr, [\Z_{a,2N-i-2},\Z_{b,2N-j-2}] \rangle,
\end{gather}
\end{subequations}
or in terms of the dynamic variables:
\begin{gather}\label{SinGPoiBra}
\begin{split}
&\{\beta_{2m-1}, \alpha_{2n}\} = \beta_{2(n+m-N)+1},\\
&\{\gamma_{2m-1}, \alpha_{2n}\} = - \gamma_{2(n+m-N)+1},\quad N-1 \leqslant m+n,  \\
&\{\beta_{2m-1}, \gamma_{2n-1}\} = - 2 \alpha_{2(m+n-N)},\quad N \leqslant m+n.
\end{split}
\end{gather}
The Lie-Poisson structure $\textsf{W}= (W_{i,j}^{a,b})$ has the form
\begin{align*}
&\textsf{W} = \begin{pmatrix} 
 0 & 0 & \dots & 0 & \w_{0} \\
 \vdots & \vdots & \dots & \w_{0}  & \w_1\\
 0 & 0 & \iddots  & \w_1 & \w_2 \\
 0 &  \w_{0} &  \iddots  & \vdots & \vdots \\
 \w_{0}  & \w_1 & \dots & \w_{N-2} & \w_{N-1}
   \end{pmatrix},\\
&\w_{0} =  \begin{pmatrix} 0 & 0 & \beta_{-1}  \\
0 & 0 & - \gamma_{-1} \\
 - \beta_{-1} & \gamma_{-1} & 0 \\
 \end{pmatrix}, \\
&\w_n =  \begin{pmatrix} 0 & - 2\alpha_{2n-2} & \beta_{2n-1}  \\
 2\alpha_{2n-2}& 0 & -\gamma_{2n-1} \\
 - \beta_{2n-1}  & \gamma_{2n-1}  & 0 \\
 \end{pmatrix},\quad n = 1,\,\dots,\, N-1.
\end{align*}
This Poisson bracket  is degenerate and not canonical.

The action of $\widetilde{G}_+$ splits $\M_N$ into orbits
$$\mathcal{O} = \{\Lmatr = \Ad^\ast_{g} \Lmatr^{\text{in}} \mid g\in \widetilde{G}_+\},\qquad \Phi \in \M_N.$$
An initial point $\Lmatr^{\text{in}}$ belongs to the Weyl chamber of $\widetilde{G}_+$,
and is represented by a diagonal matrix, that is,  spanned by $H^{\ast}_{2m}$, $m = 0$, \ldots, $N$. 
Orbits have dimension $2N$
and serve as the phase spaces for Hamiltonian systems.

Physically meaningful Hamiltonian systems arise when $\mathfrak{g}$ is 
one of the real forms of $\mathfrak{sl}(2, \Complex)$, namely $\mathfrak{sl}(2, \Real)$, $\imath \mathfrak{sl}(2, \Real)$,
$\mathfrak{su}(2)$, or $\imath\mathfrak{su}(2)$. In the cases of $\mathfrak{su}(2)$ and $\imath\mathfrak{su}(2)$,
the sine-Gordon hierarchy is obtained,
while $\mathfrak{sl}(2, \Real)$ and $\imath\mathfrak{sl}(2, \Real)$ yield the $\sinh$-Gordon hierarchy.

\begin{rem}
In $\M_N$ with $\Lmatr$ of the form \eqref{SinGPhSp}, 
coadjoint orbits of $\widetilde{G}_{-} \,{=}\, \exp(\widetilde{\mathfrak{g}}_-)$
serve as phase spaces of the mKdV hierarchy, 
which is focusing  in the cases of $\mathfrak{su}(2)$ and $\imath\mathfrak{su}(2)$,
and defocusing if $\mathfrak{g}\,{=}\,\mathfrak{sl}(2, \Real)$ or $\imath \mathfrak{sl}(2, \Real)$. 
\end{rem}

\subsection{Integrals of motion}\label{ss:IntMot}
Let
\begin{equation}\label{InvF}
\begin{split}
H(z) &= \tfrac{1}{2} \tr \Lmatr^2(z) = \alpha(z)^2 + \beta(z) \gamma(z)\\
&= h_{2N-1} z^{2N-1} + \cdots + h_1 z + h_0 + h_{-1} z^{-1},
\end{split}
\end{equation}
where
\begin{gather}\label{hExprs}
\begin{split}
& h_{2N-1} = \sfb^2,\\
& h_{2N-2} = \alpha_{2N-2}^2 + \sfb (\beta_{2N-3} + \gamma_{2N-3}),\\
&\dots \\
& h_{0} = \alpha_0^2 + \beta_1 \gamma_{-1} + \beta_{-1} \gamma_{1},\\
& h_{-1} = \beta_{-1} \gamma_{-1}.
\end{split}
\end{gather}

Evidently, $H(z)$ is invariant under the action of $\widetilde{G}_{+}$.
Every $h_n$ is an integral of motion, and $h_{2N-1} \,{=}\, \sfb^2$ is an absolute constant. 
With respect to the symplectic structure \eqref{SinGPoiBra},
$h_{N-1}$, \ldots, $h_{2N-2}$ give rise to non-trivial Hamiltonian flows; 
we refer to these as \emph{Hamiltonians}. 

The remaining functions $h_{-1}$, $h_0$, \ldots, $h_{N-2}$ annihilate 
the Poisson bracket \eqref{LiePoiBraSinG} because 
$$\sum_{i =-1}^{2N-2}  \sum_{a =1,2,3} 
W_{i,j}^{a,b} \frac{\partial h_n}{\partial L_{a;i}} = 0,  \quad n=-1, 0,\dots, N-2.$$ 
Consequently, $\{h_n, \mathcal{F}\} = 0$
for any  $\mathcal{F} \in \mathcal{C}^1 (\M_N)$. 
The relations 
\begin{gather}\label{SinGConstr}
h_{-1} = \rmr_{-1},\quad h_{0} = \rmr_{0},\quad  \ldots,\quad  h_{N-2} = \rmr_{N-2},
\end{gather}
where $\rmr_{-1}$, $\rmr_{0}$, \ldots, $\rmr_{N-2}$ are constants,
thus serve as \emph{constraints} on the symplectic manifold $\M_N$. They 
fix an orbit~$\mathcal{O}$ of dimension $2N$, which  acts 
as the phase space for an $N$-gap Hamiltonian system in the sine(sinh)-Gordon hierarchy.

\subsection{Real forms of $\mathfrak{sl}(2,\Complex)$}\label{ss:RealForms}
The algebra $\mathfrak{sl}(2,\Complex)$ has two real forms: $\mathfrak{su}(2)$
and $\mathfrak{sl}(2,\Real)$, which can be extended to four cases:
$\mathfrak{sl}(2, \Real)$, $\imath \mathfrak{sl}(2, \Real)$,
$\mathfrak{su}(2)$, $\imath\mathfrak{su}(2)$. 
In all these cases,
integrals of motion $h_{-1}$, $h_0$, \ldots, $h_{2N-2}$, $h_{2N-1} \,{=}\, \sfb^2$ are real.
The dynamic variables acquire specific properties in each case:
\begin{itemize}
\item $\mathfrak{g} \,{=}\, \mathfrak{sl}(2,\Real)$, then all coordinates 
$\beta_{2m -1}$, $\gamma_{2m - 1}$, $\alpha_{2m}$, $m\,{=}\,0$, \ldots $N-1$,
are real, and also $\sfb \in \Real$.
Thus, $h_{2N-1} \,{=}\, \sfb^2 \,{>}\, 0$;
\item $\mathfrak{g} \,{=}\, \imath \mathfrak{sl}(2,\Real)$, then for  $m=0$, \ldots $N-1$
$\alpha_{2m} \,{\in}\, \Real$, $\beta_{2m -1}$, $\gamma_{2m - 1} \,{\in}\, \imath \Real$,
and $\sfb \,{\in}\, \imath \Real$, the latter implies $h_{2N-1} \,{<}\, 0$;
\item $\mathfrak{g} \,{=}\, \mathfrak{su}(2)$, then $\alpha_{2m} \,{\in}\, \imath \Real$, 
 $\beta_{2m-1} \,{=}\, {-} \bar{\gamma}_{2m-1}$,  and $\sfb \,{\in}\, \imath \Real$, 
 $h_{2N-1} < 0$;
 \item $\mathfrak{g} \,{=}\, \imath \mathfrak{su}(2)$, then $\alpha_{2m} \,{\in}\, \imath \Real$, 
 $\beta_{2m-1} \,{=}\, \bar{\gamma}_{2m-1}$,  and $\sfb \,{\in}\, \Real$, 
 $h_{2N-1} > 0$.
\end{itemize}

\subsection{Sine(sinh)-Gordon equation}\label{ss:SineSinhParam}
Let $h_{2N-2}$ give rise to a stationary flow with a parameter $\sfx$,
and $h_{N-1}$ give rise to an evolutionary flow with a parameter $\sft$:
\begin{gather}\label{TwoFlowsEq}
\frac{\rmd L_{a;\ell}}{\rmd \sfx} = \{h_{2N-2}, L_{a;\ell}\},\qquad
\frac{\rmd L_{a,\ell}}{\rmd \sft} = \{h_{N-1}, L_{a,\ell}\}.
\end{gather}  
In more detail, the stationary flow is
\begin{subequations}\label{xFlow}
\begin{align}
&\frac{\rmd \beta_{-1}}{\rmd \sfx} = 2 \alpha_{2N-2} \beta_{-1},\qquad
\frac{\rmd \gamma_{-1}}{\rmd \sfx} = - 2 \alpha_{2N-2} \gamma_{-1}, \label{BetaGammaM1EqSt} \\
&\frac{\rmd \alpha_{2m}}{\rmd \sfx} =  \sfb (\gamma_{2m-1} - \beta_{2m-1}),\quad m=0,\dots, N-1,\\
&\frac{\rmd \beta_{2m-1}}{\rmd \sfx} = 2 \alpha_{2N-2} \beta_{2m-1} - 2 \sfb \alpha_{2m-2},\\
&\frac{\rmd \gamma_{2m-1}}{\rmd \sfx} = - 2 \alpha_{2N-2} \gamma_{2m-1} + 2 \sfb \alpha_{2m-2}, 
\quad m=1,\dots, N-1.
\end{align}
\end{subequations}
The evolutionary flow is
\begin{subequations}\label{tFlow}
\begin{align}
&\frac{\rmd \alpha_{2m-2}}{\rmd \sft} =  \gamma_{-1} \beta_{2m-1} - \beta_{-1} \gamma_{2m-1},\quad m=1,\dots, N-1,
 \label{BetaGammaM1EqEv}\\
&\frac{\rmd \beta_{2m-1}}{\rmd \sft} = 2 \alpha_{2m} \beta_{-1}, \qquad
\frac{\rmd \gamma_{2m-1}}{\rmd \sft} = - 2 \alpha_{2m} \gamma_{-1},\quad 
m=0,\dots, N-1, \label{BetaGamma2NM1EqEv} \\
&\frac{\rmd \alpha_{2N-2}}{\rmd \sft} = \sfb(\gamma_{-1}  - \beta_{-1}).  \label{Alpha2NM2EqEv}
\end{align}
\end{subequations}

The sine(sinh)-Gordon equations are obtained from the equality
\begin{gather}\label{SinGEqZeroCurv}
\frac{\rmd \alpha_{2N-2}}{\rmd \sft} = \frac{\rmd \alpha_{0}}{\rmd \sfx}.
\end{gather}
Indeed, from \eqref{BetaGammaM1EqSt} one can see that 
$$\alpha_{2N-2} = \tfrac{1}{2} \rmd \ln \beta_{-1} / \rmd \sfx = - \tfrac{1}{2} \rmd \ln \gamma_{-1} / \rmd \sfx.$$
On the other hand, \eqref{BetaGamma2NM1EqEv}  implies
$$\alpha_{0} = \tfrac{1}{2} \rmd \ln \beta_{-1} / \rmd \sft = - \tfrac{1}{2} \rmd \ln \gamma_{-1} / \rmd \sft. $$

In the case of $\mathfrak{g} \,{=}\, \mathfrak{su}(2)$  (with $\sfb \,{\in}\, \imath  \Real$), 
we have $\beta_{-1} \,{=}\, {-} \bar{\gamma}_{-1}$.
Thus, the equation 
\begin{equation}\label{hm1Rel}
\beta_{-1} \gamma_{-1} = \rmr_{-1}
\end{equation}
implies 
${-}|\gamma_{-1}|^2 \,{=}\, \rmr_{-1} \,{=}\, \const$. This implies 
that $\gamma_{-1}$ and $\beta_{-1}$ both trace a circle of radius
$\sfr \,{=}\,\sqrt{{-}\rmr_{-1}}$ if $\rmr_{-1}\,{<}\,0$.
If $\rmr_{-1}\,{=}\,0$, then $\gamma_{-1} \,{=}\, \beta_{-1}\,{=}\, 0$; 
however, if $\rmr_{-1}\,{>}\,0$, no states of the Hamiltonian system exist.
We assign
\begin{subequations}\label{SUGammaBetaParam} 
\begin{equation}
\gamma_{-1} \,{=}\, \imath \sfr \exp (\imath \phi),\qquad 
\beta_{-1} \,{=}\, \imath \sfr \exp (- \imath \phi),\quad  
\sfr \,{=}\, \sqrt{-\rmr_{-1}},\ \ \rmr_{-1} < 0,
\end{equation}
where $\phi$ is a function of $\sfx$ and $\sft$.

In the case of $\mathfrak{g} \,{=}\, \imath \mathfrak{su}(2)$  (with $\sfb \,{\in}\, \Real$),
we have 
$\beta_{-1} \,{=}\, \bar{\gamma}_{-1}$, and  so 
$ |\gamma_{-1}|^2 \,{=}\, \rmr_{-1}$. Thus,  
$\gamma_{-1}$ traces a circle of radius $\sfr \,{=}\, \sqrt{\rmr_{-1}}$ if $\rmr_{-1}\,{>}\,0$.
Similar to the previous case, $\gamma_{-1} \,{=}\, \beta_{-1}\,{=}\, 0$
if $\rmr_{-1}\,{=}\,0$, and no states of the Hamiltonian system exist
if $\rmr_{-1}\,{<}\,0$.
In this case, we assign
\begin{equation}
\gamma_{-1} \,{=}\, \sfr \exp (\imath \phi),\qquad 
\beta_{-1} \,{=}\, \sfr \exp ({-} \imath \phi),\quad  
\sfr \,{=}\, \sqrt{\rmr_{-1}},\ \ \rmr_{-1}>0.
\end{equation}
\end{subequations}

In the both cases of \eqref{SUGammaBetaParam},  equations
\eqref{BetaGammaM1EqSt} and \eqref{BetaGamma2NM1EqEv}  imply,
respectively,
\begin{gather*}
\alpha_{2N-2} = - \frac{\imath}{2} \frac{\rmd \phi}{\rmd \sfx},\qquad 
\alpha_{0} = - \frac{\imath}{2} \frac{\rmd \phi}{\rmd \sft}.
\end{gather*}
Finally,  \eqref{Alpha2NM2EqEv} produces the equation for $\phi$
\begin{subequations}\label{SinGEq} 
\begin{align}\label{SUSinGEq} 
\imath \frac{\rmd^2 \phi}{\rmd \sft \rmd \sfx} &=  4  \sfb \sfr \sin \phi,\quad \sfb \,{\in}\, \imath  \Real,\\
\frac{\rmd^2 \phi}{\rmd \sft \rmd \sfx} &=  - 4  \sfb \sfr \sin \phi,\quad \sfb \,{\in}\,  \Real.  \label{iSUSinGEq} 
\end{align}
\end{subequations}
The both equations \eqref{SinGEq}
 turn into the \emph{sine-Gordon equation} \eqref{SGEq} by the corresponding transformation 
\begin{subequations}\label{tTrans} 
\begin{align}
&\sfx \mapsto \x,\qquad\qquad \sft \mapsto \rmt = 
- \imath \sfb \sfr  \sft,\quad \sfb \,{\in}\, \imath  \Real, \label{SUtTrans} \\
&\sfx \mapsto \x,\qquad\qquad \sft \mapsto \rmt = 
- \sfb \sfr  \sft,\quad \sfb \,{\in}\,  \Real. \label{iSUtTrans}
\end{align}
\end{subequations}

In the case of $\mathfrak{g} \,{=}\, \mathfrak{sl}(2,\Real)$,  
we assign
\begin{subequations}\label{SLGammaBetaParam} 
\begin{equation}
\gamma_{-1} \,{=}\, \sfr \exp ( \phi),\qquad  \beta_{-1} \,{=}\,  \sfr \exp (- \phi),\quad
\sfr \,{=}\, \sqrt{\rmr_{-1}},\ \ \rmr_{-1}>0,
\end{equation}
and in the case of $\mathfrak{g} \,{=}\, \imath \mathfrak{sl}(2,\Real)$,  
\begin{equation}
\gamma_{-1} \,{=}\, \imath \sfr \exp ( \phi),\qquad  \beta_{-1} \,{=}\,  \imath \sfr \exp (- \phi),\quad
\sfr \,{=}\, \sqrt{-\rmr_{-1}},\ \ \rmr_{-1}<0,
\end{equation}
\end{subequations}
where $\phi$ is a function of $\sfx$ and $\sft$.
Then, \eqref{BetaGammaM1EqSt} and \eqref{BetaGamma2NM1EqEv}  imply
\begin{gather*}
\alpha_{2N-2} = - \frac{1}{2} \frac{\rmd \phi}{\rmd \sfx},\qquad 
\alpha_{0} = - \frac{1}{2} \frac{\rmd \phi}{\rmd \sft},
\end{gather*}
and \eqref{Alpha2NM2EqEv} produces the equation for $\phi$
\begin{subequations}\label{SinhGEq} 
\begin{align}\label{SLSinhGEq} 
 \frac{\rmd^2 \phi}{\rmd \sft \rmd \sfx} &= - 4 \sfb \sfr \sinh \phi,\quad \sfb \in \Real,\\
 \frac{\rmd^2 \phi}{\rmd \sft \rmd \sfx} &= - 4 \imath \sfb \sfr \sinh \phi,\quad \sfb \in \imath\Real.
\end{align}
\end{subequations}
The both equations \eqref{SinhGEq}
 turn into the \emph{sinh-Gordon equation} \eqref{ShGEq} 
 by the corresponding transformation from \eqref{tTrans}.

\begin{rem}
Note, that the sine(sinh)-Gordon equations arise when $N \,{\geqslant}\, 2$.
If $N\,{=}\,1$, there exists the only Hamiltonian $h_{0} \,{=}\, \alpha_0^2 \,{+}\, \sfb (\beta_{-1} \,{+}\, \gamma_{-1})$,
which produces the stationary flow
\begin{align*}
&\frac{\rmd \beta_{-1}}{\rmd \sfx} = 2 \alpha_{0} \beta_{-1},&
&\frac{\rmd \gamma_{-1}}{\rmd \sfx} = - 2 \alpha_{0} \gamma_{-1}, &
&\frac{\rmd \alpha_{0}}{\rmd \sfx} =  \sfb (\gamma_{-1} - \beta_{-1}).
\end{align*}
From the flow of $h_0$ we find: 
$\alpha_{0} \,{=}\, \frac{1}{2} \rmd \ln \beta_{-1} / \rmd \sfx \,{=}\, {-} \frac{1}{2} \rmd \ln \gamma_{-1} / \rmd \sfx$.
With the same substitutions for $\gamma_{-1}$ and $\beta_{-1}$, where
 $\phi$ is a function  of $\sfx$, we come to the stationary equation, which is
 the \emph{equation of motion of a simple  pendulum} if $\mathfrak{g} \,{=}\, \mathfrak{su}(2)$:
 \begin{gather*}
 \imath \frac{\rmd^2 \phi}{\rmd \sfx^2} = 4  \sfb \sfr \sin \phi,\quad \sfb \in \imath \Real,
\end{gather*}
 or $\mathfrak{g} \,{=}\,\imath \mathfrak{su}(2)$:
\begin{gather*}
\frac{\rmd^2 \phi}{\rmd \sfx^2} = - 4 \sfb \sfr \sin \phi,\quad \sfb \in  \Real.
\end{gather*}
\end{rem}

\subsection{Zero curvature representation}
The system of dynamic equations \eqref{TwoFlowsEq} admits  the matrix form
\begin{gather*}
\frac{\rmd \Lmatr}{\partial \sfx} = [\Lmatr, \nabla h_{2N-2}], \qquad
\frac{\rmd \Lmatr}{\partial \sft} = [\Lmatr, \nabla h_{N-1}],
\end{gather*}
where $\nabla h_n$ denotes the matrix gradient of $h_n$, namely,
\begin{gather*}
\nabla h_n = \sum_{i =-1}^{2N-1}  \sum_{a =1,2,3} 
\frac{\partial h_n}{\partial L_{a;i}} \Z_{a,2N-i-2}.
\end{gather*}
The matrix gradient of each flow has a complementary matrix $ \mathsf{A}$ such that
$$[\Lmatr, \nabla h_{n}] =  [\Lmatr,  \mathsf{A}].$$  Unlike $\nabla h_n$, 
the complementary matrix $ \mathsf{A}$ is defined uniquely  for all $N$;
that is, for all Hamiltonian systems in the hierarchy.
Actually,
\begin{gather}\label{PsiMatrEqs} 
\frac{\rmd \Lmatr}{\partial \sfx} = [\Lmatr,   \mathsf{A}_\text{st}], \qquad
\frac{\rmd \Lmatr}{\partial \sft} = [\Lmatr,   \mathsf{A}_\text{ev}], \\
\begin{split}
& \mathsf{A}_\text{st} = - \begin{pmatrix} 
\alpha_{2N-2} & \sfb \\  
\sfb z  & -\alpha_{2N-2} 
 \end{pmatrix},\qquad
 \mathsf{A}_\text{ev} =  \begin{pmatrix} 
0 & z^{-1} \beta_{-1} \\
 \gamma_{-1} & 0
 \end{pmatrix}.
 \end{split} \notag
\end{gather}
The zero curvature representation for the sine(sinh)-Gordon hierarchy has the form
\begin{gather*}
\frac{\rmd  \mathsf{A}_\text{st}}{\partial \sft} - \frac{\rmd  \mathsf{A}_\text{ev}}{\partial \sfx} 
= [ \mathsf{A}_\text{st},  \mathsf{A}_\text{ev}].
\end{gather*}

\section{Separation of variables}\label{s:SoV}

\subsection{Spectral curve}
The sine(sinh)-Gordon hierarchy is associated with the family of hyperelliptic curves
\begin{equation}\label{SpectrCurve}
- w^2 + z^2 H(z) = 0,
\end{equation}
where  $H$ is defined by \eqref{InvF}. This equation of a genus $N$ curve
is obtained from the characteristic polynomial of $\Lmatr$, namely $\det \big(\Lmatr(z) \,{-}\, (w/z)\big) \,{=}\, 0$.
All parameters of the curve are integrals of motion:  Hamiltonians $h_{N-1}$, \dots, $h_{2N-2}$, 
constraints $h_{-1}$, \ldots, $h_{N-2}$, and $h_{2N-1} \,{=}\, \sfb^2 \,{=}\, \const$.

\subsection{Canonical coordinates}
As shown in \cite{KK1976}, variables of separation in the sine-Gordon hierarchy 
are given by a positive non-special\footnote{A non-special divisor on a hyperelliptic curve of genus $g$
is a degree $g$ positive divisor which contains no pairs of points in involution.}  
divisor of degree equal to the genus $N$ of the spectral curve.
In fact, pairs of coordinates of $N$ points from the support of such a divisor serve as quasi-canonical variables, 
and thus lead to separation of variables.
Below, we briefly explain how the required divisor can be obtained from the dynamic variables
for an $N$-gap system in the sine-Gordon hierarchy, and prove that  
pairs of coordinates of the $N$ points serve as quasi-canonical variables.

Recall that the symplectic manifold $\M_N$, with $3N$ coordinates \eqref{SinGDynVar} defined on it, splits
into orbits $\mathcal{O}$ fixed by $N$ constraints, and so $\dim \mathcal{O}\,{=}\, 2N$.
Let  $\beta_{2m-1}$, $m\,{=}\,0$, \ldots, $N-1$, be eliminated with the help of constraints \eqref{SinGConstr}.
Since the expressions \eqref{hExprs} are linear with respect to $\beta_{2m-1}$,
it is convenient to present them in  matrix form. The constraints are written as
\begin{gather}\label{Constr}
\Gamma_{\text{c}} \bm{\beta} + \A_{\text{c}} = \bm{r},\\
\intertext{where}
\Gamma_{\text{c}} = \begin{pmatrix}
\rmb & 0 & 0 &\dots & 0  \\
0& \gamma_{-1} & \gamma_{1} & \dots & \gamma_{2N-3} \\
\vdots & 0 & \gamma_{-1} & \ddots  & \vdots  \\
0 & \vdots & \ddots & \ddots & \gamma_{1} \\
0 & 0& \dots  & 0 & \gamma_{-1}
\end{pmatrix},\qquad
\bm{\beta} = \begin{pmatrix} \sfb \\
\beta_{2N-3}  \\ \vdots \\  \beta_1 \\ \beta_{-1}
\end{pmatrix},  \notag \\
\A_{\text{c}} = \begin{pmatrix} 0 \\ 
 \sum_{k=0}^{N-2} \alpha_{2(N-2-k)} \alpha_{2k} \\ 
\vdots \\
 \sum_{k=0}^{n} \alpha_{2(n-k)} \alpha_{2k} \\ 
\vdots\\
\alpha_0^2 \\
0 \end{pmatrix},\qquad
\bm{r} = \begin{pmatrix}  \sfb^2 \\ \rmr_{N-2}  \\ \vdots  \\ \rmr_{n}  \\ 
\vdots \\ \rmr_{0} \\ \rmr_{-1} \end{pmatrix}.  \notag
\end{gather}
The first equation is an identity. From \eqref{Constr} we find
\begin{equation}\label{GammaFromConstr}
\bm{\beta} = \Gamma_{\text{c}}^{-1}(\bm{r} - \A_{\text{c}}).
\end{equation}

Let values of the Hamiltonians be denoted by $\rmh_{N-1}$, $\rmh_N$, \ldots, $\rmh_{2N-2}$.
Expressions for the Hamiltonians admit the matrix form
\begin{gather}\label{Hams}
\Gamma_{\text{h}} \bm{\beta} + \A_{\text{h}} = \bm{h},\\
\intertext{where}
\Gamma_{\text{h}} = \begin{pmatrix}
\gamma_{2N-3} & \sfb & 0 &  \dots & 0 \\
\vdots  & \gamma_{2N-3} & \ddots & \ddots & \vdots   \\
\gamma_1 & \vdots & \ddots &  \sfb &  0  \\
\gamma_{-1} & \gamma_1 & \dots &  \gamma_{2N-3} & \sfb
\end{pmatrix},  \notag \\
\A_{\text{h}} = \begin{pmatrix} 
\alpha_{2N-2}^2 \\
 \vdots \\
 \sum_{k=n}^{N-1} \alpha_{2(N-1-k+n)} \alpha_{2k} \\
 \vdots \\
 \sum_{k=0}^{N-1} \alpha_{2(N-1-k)} \alpha_{2k} 
 \end{pmatrix},\qquad
\bm{h} = \begin{pmatrix} \rmh_{2N-2}  \\ \vdots \\ \rmh_{N-1+n} \\ \vdots \\ \rmh_{N-1} \end{pmatrix}.  \notag
\end{gather}
Substituting \eqref{GammaFromConstr} into \eqref{Hams}, we obtain
\begin{equation}\label{HamsBeta}
 \bm{h} = \Gamma_{\text{h}}  \Gamma_{\text{c}}^{-1}(\bm{r} - \A_{\text{c}})+ \A_{\text{h}}.
\end{equation}

On the other hand, the values of Hamiltonians $\rmh_{2N-2}$, \ldots, $\rmh_{N-1}$ can be found from 
the equation \eqref{SpectrCurve} of the spectral curve taken at the points $\{(z_i,w_i)\}_{i=1}^N$ 
which form a non-special divisor. Namely, with $i=1$, \ldots, $N$,
\begin{equation*}
 - w_i^2 + \rmb^2 z_i^{2N+1} + \rmh_{2N-2} z_i^{2N} + \cdots + \rmh_{N-1} z_i^{N+1} 
 + \rmr_{N-2} z_i^{N} + \dots + \rmr_{-1} z_i  = 0,
\end{equation*}
or in  matrix form
\begin{equation*}
- \bm{w} +  \mathrm{Z}_{\text{h}} \bm{h} + \mathrm{Z}_{\text{c}} \bm{r}  = 0,
\end{equation*}
where
\begin{gather*}
\mathrm{Z}_{\text{c}} = \begin{pmatrix} 
z_1^{2N+1} & z_1^{N} & \dots & z_1 \\
z_2^{2N+1} & z_2^{N} & \dots & z_2 \\
\vdots &  \vdots & \ddots & \vdots \\
z_N^{2N+1} & z_N^{N} & \dots & z_N 
\end{pmatrix},\quad
\mathrm{Z}_{\text{h}} = \begin{pmatrix} 
z_1^{2N} & \dots & z_1^{N+2} & z_1^{N+1} \\ 
z_2^{2N} & \dots & z_2^{N+2} & z_2^{N+1} \\ 
\vdots &  \ddots & \vdots & \vdots \\
z_N^{2N} & \dots & z_N^{N+2} & z_N^{N+1} 
\end{pmatrix},\quad 
\bm{w} =  \begin{pmatrix}
w_1^2 \\ w_2^2 \\ \vdots \\ w_N^2
\end{pmatrix}.
\end{gather*}
The matrix $\mathrm{Z}_{\text{h}}$ is square and invertible. Thus,
\begin{equation}\label{HamsSpectrC}
 \bm{h} = \mathrm{Z}_{\text{h}}^{-1} \big(\bm{w} -  \mathrm{Z}_{\text{c}} \bm{r}\big).
\end{equation}

Equations \eqref{HamsBeta} and \eqref{HamsSpectrC}
define the same Hamiltonians. Therefore, 
\begin{gather*}
 \Gamma_{\text{h}}  \Gamma_{\text{c}}^{-1}(\bm{r} - \A_{\text{c}})+ \A_{\text{h}} = 
 \mathrm{Z}_{\text{h}}^{-1} \big(\bm{w} -  \mathrm{Z}_{\text{c}} \bm{r}\big).
\end{gather*}
Moreover, constants $\bm{r}$ can be taken arbitrarily,
and thus the corresponding coefficients coincide as well as the remaining terms:
\begin{subequations}
\begin{gather}
 \Gamma_{\text{h}}  \Gamma_{\text{c}}^{-1} = - \mathrm{Z}_{\text{h}}^{-1}  \mathrm{Z}_{\text{c}}, \label{cCoefs} \\
  -\Gamma_{\text{h}}  \Gamma_{\text{c}}^{-1}\A_{\text{c}} + \A_{\text{h}} = \mathrm{Z}_{\text{h}}^{-1} \bm{w}. \label{fCoefs}
\end{gather}
\end{subequations}
From \eqref{cCoefs} we find
$$\mathrm{Z}_{\text{h}} \Gamma_{\text{h}} + \mathrm{Z}_{\text{c}} \Gamma_{\text{c}} = 0,$$
which is equivalent to $\gamma(z_i) \,{=}\, 0$.
Then, from \eqref{fCoefs} we obtain
$$ \mathrm{Z}_{\text{c}} \A_{\text{c}} + \mathrm{Z}_{\text{h}}\A_{\text{h}} = \bm{w}, $$
which produces
$w_i^2 - z_i^2 \alpha(z_i)^2 = 0$.
Thus, the points $(z_i,w_i)$ of the required divisor are defined by
\begin{equation*}\label{PointsDef2}
\gamma(z_i) = 0,\qquad w_i^2 - z_i^2 \alpha(z_i)^2 = 0,  \quad i=1,\dots, N.
\end{equation*}
This result  was firstly discovered in  \cite{KK1976}.

\begin{theo}\label{T:QCanonVars}
Suppose, an orbit $\mathcal{O} \,{\subset}\, \M_N$ 
is described in terms of the coordinates $(\gamma_{2m-1}, \alpha_{2m})$, $m\,{=}\,0$, \ldots, $N\,{-}\,1$, 
introduced in \eqref{SinGPhSp}. 
Then the new coordinates $(z_i, w_i)$, $i\,{=}\,1$, \ldots, $N$, defined by the formulas
\begin{equation}\label{PointsDef1}
\gamma(z_i) = 0,\qquad w_i = \epsilon z_i \alpha(z_i),  \quad i=1,\dots, N,
\end{equation}
where $\epsilon^2 \,{=}\, 1$,
have the following properties:
\begin{enumerate}
\renewcommand{\labelenumi}{\arabic{enumi})}
\item  a pair $(z_i, w_i)$ is a point of the spectral curve  \eqref{SpectrCurve};
\item a pair $(z_i, w_i)$ is quasi-canonically conjugate with respect to the  Lie-Poisson
bracket \eqref{SinGPoiBra}:
\begin{equation}\label{CanonCoord}
\{z_i, z_j\} = 0,\qquad  \{z_i, w_j\} = \epsilon\, z_i^N\, \delta_{i,j},\qquad
\{w_i, w_j\} = 0;
\end{equation}
\item the canonical $1$-form is
\begin{equation}\label{Liouv1Form}
 \epsilon \sum_{i=1}^N z_i^{-N} w_i \rmd z_i.
\end{equation}
\end{enumerate}
\end{theo}
\begin{proof}
Since $z_i$, $i=1$, \ldots, $N$, depend only on $\gamma_{2m-1}$, $m=0$, \ldots, $N-1$, and the latter commute,
we have $\{z_i, z_j\} = 0$. Next,
\begin{multline*}
 \{z_i, w_j\}  = \sum_{n+m=N-1}^{2N-2} \bigg(\frac{\partial z_i}{\gamma_{2m-1}}  \frac{\partial w_j}{\alpha_{2n}}
 - \frac{\partial z_i}{\alpha_{2n}}  \frac{\partial w_j}{\gamma_{2m-1}} \bigg) \{\gamma_{2m-1}, \alpha_{2n}\} \\
 = \frac{\epsilon}{\gamma'(z_i)} \sum_{m+n =N-1}^{2N-2}  z_i^m z_j^{n+1} \gamma_{2(m+n-N)+1}
 =  \frac{\epsilon z_j}{\gamma'(z_i)}  \frac{z_i^{N-1} \gamma(z_i) - z_j^{N-1} \gamma(z_j)}{z_i - z_j},
\end{multline*}
since from \eqref{PointsDef1} we have 
\begin{gather*}
\frac{\partial z_i}{\gamma_{2m-1}}  = - \frac{z_i^m}{\gamma'(z_i)},\qquad
 \frac{\partial z_i}{\alpha_{2n}} = 0,\qquad  
 \frac{\partial w_i}{\alpha_{2n}} = \epsilon z_i^{n+1},\\
 \frac{\partial w_i}{\gamma_{2m-1}} = - \frac{\epsilon z_i^m}{\gamma'(z_i)} \big(\alpha(z_i) + z_i \alpha'(z_i) \big).
\end{gather*}
As $i \neq j$, it is evident that $\{z_i, w_j\}  = 0$, due to $\gamma(z_i)=\gamma(z_j)=0$.
As $i=j$, we get
$$  \{z_i, w_i\}  = \lim_{z_j \to z_i}   \frac{\epsilon z_j}{\gamma'(z_i)}  
\frac{z_i^{N-1} \gamma(z_i) - z_j^{N-1} \gamma(z_j)}{z_i - z_j}
= \epsilon\, z_i^N. $$
Finally, we find
\begin{multline*}
 \{w_i, w_j\}  = \sum_{n+m=N-1}^{2N-2} \bigg(\frac{\partial w_i}{\gamma_{2m-1}}  \frac{\partial w_j}{\alpha_{2n}}
 -  \frac{\partial w_i}{\alpha_{2n}}  \frac{\partial w_j}{\gamma_{2m-1}} \bigg) \{\gamma_{2m-1}, \alpha_{2n}\} \\
 = \epsilon^2 \sum_{m+n = N-1}^{2N-2}  \bigg(z_i^m z_j^{n+1} \frac{ \alpha(z_i) + z_i \alpha'(z_i) }{\gamma'(z_i)}  
 - z_i^{n+1} z_j^m \frac{ \alpha(z_j) + z_j \alpha'(z_j) }{\gamma'(z_j)} \bigg) \gamma_{2(m+n-N)+1} \\
 =   \epsilon^2 \frac{z_i^{N-1} \gamma(z_i) - z_j^{N-1} \gamma(z_j)}{z_i - z_j} 
 \bigg(z_j \frac{ \alpha(z_i) + z_i \alpha'(z_i) }{\gamma'(z_i)} 
 - z_i \frac{ \alpha(z_j) + z_j \alpha'(z_j) }{\gamma'(z_j)}  \bigg).
\end{multline*}
Thus,  $\{w_i, w_j\}  = 0$,  due to $\gamma(z_i)=\gamma(z_j)=0$.

Then, \eqref{Liouv1Form} follows from the fact that pairs $(z_i, w_i)$, $i=1$, \ldots, $N$, 
 obey \eqref{CanonCoord}.
\end{proof}

In what follows we assign $\epsilon = 1$.

\section{Algebro-geometric integration}\label{s:AGI}
\subsection{Uniformization of the spectral curve}
Separation of variables provides a solution to the Jacobi inversion problem 
on the spectral curve \eqref{SpectrCurve},
which is a hyperelliptic curve of genus $N$:
\begin{multline}\label{CurveGN}
 0 = F(z,w) \equiv - w^2 + \sfb^2 z^{2N+1} + \rmh_{2N-2} z^{2N} + \cdots + \rmh_{N-1} z^{N+1} \\
 + \rmr_{N-2} z^{N} + \dots + \rmr_0 z^2 + \rmr_{-1} z.
\end{multline}

The spectral curve \eqref{CurveGN} is reduced to the canonical  form  \eqref{V22g1Eq},
denoted by $\mathcal{V}$, by  applying
the transformation ($N\equiv g$)
\begin{equation}\label{SectrCToCanonC}
\begin{split}
z &\mapsto x,\qquad 
w \mapsto  y = w/\sfb,\qquad F(x, \sfb y) = \sfb^2 f(x,y),\\
&\rmh_{n}\; (\text{or } \rmr_n)= \sfb^2 \lambda_{4g-2-2n},\ \ n=-1,\,\dots,\, 2N-2.
\end{split}
\end{equation}
Note, that one of the branch points of \eqref{CurveGN} is fixed at $0$,
and thus the corresponding canonical  form has $\lambda_{4g+2}=0$.

We will work with the $\wp$-functions associated with $\mathcal{V}$ and use 
the coordinates $u \,{=}\, (u_1,\dots,u_{2g-1})^t$ on $\Jac(\mathcal{V})$. 
The differentials of the first and second kinds,
defined by \eqref{K1K2DifsGen}, are expressed in terms of $z$-, $w$-coordinates of 
the spectral curve \eqref{CurveGN} as follows:
\begin{subequations}
\begin{align}
& \rmd u_{2n-1} =  \frac{\sfb z^{N-n} \rmd z}{\partial_w F(z,w)},\quad n=1,\dots, N,\label{K1Difs} \\
& \rmd r_{2n-1} =  \frac{ \rmd z}{\sfb \partial_w F(z,w)} \sum_{j=1}^{2n-1} (2n-j) \rmh_{2N-j} z^{N+n-j}, 
\quad n=1,\dots, N, \label{K2Difs}
\end{align}
\end{subequations}
where the notation $\rmh_n$ is used for $\rmh_n$ and $\rmr_n$, and $\rmh_{2N-1} \,{=}\, \sfb^2$.

\begin{theo}
Let  $\gamma_{2m-1}$, $\beta_{2m-1}$, $\alpha_{2m}$ for $m\,{=}\,0$, \ldots, $N\,{-}\,1$
be dynamic variables for $\M_N$, and $D \,{=}\, \sum_{i=1}^N (z_i,w_i)$ be a divisor defined by \eqref{PointsDef1}.
Let $u\,{=}\, \mathcal{A}(D)$, where $u \,{\in}\, \Jac(\mathcal{V}) \backslash \Sigma$.
Then
\begin{subequations}\label{WPAlphaGamma}
\begin{align}
&\gamma_{2(N-k)-1} = - \sfb \wp_{1,2k-1}(u), \quad k=1,\dots, N; \label{WPGamma}\\
&\alpha_{2(N-k)} = - \frac{\sfb \wp_{1,2k-1,2N-1}(u)}{2 \wp_{1,2N-1}(u)},
\quad k=1,\dots, N, \label{WPAlpha} \\
&\beta_{2(N-k)-1} = \sfb \frac{ \wp_{2k-1,2N-1}(u)}{\wp_{1,2N-1}(u)},
\ \   k=1,\dots, N-1,\quad
\beta_{-1} = \frac{- \rmr_{-1}}{\sfb \wp_{1,2N-1}(u)}. \label{WPBeta}
\end{align}
\end{subequations}
\end{theo}
\begin{proof}
Applying the transformation \eqref{SectrCToCanonC} to  \eqref{EnC22g1},
we find a solution to the Jacobi inversion problem on the spectral curve \eqref{CurveGN}.
That is, coordinates of the support of $D$ are expressed in terms of the $\wp$-functions
associated with the spectral curve. Taking into account \eqref{PointsDef1}, we obtain
\begin{align*}
&\gamma_{2(N-k)-1} = - \sfb \wp_{1,2k-1}(u), \quad k=1,\dots, N;\\
&\alpha_{2N-2} = - \frac{\sfb \wp_{1,1,2N-1}(u)}{2 \wp_{1,2N-1}(u)}, \\
&\alpha_{2(N-k)} = - \frac{\sfb}{2} \Big(\wp_{1,1,2k-3}(u) -  
\frac{\wp_{1,2k-3}(u) \wp_{1,1,2N-1}(u)}{\wp_{1,2N-1}(u)} \Big), 
\ \  k=2,\dots, N. 
\end{align*}
By substituting these expressions into  \eqref{GammaFromConstr}, we find also 
expressions for $\beta_{2m-1}$, $m\,{=}\,1$, \ldots, $N$, in terms of the $\wp$-functions.
Next, identities for $\wp$-functions (see \cite{BerWPFF2025} for more details)
are applied to simplify the obtained expressions for the dynamic variables, and thus \eqref{WPAlpha}
and \eqref{WPBeta} are obtained
\end{proof}

Finally, we have 
\begin{equation}\label{SinGSolGamma}
\gamma_{-1}  = - \sfb \wp_{1,2N-1}(u),
\end{equation}
 which yields the $N$-gap solution $\phi$ 
to equations \eqref{SinGEq} and \eqref{SinhGEq},
according to the chosen parametrization for for $\gamma_{-1}$.
This solution coincides with the one proposed in \cite[Theorem 4.13]{belHKF}.

\begin{rem}
The sine(sinh)-Gordon equations, derived from \eqref{SinGEq} and \eqref{SinhGEq},
are equivalent to the following identity for the $\wp$-functions
associated with $\mathcal{V}$ of the form \eqref{V22g1Eq}, provided that $\lambda_{4N+2}=0$:
\begin{gather}\label{SinGWPEq}
\frac{\rmd}{\rmd u_{2N-1}} \frac{\wp_{1,1,2N-1}(u)}{\wp_{1,2N-1}(u)} = 
2 \Big(\wp_{1,2N-1}(u) -  \frac{\lambda_{4 N}}{\wp_{1,2N-1}(u) }\Big).
\end{gather}
The identity coincides with the one obtained from \eqref{Alpha2NM2EqEv}
by applying \eqref{WPAlphaGamma}. 
\end{rem}

\begin{rem}
Parameters $\lambda_2$, \ldots, $\lambda_{2N}$ of \eqref{V22g1Eq}
are related to the values $\rmh_{2N-2}$, \ldots, $\rmh_{N-1}$  of the Hamiltonians, 
namely $\rmh_{n} \,{=}\, \sfb^2 \lambda_{4N-2-2n}$. 
Let $J_{2N+2+2n}$ be functions of $\wp_{1,2i-1}$, $\wp_{1,1,2i-1}$, $i=1$, \ldots, $N$,
and $\lambda_{2k}$, $k=1$, \ldots, $2N+1$,
such that the equations 
\begin{equation}\label{JacEqs}
J_{2N+2+2n} = 0,\qquad n=1, \dots, N, 
\end{equation}
serve as an algebraic model of $\Jac(\mathcal{V})$, 
see \cite[Sect.\,4]{BerWPFF2025}.
The Hamiltonians from \eqref{hExprs} expressed
 in terms of the $\wp$-functions by means of \eqref{WPAlphaGamma} 
 coincide with \eqref{JacEqs}, 
 provided that $\lambda_{4N+2}=0$.
\end{rem}

\subsection{Equations of motion in variables of separation}
From \eqref{PsiMatrEqs} we find
\begin{subequations}
\begin{align*}
&\frac{\rmd}{\rmd \sfx} \gamma(z) = - 2 \alpha_{2N-2} \gamma(z) + 2 \sfb z \alpha(z),\\
&\frac{\rmd}{\rmd \sft} \gamma(z) = - 2 \gamma_{-1} \alpha(z),
\end{align*}
\end{subequations}
where all dynamic variables are functions of $\sfx$ and $\sft$. Therefore, 
the zeros of $\gamma(z)$ are  functions of $\sfx$ and $\sft$ as well, namely
 $\gamma(z) \,{=}\, \sfb \prod_{i=1}^N (z-z_i(\sfx,\sft))$.
Then
\begin{subequations}
\begin{align*}
&\frac{\rmd}{\rmd \sfx} \log \gamma(z) = - \sum_{j=1}^N \frac{1}{z-z_j}\frac{\rmd z_j}{\rmd \sfx}
= - 2 \alpha_{2N-2} + 2 \sfb z \frac{\alpha(z)}{\gamma(z)},\\
&\frac{\rmd}{\rmd \sft} \log \gamma(z) = - \sum_{j=1}^N \frac{1}{z-z_j}\frac{\rmd z_j}{\rmd \sft}
= - 2 \gamma_{-1} \frac{\alpha(z)}{\gamma(z)}.
\end{align*}
\end{subequations}
Taking into account \eqref{PointsDef1}, we find as $z\to z_i$, $i=1$, \ldots, $N$,
\begin{gather}\label{DzDxEqs}
\frac{\rmd z_i}{\rmd \sfx}
= - \frac{2 w_i}{\prod_{j\neq i}^N (z_i - z_j)},\qquad
\frac{\rmd z_i}{\rmd \sft}
= - \frac{2 w_i \prod_{j\neq i} (-z_j)}{\prod_{j\neq i} (z_i - z_j) }.
\end{gather}

Let $D = \sum_{i=1}^N (z_i,w_i)$, and the $N$ points satisfy \eqref{PointsDef1}.
The Abel image of $D$
\begin{equation*}
u = \mathcal{A}(D) = \sum_{i=1}^N \int_{\infty}^{(z_i,w_i)} \rmd u = 
\sum_{i=1}^N \int_{\infty}^{(z_i,w_i)} 
\begin{pmatrix} z^{N-1} \\ \vdots \\ z \\ 1 \end{pmatrix} 
\frac{\sfb\, \rmd z}{-2w}
\end{equation*}
depends on $\sfx$ and $\sft$. Applying \eqref{DzDxEqs}, we find
\begin{subequations}
\begin{align*}
&\frac{\rmd u_{2n-1}}{\rmd \sfx} = \sum_{i=1}^N \frac{\sfb z_i^{N-n}}{-2w_i} \frac{\rmd z_i}{\rmd \sfx}
= \sum_{i=1}^N \frac{\sfb z_i^{N-n}}{\prod_{j\neq i}^N (z_i - z_j)} =  \sfb \delta_{n,1}, \\
& \frac{\rmd u_{2n-1}}{\rmd \sft} = \sum_{i=1}^N \frac{\sfb z_i^{N-n}}{-2w_i} \frac{\rmd z_i}{\rmd \sft}
=  \sum_{i=1}^N  \frac{\sfb z_i^{N-n}\prod_{j\neq i} (-z_j)}{\prod_{j\neq i} (z_i - z_j) } = \sfb \delta_{n,N}.
\end{align*}
\end{subequations}
And so, we obtain
\begin{gather}\label{uParam}
\begin{split}
u_{1} &= \sfb \sfx + C_1,\\
u_{2n-1} &= \sfb \mathsf{c}_{2n-1} + C_{2n-1} = \const,\quad n=2, \ldots, N-1, \\
u_{2N-1} &= \sfb \sft + C_{2N-1}.
\end{split}
\end{gather}
Here  $\mathsf{c}_{2n-1}$, $n=2, \ldots, N-1$, are arbitrary real constants, and
 $\bm{C} = (C_1$, $C_3$, \ldots, $C_{2N-1})$ is a constant vector chosen in such a way that
 $\gamma_{-1} \,{=}\, {-} \sfb \wp_{1,2N-1}(u)$ has the desired property, see subsection~\ref{ss:SineSinhParam}.
Thus, \eqref{SinGSolGamma} acquires the form 
\begin{equation}\label{KdVSolRealCond}
\gamma_{-1}(\sfx, \sft) 
= - \sfb \wp_{1,2N-1}\big(\sfb (\sfx, \mathsf{c}_{3}, \dots, \mathsf{c}_{2N-3}, \sft)^t  + \bm{C}\big),
\end{equation}
 where $\bm{C}$ serves as a shift.

Applying the transformation \eqref{tTrans}, we find
\begin{theo}\label{T:SineGSol}
A finite-gap solution to the sine-Gordon equation \eqref{SGEq}
associated with the genus $N$ spectral curve \eqref{CurveGN} is given by
\begin{subequations}\label{SinGSol}
\begin{align}
\begin{split}\label{sinGSolIm} 
&\phi(\x,\rmt) = - \imath \log \Big( \frac{\imath \sfb}{\sqrt{-\rmr_{-1}}}  
\wp_{1,2N-1}\big(u_{\ImN}(\rmx, \rmt)\big) \Big),
\quad  \sfb \in \imath \Real,\ \ \rmr_{-1}<0, \\
&\qquad\qquad\ \  u_{\ImN}(\rmx, \rmt) =  \sfb \Big(\x, \mathsf{c}_{3}, \dots, \mathsf{c}_{2N-3}, 
\frac{\rmt}{{-}\imath \sfb \sqrt{-\rmr_{-1}}} \Big)^t  + \bm{C}; 
\end{split} \\
\begin{split}\label{sinGSolRe}
&\phi(\x,\rmt) = - \imath \log \Big({-} \frac{\sfb}{\sqrt{\rmr_{-1}}}  
\wp_{1,2N-1} \big(u_{\ReN}(\rmx, \rmt)\big)  \Big), \quad
\sfb \in \Real, \ \ \rmr_{-1}>0, \\
&\qquad\qquad\ \  u_{\ReN}(\rmx, \rmt) =  \sfb \Big(\x, \mathsf{c}_{3}, \dots, \mathsf{c}_{2N-3}, 
\frac{\rmt}{{-} \sfb \sqrt{\rmr_{-1}}} \Big)^t  + \bm{C}.
\end{split}
\end{align}
\end{subequations}
\end{theo}

\begin{theo}\label{T:SinhGSol}
A finite-gap solution to the sinh-Gordon equation \eqref{ShGEq}
associated with the genus $N$ spectral curve \eqref{CurveGN} is given by
\begin{subequations}\label{SinhGSol}
\begin{align}
&\phi(\x,\rmt) =  \log \Big({-} \frac{\sfb}{\sqrt{\rmr_{-1}}}  
\wp_{1,2N-1} \big(u_{\ReN}(\rmx, \rmt)\big)  \Big), \quad
\sfb \in \Real, \ \ \rmr_{-1}>0; \label{sinhGSolRe} \\
&\phi(\x,\rmt) =  \log \Big( \frac{\imath \sfb}{\sqrt{-\rmr_{-1}}}  
\wp_{1,2N-1}\big(u_{\ImN}(\rmx, \rmt)\big) \Big),
\quad  \sfb \in \imath \Real, \ \ \rmr_{-1}<0. \label{sinhGSolIm} 
\end{align}
\end{subequations}
\end{theo}

\begin{rem}\label{R:bRealSol}
If the spectral curve \eqref{CurveGN} is reduced to the canonical 
form \eqref{V22g1Eq} with $\lambda_{4g+2}\,{=}\,0$,
$g\,{=}\,N$, then $\rmr_{-1} \,{=}\, \sfb^2 \lambda_{4g}$. 
The both conditions: $\rmr_{-1} \,{<}\, 0$ if $\sfb \,{\in}\, \imath \Real$
and $\rmr_{-1} \,{>}\, 0$ if $\sfb \,{\in}\, \Real$ are equivalent to the condition $\lambda_{4g} \,{>}\, 0$
for the both cases $\sfb \,{\in}\, \Real$ and $\sfb \,{\in}\, \imath \Real$.
The two expressions \eqref{SinGSol} acquire the form
\begin{equation}\label{sinGSolCanon} 
\phi(\x,\rmt) = - \imath \log \Big( {\pm} \lambda_{4g}^{-1/2}
\wp_{1,2N-1} \Big(\sfb \big(\rmx, \mathsf{c}_{3}, \dots, \mathsf{c}_{2N-3}, 
{\pm} \sfb^{-2} \lambda_{4g}^{-1/2}\rmt  \big)^t  + \bm{C}\Big) \Big),  \tag{\ref{SinGSol}'}
\end{equation}
and  the two expressions \eqref{SinhGSol} acquire the form
\begin{equation}\label{sinhGSolCanon} 
\phi(\x,\rmt) =  \log \Big({\pm} \lambda_{4g}^{-1/2}
\wp_{1,2N-1} \Big(\sfb \big(\rmx, \mathsf{c}_{3}, \dots, \mathsf{c}_{2N-3}, 
{\pm} \sfb^{-2} \lambda_{4g}^{-1/2}\rmt  \big)^t  + \bm{C}\Big) \Big),    \tag{\ref{SinhGSol}'}
\end{equation}
`$+$' for $\sfb\,{\in}\, \Real$, `$-$' for $\sfb\,{\in}\, \imath \Real$.
\end{rem}

\section{Reality conditions}\label{s:RealCond}
Reality conditions for the  sine(sinh)-Gordon equations single out real-valued and smooth solutions.
Real-valued solutions are obtained if the desired properties of the dynamic variables for the Hamiltonian systems from the corresponding hierarchy are satisfied (see subsection~\ref{ss:RealForms}).
By applying the parametrization \eqref{WPAlphaGamma} in terms 
of the $\wp$-functions, we come to the problem of finding particular paths 
in the Jacobian variety $\Jac(\mathcal{V})$ where these conditions are met.
Taking into account  that the argument $u$ 
of the $\wp$-functions is given by \eqref{uParam}, 
we consider affine subspaces in $\Jac(\mathcal{V})$
of the form 
\begin{equation}\label{AffSubSp}
u_{\ReN}\,{=}\, s \,{+}\, \bm{C},\quad \text{and} \quad
u_{\ImN}\,{=}\, \imath s \,{+}\, \bm{C},\quad  s\in\Real^g.
\end{equation}
In this paper, we focus on half-period shifts $\bm{C}$ and use
the addition law on $\mathcal{V}$ to find such $\bm{C}$ that 
the dynamic variables possess the desired properties.
Among real-valued solutions, we choose those without singularities, which we call smooth.

In what follows, we work with the canonical form $\mathcal{V}$ of the
 spectral curve \eqref{CurveGN}, associated with an $N$-gap
Hamiltonian system from the sine(sinh)-Gordon hierarchy.
The canonical form $\mathcal{V}$ is defined by equation
 \eqref{V22g1Eq} with $\lambda_{4g+2}\,{\equiv}\, 0$,
implying that the curve $\mathcal{V}$ has branch points at both zero and infinity.
 All parameters $\lambda_k$ of  \eqref{V22g1Eq}  are real, as the 
 integrals of motion $\rmh_n$ and $\rmr_n$ in \eqref{CurveGN} are required to be real.
 
For brevity, we consider equation \eqref{V22g1Eq} in the  form $y^2 \,{=}\, \Lambda(x)$. 
Let $\{e_k\}_{k=1}^{2g+1}$ be roots of $\Lambda$,
enumerated  in the ascending order 
of their real and imaginary parts, as shown in fig.\,\ref{cyclesOddCC}.  

The sinh-Gordon equation arises on the affine algebras of $\mathfrak{sl}(2,\Real)$ (where $\sfb\,{\in}\,\Real$)
and $\imath \mathfrak{sl}(2,\Real)$  (where $\sfb\,{\in}\,\imath \Real$).
The properties of $\wp$-functions  established in \cite[Proposition\;1]{BerKdV2024} guarantee that 
the dynamic variables for Hamiltonian systems in the sinh-Gordon hierarchy
are consistent with these affine algebras, see subsection~\ref{ss:RealForms}.
 
The sine-Gordon equation arises on the affine algebras of $\mathfrak{su}(2)$ (where $\sfb\,{\in}\,\imath \Real$)
and $\imath \mathfrak{su}(2)$  (where $\sfb\,{\in}\,\Real$).
In these cases, the conditions on the dynamic variables to be consistent with 
the affine algebras are more complicated. Specifically, as shown in subsection\;\ref{ss:SineSinhParam},
the constraint \eqref{hm1Rel} implies:
\begin{gather}\label{su2Cond}
\begin{aligned}
&\mathfrak{g} = \mathfrak{su}(2): & 
&|\gamma_{-1}|^2 = -\rmr_{-1} = -\sfb^2 \lambda_{4g}, \quad \sfb \in \imath \Real, \\
&\mathfrak{g} = \imath \mathfrak{su}(2): & & |\gamma_{-1}|^2 = \rmr_{-1} = \sfb^2 \lambda_{4g},
\quad \sfb \in  \Real.
\end{aligned}
\end{gather}
This condition is used to identify real-valued solutions to the sine-Gordon equation
and resembles the one presented in \cite[Lemma\;3]{DN1982}.

\subsection{Complex conjugate branch points}
We assume that all non-zero branch points of $\mathcal{V}$ come in complex conjugate pairs.
The enumeration of branch points implies that the zero branch point has an odd index, 
and its position can be arbitrary.

\begin{prop}\label{P:CCbp}
Suppose a curve $\mathcal{V}$ of the form  \eqref{V22g1Eq} 
possesses the branch points $\{(e_k,0)\}_{k=1}^{2g+1}$ such that
$e_{2n+1}\,{=}\,0$.
If there are $g$ pairs of complex conjugate branch points given by 
$\{e_{2i-1}, e_{2i}\}$ for $i\,{=}\,1$, \ldots, $n$ and 
$\{e_{2i}, e_{2i+1}\}$ for  $i\,{=}\,n+1$,\ldots, $g$, then: 
\begin{equation}\label{Lambda4g}
\lambda_{4g} = \prod_{i=1}^n e_{2i-1} e_{2i} \prod_{i=n+1}^{g} e_{2i} e_{2i+1}
= \prod_{i=1}^g |e_{2i}|^2.
\end{equation}
\end{prop}
\begin{proof}
The equality is obtained by direct computation. 
Note that each pair of complex conjugate roots of $\Lambda$
contains a root with an even index, say $e_{2i}$.
The other root in each pair has an odd index; specifically, $e_{2i-1} \,{=}\, \bar{e}_{2i}$
for $i\,{=}\,1$, \ldots, $n$ and $e_{2i+1} \,{=}\, \bar{e}_{2i}$ for $i\,{=}\,n\,{+}\,1$, \ldots, $g$.
Thus, the product of the roots in the $i$-th pair equals $|e_{2i}|^2$. 
\end{proof}

Therefore, the spectral curves associated with the sine-Gordon hierarchy 
take the canonical form $\mathcal{V}$ with 
$\lambda_{4g} \,{>}\, 0$, satisfying the conditions on $\rmr_{-1}$ detailed in subsection
\ref{ss:RealForms}. It follows that
\begin{equation}\label{SqLambda4g}
\sqrt{\lambda_{4g}} = \prod_{i=1}^g |e_{2i}|.
\end{equation}
Each root $e_{2i}$ in the product on the right-hand side of \eqref{SqLambda4g}
may be replaced by its complex conjugate.
Consequently, we obtain $2^g$ distinct collections of $g$ non-zero roots that yield the same product.

\begin{prop}\label{P:Radius}
There exist $2^g$ ways to compose the product \eqref{SqLambda4g} 
from $g$ pairs of complex conjugate roots of $\Lambda$
by choosing one point from each pair. 
These  $2^g$ collections of $g$ branch points form $2^g$ non-special divisors.
\end{prop}

The branch point divisors singled out by Proposition~\ref{P:Radius} necessarily belong to 
the domain of $\wp_{1,2g-1}$ in \eqref{SinGSol}.
Indeed, a $g$-gap  Hamiltonian system, 
expressed in terms of variables of separation $(z_i,w_i)$, 
splits into $g$ independent systems. Each system is defined by
the coordinate $z_i$, the momentum $w_i$, 
and a Hamiltonian function given by the spectral curve equation \eqref{CurveGN}. 
Consequently, the trajectories of the Hamiltonian system contain branch points as turning points. 

\begin{rem}
The  $2^g$ collections $\I$ of $g$ branch points defined by Proposition~\ref{P:Radius}
correspond to the half-periods $\Omega_{\I}$ such that
\begin{equation}\label{HPwp}
|\wp_{1,2g-1}(\Omega_{\I})| = \sqrt{\lambda_{4g}}.
\end{equation}
The equality follows from the solution of the Jacobi inversion problem given by \eqref{EnC22g1}.
Taking into account that $\gamma_{-1} \,{=}\, {-}\sfb \wp_{1,2g-1}(u)$, we see that $\Omega_{\I}$
defined in \eqref{HPwp} comply with \eqref{su2Cond} and, therefore, belong to 
the required path in $\Jac(\mathcal{V})$.
Note that if $\mathcal{V}$ remains real but not all non-zero branch points are complex conjugates,
then \eqref{SqLambda4g} and also \eqref{HPwp}  do not hold; 
consequently, property \eqref{su2Cond} fails on the affine subspaces \eqref{AffSubSp}
with half-period shifts. 
\end{rem}

\subsection{Singularities of $\wp$-functions}
Each collection of $g$ branch points forms a divisor whose Abel image is an
even non-singular half-period. Odd  non-singular and singular half-periods are Abel images
of divisors of degrees less than $g$.
As explained in subsection~\ref{ss:CharPart}, 
each half-period is associated with a partition $\I\cup \J$ of the indices of all branch points,
where the cardinality of $\I$ is equal to or less than $g$, that is $|\I| \,{\leqslant}\, g$.
The $\wp$-functions have singularities at all half-periods with $|\I| \,{<}\, g$.
Therefore,  all half-periods in the domain of smooth solutions 
 correspond to non-special divisors, which consist of $g$ distinct branch points.
This guarantees that the $\wp$-functions have finite values.

\begin{theo}\label{P:ReImPeriods}
Let a curve $\mathcal{V}$ of the form  \eqref{V22g1Eq} 
 have branch points as defined in Proposition~\ref{P:CCbp}. 
For the choice of cycles shown in  fig.\,\ref{cyclesOddCC} 
and the standard non-normalized holomorphic differentials \eqref{K1DifsGen}, 
the period lattice generated by each pair of periods $\{\omega_{k}$,
$\omega'_{k}\}$ is rhombic, with generators:
\begin{enumerate}
\renewcommand{\labelenumi}{\arabic{enumi})}
\item $\omega'_{k}$ and
$\omega_{k} - \omega'_{k}$, $k=1$, \dots, $n$, where 
\begin{subequations}\label{PerRels}
\begin{equation}\label{ImRels}
\omega_k \in \Real,\qquad  \ReN \omega'_k = \tfrac{1}{2} \omega_k;
\end{equation}
\item  $\omega_{k}$ and $\omega_{k} - 2 \imath \ImN \omega_{k}$, $k=n+1$, \dots,~$g$, where 
\begin{equation}\label{ReRels}
\begin{split}
\omega'_k \in \imath \Real,\quad 
&\imath \ImN \omega_{n+1} =  - \tfrac{1}{2}  (\omega'_{n+1} - \omega'_{n+2}),\\
&\imath \ImN \omega_j = 
- \tfrac{1}{2} (\omega'_{j} - \omega'_{j+1})  
+ \tfrac{1}{2} (\omega'_{j-1} - \omega'_{j}),\  j=n+2,\, \dots,\, g-1,\\
&\imath \ImN \omega_g = - \tfrac{1}{2}  \omega'_g  + \tfrac{1}{2} (\omega'_{g-1} - \omega'_g ).
\end{split}
\end{equation}
\end{subequations}
\end{enumerate}
\end{theo}
\begin{proof}
Since the only real branch point is $e_{2n+1}\,{=}\,0$,
we have $\Lambda(x)\,{<}\,0$ for $x\,{<}\,e_{2n+1}$
and $\Lambda(x)\,{>}\,0$ for $x\,{>}\,e_{2n+1}$.

(I) For $k\,{=}\,1$, \ldots, $n$,  we have $\tfrac{1}{2} \omega_k \,{=}\, \mathcal{A} (e_{2k}) \,{-}\,  \mathcal{A} (\bar{e}_{2k})$,
where $\mathcal{A}$ is computed along the monodromy path 
that goes below and counter-clockwise around all cuts, as explained in \cite[Sect.\,3]{BerCompWP2024}.
Expanding $\mathcal{A} (e_{2k})$ and $\mathcal{A} (\bar{e}_{2k})$  in a Taylor series about $\ReN e_{2k}$, 
and taking into account that the first kind integral along the real axis from infinity to $\ReN e_{2k}$  
is purely imaginary (since $\Lambda(x)\,{<}\,0$ to the left of $\ReN e_{2k}$), 
we find that $\ImN \mathcal{A} (\bar{e}_{2k}) \,{=}\, \ImN \mathcal{A} (e_{2k})$.
Thus,  $\omega_k$ is real.   

(II) Next, we consider
$\tfrac{1}{2}  \omega'_k \,{=}\, {-}\sum_{i=k}^g \big(\mathcal{A}(\bar{e}_{2i}) \,{-}\,  \mathcal{A} (e_{2i})\big)$
for $k\,{=}\,n\,{+}\,1$, \ldots, $g$. Since $\Lambda(x) \,{>}\, 0$ to the right of $\ReN e_{2k}$, 
we find that  $\ReN \mathcal{A} (\bar{e}_{2k}) \,{=}\,  \ReN \mathcal{A} (e_{2k})$, 
which implies $\omega'_k \,{\in}\, \imath \Real$.

(III) For $k\,{=}\,1$, \ldots, $n$, we have
\begin{equation}\label{FKprim}
 \tfrac{1}{2}  \omega'_k = - \sum_{j=k}^n \big(\mathcal{A}(e_{2j+1}) -  \mathcal{A}( e_{2j}) \big)
- \sum_{j=n+1}^g \big( \mathcal{A}(\bar{e}_{2j}) -  \mathcal{A} (e_{2j})\big).
\end{equation}
The second sum  is purely imaginary, as  shown above. 
Each term of the first sum is computed as follows:
\begin{multline}\label{eeInt}
\mathcal{A}(e_{2j+1}) -  \mathcal{A} (e_{2j}) =
 \big(\mathcal{A} ( e_{2j+1}) -  \mathcal{A} (\ReN e_{2j+1})  \big) \\
 + \big( \mathcal{A} (\ReN e_{2j+1})  - \mathcal{A} (\ReN e_{2j}) \big)
 + \big(\mathcal{A} (\ReN e_{2j}) -  \mathcal{A} (e_{2j}) \big).
 \end{multline}
The middle term on the right hand side is purely imaginary. 
Taking into account (I), we find
\begin{align*}
&\ReN \big(\mathcal{A} ( e_{2j+1}) -  \mathcal{A} (\ReN e_{2j+1}) \big) = - \tfrac{1}{4}  \omega_{j+1},& 
& j=k,\dots, n-1,\\
&\ReN \big(\mathcal{A} (\ReN e_{2j}) -  \mathcal{A} (e_{2j}) \big) = - \tfrac{1}{4}  \omega_{j},& & j=k,\dots, n,
\end{align*}
and $\mathcal{A} ( e_{2n+1}) \,{-}\,  \mathcal{A} (\ReN e_{2n+1})  \,{=}\, 0$ due to $e_{2n+1}\,{=}\,0$. 
Thus,  $\ReN \omega'_k \,{=}\, \frac{1}{2} \ReN \omega_k$, cf.\;\eqref{ImRels}.

(IV) Finally, for $k\,{=}\,n\,{+}\,1$, \ldots, $g$, we have 
$\tfrac{1}{2} \omega_k \,{=}\,  \mathcal{A} (e_{2k}) \,{-}\, \mathcal{A} (e_{2k-1})$,
which is computed similarly to \eqref{eeInt}. Namely,
\begin{multline*}
 \mathcal{A} (e_{2k}) - \mathcal{A} (e_{2k-1}) =
 \big(\mathcal{A} ( e_{2k}) -  \mathcal{A} (\ReN e_{2k})  \big) 
 + \big( \mathcal{A} (\ReN e_{2k})  - \mathcal{A} (\ReN e_{2k-1}) \big) \\
 + \big(\mathcal{A} (\ReN e_{2k-1}) -  \mathcal{A} (e_{2k-1}) \big).
 \end{multline*}
 The middle term on the right hand side is real,
since $\Lambda(x)\,{>}\,0$ on the interval $[\ReN e_{2k-1},\ReN e_{2k}]$.
The first kind integral  $\mathcal{A} (\ReN e_{2k-1}) \,{-}\,  \mathcal{A} (e_{2k-1})$ vanishes at $k\,{=}\,n\,{+}\,1$;
and for $k\,{>}\,n\,{+}\,1$, the integral is taken along the right edge of the cut $[e_{2k-2},e_{2k-1}]$, 
that is, in the negative direction.
Then,  $\mathcal{A} ( e_{2k}) -  \mathcal{A} (\ReN e_{2k})$ 
is taken along the left edge of the cut $[e_{2k},e_{2k+1}]$,
 which coincides with the positive direction of the monodromy path.
 Taking into account (II), we find
\begin{align*}
&\ImN \big(\mathcal{A} ( e_{2k}) -  \mathcal{A} (\ReN e_{2k}) \big)  = 
 -  \tfrac{1}{2} \ImN \big(\omega'_k -  \omega'_{k+1}\big),\quad j=n+1,\dots, g-1,\\
&\ImN \big(\mathcal{A} ( e_{2g}) -  \mathcal{A} (\ReN e_{2g}) \big)  
=  - \tfrac{1}{2} \ImN \omega'_g,\\
&\ImN \big( \mathcal{A} (\ReN e_{2k-1})  - \mathcal{A} ( e_{2k-1}) \big) 
= \tfrac{1}{2} \ImN \big(\omega'_{k-1} -  \omega'_{k}\big),\quad j=n+2,\dots, g.
\end{align*}
This implies  \eqref{ReRels}.

The relations \eqref{PerRels} guarantee that the period lattice formed by each pair of periods $\omega_{k}$,
$\omega'_{k}$ is rhombic, with the indicated generators.
\end{proof}

Let `$\sim$' denote the congruence relation on $\Jac(\mathcal{V})$.
\begin{cor}
Under the conditions of Theorem~\ref{P:ReImPeriods}, we have
\begin{subequations}\label{ReImOmega0}
\begin{align}
& \ReN \Omega \sim 0,\quad \text{if} \quad \Omega = \tfrac{1}{2}\omega_k,\quad k=1,\,\dots,\,n,  \label{ReOmega0} \\
& \ImN \Omega \sim 0,\quad \text{if} \quad \Omega = \tfrac{1}{2}\omega'_k, \quad k=n+1,\,\dots,\,g.  \label{ImOmega0} 
\end{align}
\end{subequations}
\end{cor}
\begin{proof}
Evidently, \eqref{ReOmega0} follows from \eqref{ImRels}.
Next, recurrence relations for $\tfrac{1}{2} \omega'_k$, $k=n+1$, \ldots, $g$
are obtained from \eqref{ReRels}, namely
\begin{gather}
\begin{split}
& \tfrac{1}{2} \omega'_g = \imath \textstyle \sum_{j=n+1}^g \ImN \omega_j,\\
& \tfrac{1}{2} \omega'_{g-1} = \omega'_g + \imath \ImN \omega'_g,\\
& \tfrac{1}{2} \omega'_{j} = \omega'_{j+1} + \imath \ImN \omega_{j+1} - \tfrac{1}{2}  \omega'_{j+2},
\quad j=g-2,\, \dots,\, n+2,\\
& \tfrac{1}{2} \omega'_{n+1} = \tfrac{1}{2} \omega'_{n+2} - \imath \ImN \omega_{n+1} .
\end{split}
\end{gather}
which implies \eqref{ImOmega0}.
\end{proof}

Let $\JFr = \Complex^g$ be the vector space 
where $\Jac(\mathcal{V})$ is embedded.
Let $\JFr_{\ReN} \sim \Real^g$ be the span of the real axes of $\JFr$ over $\Real$,
and $\JFr_{\ImN} \sim \Real^g$ be the span of the imaginary axes of $\JFr$ over~$\Real$.
Then $\JFr = \JFr_{\ReN} \oplus  \JFr_{\ImN}$.  As seen from \eqref{SinGSol}, 
the domain of $\wp_{1,2N-1}$  is an affine subspace of one of the forms \eqref{AffSubSp}.
For brevity, we denote $\bm{C}\,{+}\,\JFr_{\ReN} = \{u_{\ReN}\,{=}\, s\,{+}\,\bm{C}$, 
$s\,{\in}\,\Real^g\}$ and $\bm{C}\,{+}\,\JFr_{\ImN} = \{u_{\ImN}\,{=}\, \imath s\,{+}\,\bm{C}$, 
$s\,{\in}\,\Real^g\}$. The affine subspaces $\bm{C}\,{+}\,\JFr_{\ReN}$ are parallel to $\JFr_{\ReN}$
and are used if $\sfb\,{\in}\, \Real$. The affine subspaces $\bm{C}\,{+}\,\JFr_{\ImN}$ are parallel to $\JFr_{\ImN}$
and are used if $\sfb\,{\in}\,\imath \Real$. Note that $\bm{C}$ is restricted to be a half-period.
Also, $\wp_{1,2g-1}$ is required to be  bounded. 

\begin{theo}\label{P:Cshift}
The $2^g$ half-periods obtained from the  divisors identified in Proposition~\ref{P:Radius}
lie within the affine subspace $\bm{K} \,{+}\, \JFr_{\ImN}$,
where
\begin{equation}\label{KDef}
\bm{K} = \sum_{i=0}^{[(g-1)/2]} \tfrac{1}{2}  \omega_{g-2i}  + \sum_{k=1}^g \tfrac{1}{2}  \omega'_k
= \mathcal{A} \Big(\sum_{i=1}^g e_{2i} \Big).
\end{equation}
Simultaneously, these $2^g$ half-periods are contained within the affine subspace $\bm{K} \,{+}\,  \JFr_{\ReN}$.
\end{theo}
\begin{proof}
With the cycles chosen as shown in fig.\,\ref{cyclesOddCC}, we have
the following correspondence between sets $\I$ of cardinality $1$ and half-periods:
\begin{gather}\label{HPchar}
\begin{split}
&\mathcal{A}(\{2k-1\}) \sim \tfrac{1}{2} \omega'_{k} + \sum_{i=1}^{k-1} \tfrac{1}{2}  \omega_i,\quad
\mathcal{A}(\{2k\}) \sim  \tfrac{1}{2} \omega'_k + \sum_{i=1}^k \tfrac{1}{2}  \omega_i,\quad k=1,\dots, g, \\
&\mathcal{A}(\{2g+1\}) \sim  \sum_{i=1}^g \tfrac{1}{2}  \omega_i.
\end{split}
\end{gather}
Eqs.\,\eqref{HPchar} imply
$$\{2,4,\dots,2g\} \sim  \sum_{i=0}^{[(g-1)/2]} \tfrac{1}{2}  \omega_{g-2i}  
+ \sum_{k=1}^g \tfrac{1}{2}  \omega'_k,$$
which proves \eqref{KDef}.
In the expressions above, $\{n\}$ represents the divisor $D_1\,{=}\,e_n$, 
which is composed of a single brach point $e_n$,
while $\{2,4,\dots,2g\}$ represents the divisor $D_K=\sum_{i=1}^g e_{2i}$.

Analyzing \eqref{HPchar} and taking into account \eqref{PerRels}, we find the following:
\begin{itemize}
\item if $k=1$, \ldots $n$, then
\begin{align*}
&\ReN \mathcal{A}(\{2k\}) \sim \ReN \mathcal{A}(\{2k-1\}) \sim \tfrac{1}{2} \ReN \omega'_k,& \\
&\ImN \mathcal{A}(\{2k\}) \sim \ImN \mathcal{A}(\{2k-1\}) \sim \tfrac{1}{2} \ImN \omega'_k,& 
\end{align*}
since $\tfrac{1}{2} \ReN \omega_k \sim \ReN \omega'_k \sim 0$,
and $\tfrac{1}{2} \ImN \omega_k =0$ for $k=1$, \ldots, $n$;

\smallskip
\item if $k=n+1$, \ldots $g$, then
\begin{align*}
&\ReN \mathcal{A}(\{2k\}) \sim \ReN \mathcal{A}(\{2k+1\}) \sim   
\tfrac{1}{2} \textstyle  \sum_{i=n+1}^{k} \ReN \omega_i,& \\
&\ImN \mathcal{A}(\{2k\}) \sim \ImN \mathcal{A}(\{2k+1\}) \sim   
\tfrac{1}{2} \textstyle  \sum_{i=n+1}^{k} \ImN \omega_i,&
\end{align*}
which follows from \eqref{PerRels}, and \eqref{ReImOmega0}.
\end{itemize}
This implies that any even number in the set $\I = \{2,4,\dots,2g\}$ 
can be replaced with its odd counterpart---namely $2i-1$ for $i=1$, \ldots $n$, 
or $2i+1$ for $i=n+1$, \ldots, $g$---and
 $\mathcal{A}(\I) $ remains unchanged and congruent to $\bm{K}$.
\end{proof}

\begin{rem}
Assuming that $\mathcal{V}$ is defined as in Proposition\;\ref{P:CCbp},
the $2^{2g}$ half-periods are distributed among $2^g$ affine subspaces $\bm{C} \,{+}\,\JFr_{\ImN}$.
The half-period shifts $\bm{C}$ are generated from
$\{\frac{1}{2}\omega_k \mid k\,{=}\,1,\, \dots,\,n\}\cup
\{\frac{1}{2}\ReN \omega_k \mid k\,{=}\,n\,{+}\,1,\, \dots,\, g\}$, and
each subspace contains $2^{g}$ half-periods. 
According to Theorem~\ref{P:Cshift},
the half-periods located within $\bm{K} \,{+}\, \JFr_{\ImN}$  are identified by Proposition~\ref{P:Radius}.
In contrast, the remaining subspaces $\bm{C} \,{+}\, \JFr_{\ImN}$
contain half-periods belonging to $\Sigma$, where the $\wp$-functions have singularities.

Similarly, the $2^{2g}$ half-periods also partition into $2^g$ 
affine subspaces $\bm{C} \,{+}\,\JFr_{\ReN}$.
In this case, the shifts $\bm{C}$ are generated from 
$\{\frac{1}{2} \ImN \omega'_k \mid k\,{=}\,1,\, \dots,\,n\}\cup 
\{\frac{1}{2} \omega'_k \mid k\,{=}\,n\,{+}\,1$, \ldots, $g\}$.
Each such subspace contains $2^{g}$ half-periods. As before, the half-periods within 
 $\bm{K} \,{+}\, \JFr_{\ReN}$  are identified by Proposition~\ref{P:Radius}, 
(see Theorem~\ref{P:Cshift}), while the
other subspaces $\bm{C} \,{+}\, \JFr_{\ReN}$ contain half-periods that belong to $\Sigma$.
\end{rem}
\begin{cor}
The $\wp$-functions associated with $\mathcal{V}$ are bounded on
the affine subspaces $\bm{K} + \JFr_{\ImN}$
and  $\bm{K} + \JFr_{\ReN}$, where $\bm{K}$ is defined as in \eqref{KDef}.
\end{cor}

\begin{rem}
Note that $\bm{K}$ serves as the vector of Riemann constants in not normalized coordinates
on $\Jac(\mathcal{V})$, cf.\;\eqref{Kchar}.
\end{rem}

In the sinh-Gordon hierarchy, $\wp_{1,3}(u)\,{\in}\,\Real$ and is bounded on $\bm{K} \,{+}\, \JFr_{\ReN}$
and $\bm{K} \,{+}\, \JFr_{\ImN}$ (see \cite[Propositions\;2 and 3]{BerKdV2024}).
At the same time, the finite-gap solution \eqref{SinhGSol} to the sinh-Gordon equation
has further restrictions: a singularity arises when the argument of the logarithm reaches $0$.
Therefore, we exclude curves $\mathcal{V}$ for which the associated $\wp_{1,2g-1}$
acquires both positive and negative values.

\begin{theo}\label{P:KshiftRe}
Assume $\mathcal{V}$ has only real branch points, with the zero
branch point $e_n$ at an arbitrary position, and let $\bm{K}$ be defined by \eqref{KDef}. 
The $2^g$ half-periods located (up to congruence)
within the affine subspace $\bm{K} \,{+}\, \JFr_{\ReN}$  
correspond to divisors consisting of one branch point from each of the $g$ pairs $\{e_{2i-1},e_{2i}\}_{i=1}^g$.
Similarly, the $2^g$ half-periods within
the affine subspace $\bm{K} \,{+}\, \JFr_{\ImN}$ correspond to divisors consisting of
one branch point from each of the $g$ pairs $\{e_{2i}, e_{2i+1}\}_{i=1}^g$.
\end{theo}
\begin{proof}
If all branch points of $\mathcal{V}$ are real, then $\omega_k \,{\in}\, \JFr_{\ReN}$ and 
$\omega'_k \,{\in}\, \JFr_{\ImN}$ for  $k\,{=}\, 1$, \ldots, $g$.
The affine subspace $\bm{K} +  \JFr_{\ReN}$ is characterized by 
$\ImN u_{\ReN} \,{=}\, \ImN \bm{K}\,{=}\,  \sum_{k=1}^g \tfrac{1}{2} \ImN  \omega'_k$.
In view of \eqref{HPchar}, we see that $\ImN \mathcal{A}(\{2k-1\}) \,{=}\, \ImN \mathcal{A}(\{2k\})$;
thus, one point from each pair $\{e_{2k-1},e_{2k}\}$ 
must be chosen to form a divisor that maps into $\bm{K} +  \JFr_{\ReN}$. 
Similarly, the subspace $\bm{K} +  \JFr_{\ImN}$ is characterized by
$\ReN u_{\ImN} \,{=}\, \ReN \bm{K}\,{=}\,  \sum_{i=0}^{[(g-1)/2]} \tfrac{1}{2}  \omega_{g-2i}$.
From \eqref{HPchar}, it follows that $\ReN \mathcal{A}(\{2k\}) \,{=}\, \ReN \mathcal{A}(\{2k+1\})$,
and therefore, one point from each pair $\{e_{2k}, e_{2k+1}\}$ is selected to compose
a  divisor mapping into $\bm{K} \,{+}\,  \JFr_{\ImN}$.
\end{proof}

\begin{cor}\label{C:ShGReal}
Let all branch points of $\mathcal{V}$ be real, and let $e_{n}\,{=}\,0$ for some
$n\,{\in}\,\{i\}_{i=1}^{2g+1}$. Then  $\wp_{1,2g-1}(u)\,{<}\,0$ for all
$u \,{\in}\, \bm{K} \,{+}\, \JFr_{\ReN}$ if and only if $n\,{=}\,2g\,{+}\,1$,
which means all non-zero branch points are negative.

Similarly, $\wp_{1,2g-1}(u)\,{>}\,0$
or $\wp_{1,2g-1}(u)\,{<}\,0$ for all
$u \,{\in}\, \bm{K} \,{+}\, \JFr_{\ImN}$ if and only if $n\,{=}\,0$,
which means all non-zero branch points are positive.
\end{cor}
\begin{proof}
Recall that $\wp_{1,2g-1}(u)\,{=}\, (-1)^{g-1}  \prod_{i=1}^{g} x_i$,
provided $u\,{=}\,\mathcal{A}(D)$ and $D$ is 
a divisor on $\mathcal{V}$ of the form $D=\sum_{i=1}^g (x_i,y_i)$.
Theorem\;\ref{P:KshiftRe} describes divisors corresponding to the half-periods from 
$\bm{K} \,{+}\, \JFr_{\ReN}$ as those containing
one branch point from each of the $g$ pairs $\{e_{2i-1},e_{2i}\}_{i=1}^g$.
Thus, if any branch point $e_k$ ($k=1$, \ldots, $2g$)
involved in constructing these divisors is zero,
then $\wp_{1,2g-1}$ vanishes on $\bm{K} \,{+}\, \JFr_{\ReN}$,
violating the required condition on $\wp_{1,2g-1}$.
Further, if the zero branch point is $e_{2g+1}$, this implies $e_k\,{<}\, e_{2g+1}$ for
$k\,{=}\,1$, \ldots, $2g$,  based on the ordering of branch points.
Consequently, $\wp_{1,2g-1}$ at all half-periods within  $\bm{K} \,{+}\, \JFr_{\ReN}$ has the same sign,
depending on the parity of $g$. The entire evolution of $\wp_{1,2g-1}$ within  $\bm{K} \,{+}\, \JFr_{\ReN}$
lies between the values of $\wp_{1,2g-1}$ at half periods, which serve as turn points.
Therefore, for all $u\,{\in}\,\bm{K} \,{+}\, \JFr_{\ReN}$ we have $\wp_{1,2g-1}(u)\,{<}\,0$.

Similar considerations show that for all $u\,{\in}\,\bm{K} \,{+}\, \JFr_{\ImN}$
$\wp_{1,2g-1}(u)\,{>}\,0$ if $g$ is odd, and $\wp_{1,2g-1}(u)\,{<}\,0$ if $g$ is even.

\end{proof}

\subsection{Real-valued $\wp$-functions}
Recall that the function \eqref{KdVSolRealCond}
satisfies the constraint \eqref{hm1Rel}, which  in
 the sine-Gordon hierarchy is equivalent to
\begin{equation}\label{AbsWPLambda4g}
|\wp_{1,2g-1}\big(\sfb (\x, \mathrm{c}_{3}, \dots, \mathrm{c}_{2g-3}, \rmt)^t  + \bm{C}\big)|^2 = \lambda_{4g},
\end{equation}
where $g\,{\equiv}\, N$.
Below, we prove that  \eqref{AbsWPLambda4g} holds if $\mathcal{V}$ is
as  described in Proposition~\ref{P:CCbp} and
 $\bm{C}\,{=}\,\bm{K}$ is defined by \eqref{KDef}. 
Recall that $\wp_{i,j}(s)$,  $\wp_{i,j,k}(s)$, and $\wp_{i,j}(\imath s)$ are real-valued as $s \,{\in}\, \Real^g$, and
$\wp_{i,j,k}(\imath s)$ have purely imaginary values.
Values of $\wp_{1,2i-1}(\bm{K})$ are symmetric functions in $\{e_{2i}\}_{i=1}^g$
 computed from \eqref{R2g},
and $\wp_{1,1,2i-1}(\bm{K})\,{=}\,0$, as follows from \eqref{R2g1}.
Let $u_\text{I} \equiv  \bm{K}$, and 
$u_\text{II} \,{=}\,  s$ if $\sfb\,{\in}\,\Real$ or
$u_\text{II} \,{=}\,  \imath s$ if $\sfb\,{\in}\,\imath \Real$. 

In this subsection, we use the addition law as presented in Appendix\;\ref{A:AddLaw}.
In particular, quantities $\{\nu_i\}$ are introduced by \eqref{NuIntro}.

\begin{lemma}\label{P:EtaEta}
Let  $\bm{K}$ be defined by \eqref{KDef} and $s \in \Real^g$. Then
\begin{equation}\label{etaSqRe}
\nu_{3g} \bar{\nu}_{3g} = (-1)^{g-1} \lambda_{4g} \wp_{1,2g-1}(s).
\end{equation}
\end{lemma}

\begin{proof}
By Cramer's rule, from the addition law \eqref{AddLawEqs},  we find
\begin{gather}\label{eta3gGen}
\nu_{3g} = \frac{\begin{vmatrix} \Cp_g(\bm{K}) & \Qp(\bm{K})  \\ \Cp_g (s) & \Qp(s) 
\end{vmatrix}}
{\begin{vmatrix} 1_g & \Qp(\bm{K}) \\
1_g & \Qp(s) \end{vmatrix}}
=  \frac{\begin{vmatrix} \Cp_g(\bm{K}) & \Pp (\bm{K}) & 0   \\ \Cp_g (s) &  \Pp(s) & \Rp(s)
\end{vmatrix}}
{\begin{vmatrix} 1_g & \Pp (\bm{K}) & 0  \\
1_g &  \Pp(s) & \Rp(s) \end{vmatrix}} = \frac{|\Np|}{|\Dp|}.
\end{gather}
where $\Cp_g(u)$ is the identity matrix with its last column replaced by ${-} \mathbf{q}(u)$,
while $\Np$ and $\Dp$~denote the matrices that arise in the numerator and denominator, respectively.
The matrix $\Qp(u)$  contains $[\frac{1}{2}g]$ columns linear in 
the functions $\wp_{1,1,2i-1}(u)$ (we denote this block by $\Rp(u)$),
and $g-[\frac{1}{2}g]$ columns containing only the functions $\wp_{1,2i-1}(u)$ (denoted by $\Pp(u)$).
Entries of $\Pp (\bm{K})$ and $\Cp_g(\bm{K})$ are expressed in terms
of elementary symmetric functions in $\{e_{2i}\}_{i=1}^g$, and
 $\Rp(\bm{K}) \,{=}\, 0$, as $\wp_{1,1,2i-1}$ vanish at half-periods. 

Note that the numerator and the denominator in \eqref{eta3gGen}
 contain the same block $\Rp(s)$ of size $g\times [\frac{1}{2}g]$.
 Let $[\Rp]_{I}$ denote a principal minor of $\Rp(s)$ of order $[\frac{1}{2}g]$, 
 where $I$ shows the selected rows and runs over all combinations of $[\frac{1}{2}g]$ integers not greater than $g$. 
 Each block $[\Rp]_{I}$ contains  full columns of $\Rp$.
 Let  $(\adj \Np)_{I}$  and $(\adj \Dp)_{I}$ be the cofactors of $\Np$ and $\Dp$, 
 respectively,
 corresponding to the submatrix $[\Rp]_{I}$.
Both the numerator and  denominator are computed by means of the Cauchy---Binet formula:
 \begin{gather}\label{NDexpans}
 |\Np| = \sum_{I} \big(\adj \Np\big)_{I}[\Rp]_{I}, \qquad
 |\Dp| = \sum_{I} \big(\adj \Dp\big)_{I}[\Rp]_{I}.
 \end{gather}

Next, we compute
\begin{gather}
\nu_{3g} \bar{\nu}_{3g}  = \frac{|\Np| |\bar{\Np}|}{|\Dp| |\bar{\Dp}|},
\end{gather}
where the bar denotes the complex conjugate, which affects only entries of 
$\Pp (\bm{K})$ and $\Cp_g(\bm{K})$ by replacing values $e_i$  with their complex conjugates $\bar{e}_i$. 
Note that $\Pp (s)$, $\Cp_g(s)$, and $\Rp(s)$ are real-valued.

Note that $|\Np| |\bar{\Np}|$ and $|\Dp| |\bar{\Dp}|$ are
homogeneous in $\wp_{1,1,2i-1}$ of degrees $2[\frac{1}{2}(g+1)]$ and $2[\frac{1}{2}g]$, respectively.
Next, we apply the fundamental cubic relations, see \cite[Theorem 3.2]{belHKF},  
(the argument $s$ of $\wp$-functions is omitted for simplicity):
\begin{multline}\label{WP3IndSqRel}
\wp_{1,1,2i-1} \wp_{1,1,2j-1} = 4(\wp_{1,1}+\lambda_2) \wp_{1,2i-1} \wp_{1,2j-1} \\
+ 4 (\wp_{1,2i-1} \wp_{1,2j+1} + \wp_{1,2i+1} \wp_{1,2j-1}) + 4 \wp_{2i+1,2j+1} \\
- 2 (\wp_{3,2i-1} \wp_{1,2j-1} + \wp_{1,2i-1} \wp_{3,2j-1}) - 2 (\wp_{2i-1,2j+3} + \wp_{2i+3,2j-1}) \\
+ \lambda_4 (\wp_{1,2j-1} \delta_{1,i}  + \wp_{1,2i-1}\delta_{1,j} )
+ 4 \lambda_{2+4i} \delta_{i,j} + 2 (\lambda_{4i} \delta_{i-1,j} + \lambda_{4j} \delta_{i,j-1}).
\end{multline}
This eliminates $\wp_{1,1,2i-1}$ from the expression for $\nu_{3g} \bar{\nu}_{3g}$
and leads to the identity
\begin{equation}\label{NNDDRel}
 |\Np| |\bar{\Np}| + (-1)^g \lambda_{4g} \wp_{1,2g-1}(s) |\Dp| |\bar{\Dp}| = 0,
\end{equation}
which proves \eqref{eta3gGen}. The direct computation of \eqref{NNDDRel} for genera greater than $2$ uses
the identities that define the Kummer variety of $\mathcal{V}$, 
see \cite[\S\,4.1]{belHKF} and \cite[Sect.\,5]{BerWPFF2025}.
Computations are done in genera $1$,
$2$, $3$, see Appendix~\ref{A1}. 
The identity \eqref{NNDDRel} holds in higher genera as well.
\end{proof}

\begin{lemma}\label{P:EtaEtaIm}
Let  $\bm{K}$ be defined by \eqref{KDef}, and let $s \in \Real^g$. Then
\begin{equation}\label{etaSqIm}
\nu_{3g} \bar{\nu}_{3g} = (-1)^{g-1} \lambda_{4g}  \wp_{1,2g-1}(\imath s).
\end{equation}
\end{lemma}
\begin{proof}
The statement of the lemma follows immediately from the identity
\begin{equation}\label{NNDDRelIm}
 |\Np| |\bar{\Np}| + (-1)^g \lambda_{4g} \wp_{1,2g-1}(\imath s) |\Dp| |\bar{\Dp}| = 0,
\end{equation}
obtained in the same way as  \eqref{NNDDRel}, but with the argument $\imath s$ instead of $s$.
Recall that $\wp_{i,j}(\imath s)$ are even functions, and thus real-valued;
 $\wp_{i,j,k}(\imath s)$ are odd functions and are purely imaginary.
All terms in the fundamental cubic relations \eqref{WP3IndSqRel} remain real-valued
after the substitution $s \mapsto \imath s$, as do the identities defining the Kummer variety of $\mathcal{V}$.
The proof of Lemma~\ref{P:EtaEta} applies to Lemma~\ref{P:EtaEtaIm}.
\end{proof}

\begin{theo}
Let  $\bm{K}$ be defined by \eqref{KDef}, and let $s \in \Real^g$. Then
\begin{gather}\label{Abswp12gm1Sq}
\begin{split}
&| \wp_{1,2g-1}\big(s + \bm{K}\big)  |^2  = \lambda_{4g},\\
&| \wp_{1,2g-1}\big(\imath s + \bm{K}\big)  |^2  = \lambda_{4g}.
\end{split}
\end{gather}
\end{theo}
\begin{proof}
Using \eqref{wp12gm1} with $- u_{\text{III}} \,{=}\, s \,{+}\, \bm{K}$,  $u_{\text{I}} \,{=}\, \bm{K}$,
$u_{\text{II}} \,{=}\, s$, and $\lambda_{4g+2} \,{\equiv}\, 0$, and
taking into account that $\wp_{1,2g-1}(\bm{K}) \,{=}\, \prod_{i=1}^g e_{2i}$,  
which implies  $\wp_{1,2g-1}(\bm{K}) \wp_{1,2g-1}(\bar{\bm{K}}) \,{=}\, \lambda_{4g}$, cf.\;\eqref{Lambda4g}, we find
\begin{gather*}
| \wp_{1,2g-1}\big(s + \bm{K}\big)  |^2 = 
\frac{(\nu_{3g} \bar{\nu}_{3g} )^2 }{\lambda_{4g} \wp_{1,2g-1}^2(s)}.
\end{gather*}
Similarly, with $- u_{\text{III}} = \imath s + \bm{K}$,  $u_{\text{I}} = \bm{K}$, and
$u_{\text{II}} = \imath s$, we get
\begin{gather*}
| \wp_{1,2g-1}\big(\imath s + \bm{K}\big)  |^2 = 
\frac{(\nu_{3g} \bar{\nu}_{3g} )^2 }{\lambda_{4g} \wp_{1,2g-1}^2(\imath s)}. 
\end{gather*}
Then, applying  \eqref{etaSqRe} and \eqref{etaSqIm}, respectively, we immediately obtain \eqref{Abswp12gm1Sq}.
\end{proof}

\subsection{Reality conditions summarized}\label{ss:RCsummary}

\begin{theo}\label{T:RealCondSinhG}
A real-valued and smooth finite-gap solution to the sinh-Gordon equation exists
if the corresponding spectral curve  has the following properties:
two  branch points are fixed at $(0,0)$ and infinity, while all other branch points are either all real
and  negative (for the case $\sfb\,{\in}\,\Real$) or all real and positive (for the case $\sfb\,{\in}\,\imath \Real$).
The finite-gap solution is given by \eqref{SinhGSol} with $\bm{C}\,{=}\,\bm{K}$ (defined by \eqref{KDef}).
\end{theo}
 
Based on the above analysis of half-period shifts, we have
\begin{theo}\label{T:RealCondSinG}
Among the half-period shifts $\bm{C}$, the two finite-gap solutions \eqref{SinGSol} 
 to the sine-Gordon equation exist if $\bm{C}\,{=}\,\bm{K}$ (defined by \eqref{KDef}) and
 the corresponding spectral curve  has the following properties:
two  branch points are fixed at $(0,0)$ and infinity, 
with all other branch points occurring in complex conjugate pairs.
\end{theo}
At the same time, 
in the presence of real non-zero branch points,  solutions
with quarter-period shifts $\bm{C}$ exist in the sine-Gordon hierarchy,
as follows from \cite{DN1982, EF1985}. 
Each pair of real branch points gives rise to two solutions with quarter-period shifts,
see \cite[Eq.\,(II.15b)]{EF1985}. Real-valued solutions exist only if 
all real branch points are  negative (the case $\sfb\,{\in}\,\imath \Real$), 
or all are positive (the case $\sfb\,{\in}\,\Real$). 
Note that such solutions are derived from one of the formulas in \eqref{SinGSol}.
On the contrary, in the case of $g$ pairs of complex conjugate branch points,
the both formulas in \eqref{SinGSol} produce solutions 
with the half-period shift $\bm{C}\,{=}\,\bm{K}$.

 \begin{rem}\label{R:Invol}
 By taking into account the involution $\sfb \,{\mapsto}\, \imath \sfb$, $e_i \,{\mapsto}\,{-}e_i$ 
 on the spaces of solutions
 to the sine(sinh)-Gordon equations, we can restrict our consideration to the 
 $\mathfrak{sl}(2,\Real)$ and $\mathfrak{su}(2)$ cases,  
 aligning with the requirements in the literature \cite{ForMcL1982, DN1982, EF1985}.
Under this restriction, the two solutions in the sine-Gordon hierarchy---which arise 
when the spectral curve possesses the maximum number of 
complex conjugate branch point pairs---belong to two different spectral curves connected by the involution.
 \end{rem}

\subsection{Numerical computation of non-linear waves}\label{s:NLW}
Using $\wp$-functions, we achieve an effective computation of 
quasi-periodic finite-gap solutions to integrable systems.

The computations presented below were performed in  Wolfram Mathematica 12.
Integrals between branch points are computed using \texttt{NIntegrate}
with a \texttt{WorkingPrecision} of at least $30$.
The canonical curve  \eqref{V22g1Eq} is defined through its parameters $\lambda_k$,
and the branch points are computed via \texttt{NSolve},
using the same \texttt{WorkingPrecision}.

The $\wp$-functions are computed by the formulas:
\begin{gather}\label{WPdefComp}
\begin{split}
&\wp_{i,j}(u)  = \varkappa_{i,j} - \frac{\partial^2}{\partial u_i \partial u_j } 
\log \theta[K](\omega^{-1} u; \omega^{-1} \omega'),\\
& \wp_{i,j,k}(u)  =- \frac{\partial^3}{\partial u_i \partial u_j \partial u_k} 
\log \theta[K](\omega^{-1} u; \omega^{-1} \omega'),
\end{split}
\end{gather}
where $\varkappa_{i,j}$ are entries of the symmetric matrix $\varkappa = \eta \omega^{-1}$.
The period matrices $\omega$,  $\omega'$,  and $\eta$ are obtained by integrating 
the differentials \eqref{K1DifsGen} and \eqref{K2DifsGen} along the canonical cycles (see
fig.~\ref{cyclesOddCC}). Specifically, the columns of matrices $\omega$,  $\omega'$, and $\eta$
are computed as follows:
\begin{gather*}
\begin{split}
  &\omega_k = 2\int_{e_{2k-1}}^{e_{2k}} \rmd u,\qquad\quad \eta_k = 2\int_{e_{2k-1}}^{e_{2k}} \rmd r,\\
  &\omega'_k = - 2 \sum_{i=1}^k \int_{e_{2i-2}}^{e_{2i-1}} \rmd u  
  = 2 \sum_{i=k}^g  \int_{e_{2i}}^{e_{2i+1}} \rmd u.
  \end{split}
\end{gather*}
To compute the theta function in \eqref{WPdefComp}, we use
a partial sum of the series \eqref{ThetaDef} restricted to
$|n_i|\,{\leqslant}\,6$, where $n_i$ represents a component of the integer vector $n$.

In the following sections, we present an analysis of one- and two-gap Hamiltonian systems. 
The finite-gap solutions, computed via formulas \eqref{sinGSolCanon} and \eqref{sinhGSolCanon}, 
are illustrated graphically.

In genus one, replacing $\rmx$ with $\rmx + \rmt \sfb^{-2} \lambda_{4}^{-1/2}$,
yields elliptic traveling-wave solutions.
In genus two, we present two-periodic solutions, including the two-phase waves \cite{bb1982}.
In genus three,  the argument of $\wp_{1,2g-1}$ is given by  $u\,{=}\,\sfb (\rmx, \mathsf{c}_3, -\sfb^{-2} \lambda_{12}^{-1/2} \rmt)^t\,{+}\,\bm{C}$, 
where the real constant $\mathsf{c}_3$  serves as a parameter that 
modifies the shape of three-periodic wave solutions. In genus $N$, 
there exist $N\,{-}\,2$ such parameters.

\section{Finite-gap solutions in genus $1$}
A Hamiltonian system of the sine(sinh)-Gordon equation in $\M_1$ 
with variables $\gamma_{-1}$, $\beta_{-1}$, $\alpha_0$ 
lives on the submanifold defined by the constraint
$$ \beta_{-1} \gamma_{-1} =  \rmr_{-1}, $$
which we call an orbit.
The dynamics of the system are governed by the Hamiltonian
$$ h_0 = \alpha_0^2 + \sfb (\gamma_{-1} + \beta_{-1}).$$
The  spectral curve of the system is
\begin{equation*}
-w^2 + \sfb^2 z^3 + \rmh_0 z^2 + \rmr_{-1} z = 0,
 \end{equation*}
 which we replace with its  canonical form
\begin{equation}\label{G1HamCanon} 
-y^2 + x^3 + \lambda_2 x^2 + \lambda_4 x = 0,
 \end{equation}
 using the transformation: $x \,{=}\, z$, $y \,{=}\, \sfb w$, 
 $\rmh_0 \,{=}\, \sfb^2 \lambda_2$, $\rmr_{-1} \,{=}\, \sfb^2 \lambda_4$.
 Evidently, all curves which differ in a choice of $\sfb$ are bi-rationally equivalent.
 Let $\Delta \,{=}\, \lambda_2^2 - 4 \lambda_4$ be the discriminant of the curve.

We analyze the behaviour of solutions to the sine(sinh)-Gordon equation
depending on the choice of $\lambda_4$ and $\lambda_2$. 
In fig.\;\ref{f:G1RootDiag}, three regions on the $\lambda_2 \lambda_4$-plane are marked in different colors;
these regions correspond to different regimes of dynamics. 
Region\;I (blue) and Region\;II (orange) are separated by the contour $\Delta \,{=}\, 0$,
where two non-zero branch points coincide at ${-}\tfrac{1}{2}\lambda_2$,
and the spectral curve degenerates to an elementary one.
In Region\;II, we have $\Delta \,{<}\, 0$ and the two non-zero branch points are complex conjugates.
In the remaining part of the $\lambda_2 \lambda_4$-plane, $\Delta \,{>}\, 0$, and all branch points are real.
Another line where the spectral curve degenerates to genus zero is $\lambda_4\,{=}\,0$,
which separates Region\;I  and Region\;III (green).
On this line,  the spectral curve possesses a double branch point at zero, 
and the remaining branch point is ${-}\lambda_2$.
In Region\;I, the two non-zero branch points are located on the same side of the origin;
they are positive in the left part of Region\;I and negative in the right part.
In Region\;III, the two non-zero branch points are located on the different sides of the origin.
\begin{figure}[h]
\caption{Regions on  $\lambda_2 \lambda_4$-plane.} \label{f:G1RootDiag}
\medskip
\parbox[c]{0.4\textwidth}{\includegraphics[width=0.4\textwidth]{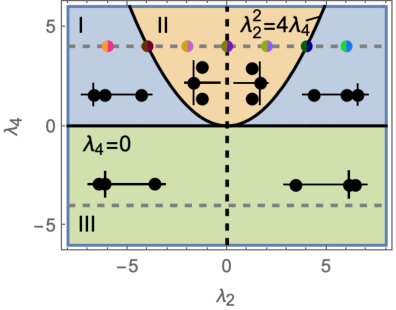}} 
\end{figure}

We analyze orbits by presenting phase portraits of the corresponding
Hamiltonian systems. Below, we focus on the orbit with $\lambda_4 \,{=}\,4$
and $\lambda_2$ running from $-\infty$ to $\infty$, see fig.\;\ref{f:G1RootDiag}. 
The two-colored dots along this line 
indicate points where trajectories on the phase portrait are calculated.
At the values $\lambda_2$ such that $\lambda_2^2\,{=}\,4\lambda_4$,
curve \eqref{G1HamCanon}  degenerates to an elementary one, and  we obtain
separatrix trajectories, which correspond to soliton solutions.
On the orbit fixed by $\lambda_4\,{=}\,0$, curve \eqref{G1HamCanon} 
is elementary for all values of $\lambda_2$, and the solutions are solitons; 
these reduce to a rational soliton at $\lambda_2\,{=}\,0$.
Orbits with \(\lambda _{4}\,{<}\,0\) 
are omitted hereafter due to the absence of smooth solutions 
to the sine(sinh)-Gordon equations in the corresponding Hamiltonian systems.

According to \eqref{WPAlphaGamma}, each trajectory is parametrized as follows
\begin{gather}\label{UniformG1}
\gamma_{-1} = - \sfb \wp_{1,1}(u),\qquad 
\beta_{-1} = \frac{- \sfb \lambda_4}{ \wp_{1,1}(u)}, \qquad
\alpha_0 = - \frac{\sfb}{2}  \frac{\wp_{1,1,1}(u)}{\wp_{1,1}(u)},
\end{gather}
with $u\,{=}\,\sfb \rmx + \bm{C}$,
where the parameter $\rmx\,{\in}\, \Real$ serves as the spatial coordinate, 
the independent variable of the sine(sinh)-Gordon equations.

If curve \eqref{G1HamCanon} degenerates, formulas \eqref{UniformG1} 
simplify as follows.
Along the contour $\Delta\,{=}\,0$ the $\sigma$- and $\wp$-functions degenerate to
\begin{subequations}\label{G1Deg1}
\begin{align}
&\sigma(u) =
 \left\{ \begin{array}{ll}
({-}\tfrac{1}{2}\lambda_2)^{-1/2} \exp\big(\tfrac{1}{4} \lambda_2 u^2\big)
\sinh\big(({-}\tfrac{1}{2}\lambda_2)^{1/2} u \big), & \lambda_2<0,\vphantom{{y}_{\int_A}} \\
(\tfrac{1}{2}\lambda_2)^{-1/2} \exp\big(\tfrac{1}{4} \lambda_2 u^2\big)
\sin \big((\tfrac{1}{2}\lambda_2)^{1/2} u \big), & \lambda_2 > 0;
\end{array} \right. \\
&\wp_{1,1}(u) = 
 \left\{ \begin{array}{ll}
 - \tfrac{1}{2}\lambda_2\tanh^{-2}\big(({-}\tfrac{1}{2}\lambda_2)^{1/2}  u\big),
& \lambda_2<0,\vphantom{{y}_{\int_A}} \\
 \tfrac{1}{2}\lambda_2 \tan^{-2}\big((\tfrac{1}{2}\lambda_2)^{1/2}  u\big),
 & \lambda_2 > 0.
 \end{array} \right.
\end{align}
The periods turn into
\begin{equation}
\begin{array}{lll}
\omega \to \infty, & \omega' = \imath \pi \big({-}\tfrac{1}{2}\lambda_2\big)^{-1/2} &\text{if }\lambda_2<0, \\
\omega = \pi \big(\tfrac{1}{2}\lambda_2\big)^{-1/2}, &\omega' \to \imath \infty &\text{if } \lambda_2>0.
\end{array}
\end{equation}
\end{subequations}
Along the line $\lambda_4 \,{=}\,0$
the $\sigma$- and $\wp$-functions degenerate to 
\begin{subequations}\label{G1Deg2}
\begin{align}
&\sigma(u) =
 \left\{ \begin{array}{ll}
(-\lambda_2)^{-1/2} \sin \big(\sqrt{-\lambda_2}\, u \big), & \lambda_2 < 0,\vphantom{{y}_{\int_A}} \\
\lambda_2^{-1/2} \sinh\big(\sqrt{\lambda_2}\, u \big), & \lambda_2 > 0;
\end{array} \right. \\
&\wp_{1,1}(u) = 
 \left\{ \begin{array}{ll}
 - \lambda_2\sin^{-2}\big(\sqrt{-\lambda_2}\, u\big),
& \lambda_2<0,\vphantom{{y}_{\int_A}} \\
 \lambda_2 \sinh^{-2}\big(\sqrt{\lambda_2}\, u\big),
 & \lambda_2 > 0,
 \end{array} \right.
\end{align}
and the periods turn into
\begin{equation}
\begin{array}{lll}
\omega = \pi \big({-}\lambda_2 \big)^{-1/2}, &\omega' \to \imath \infty,& \text{if }\lambda_2<0, \\
\omega \to  \infty, &\omega' = \pi \lambda_2^{-1/2}  &\text{if }\lambda_2>0.
\end{array}
\end{equation}
\end{subequations}

Each trajectory on a phase portrait corresponds to a real-valued solution. 
A closed trajectory corresponds to a smooth solution, and an open trajectory to 
a solution with singularities.
In the sinh-Gordon hierarchy, real-valued solutions arise when $\bm{C}$ is a half-period,
with one smooth solution existing at $\bm{C}\,{=}\,\bm{K}\,{=}\,\frac{1}{2}\omega \,{+}\,\frac{1}{2}\omega'$.  
In the sine-Gordon hierarchy,
real-valued solutions are obtained at $\bm{C}\,{=}\,\bm{K}$ 
if $\lambda_2^2\,{<}\,4\lambda_4$, or at one of two appropriately chosen quarter periods for
$\bm{C}$  if $\lambda_2^2\,{\geqslant}\,4\lambda_4$ and $\lambda_4\,{\geqslant}\,0$.

Recall that the sinh-Gordon hierarchy arises when 
$\mathfrak{g}\,{=}\, \mathfrak{sl}(2,\Real)$ (the case $\sfb\,{\in}\,\Real$)
or $\mathfrak{g}\,{=}\, \imath\mathfrak{sl}(2,\Real)$ (the case $\sfb\,{\in}\,\imath \Real$), whereas  
the sine-Gordon hierarchy arises when $\mathfrak{g}\,{=}\, \mathfrak{su}(2)$  (the case $\sfb\,{\in}\,\imath \Real$)
or $\mathfrak{g}\,{=}\, \imath \mathfrak{su}(2)$ (the case $\sfb\,{\in}\, \Real$).
As shown in Theorems\;\ref{T:SineGSol} and \ref{T:SinhGSol} and Remark\;\ref{R:bRealSol},
a one-gap solution to the sine-Gordon equation is
\begin{equation}\label{SinGSolN1}
\phi(\x) =  - \imath \log \Big({\pm} \lambda_{4}^{-1/2} \wp_{1,1}(\sfb \rmx + \bm{C}) \Big),
\end{equation}
and a one-gap solution to the sinh-Gordon equation is
\begin{equation}\label{SinhGSolN1}
\phi(\x) =  \log \Big({\pm} \lambda_{4}^{-1/2} \wp_{1,1}(\sfb \rmx + \bm{C}) \Big).
\end{equation}

Smooth solutions to the sine(sinh)-Gordon hierarchies exist if and only if 
$\lambda_4 \,{\geqslant}\,0$. Hamiltonian systems within the sinh-Gordon hierarchy also
possess open trajectories when $\lambda_4 \,{<}\,0$; these trajectories correspond 
to solutions with singularities at finite values of $\rmx$. 
Conversely, Hamiltonian systems within 
the sine-Gordon hierarchy do not exist for $\lambda_4\,{<}\,0$.
Consequently, Region\;III contains no points which produce smooth solutions.

Below, phase portraits of  one-gap Hamiltonian systems associated with 
the sine- and sinh-Gordon hierarchies are analysed, and a choice of $\bm{C}$
for each trajectory is specified.

\subsection{One-gap sinh-Gordon system}
When dealing with a Hamiltonian system,
it is convenient to factor out the constant $\sfb$, which serves as  a scaling factor.
Thus, we replace the variables $\gamma_{-1}$, $\beta_{-1}$, $\alpha_0$ with their 
 normalized counterparts 
 \begin{equation}
 \alpha = \sfb^{-1} \alpha_0,\qquad  \beta = \sfb^{-1} \beta_{-1},\qquad
  \gamma = \sfb^{-1} \gamma_{-1}.
 \end{equation}
 Then, the constraint and the Hamiltonian are expressed in terms of the
 parameters $\lambda_4$, $\lambda_2$ of the canonical curve, namely
 \begin{subequations}\label{slCnstrEq}
 \begin{align}
 &\beta \gamma = \lambda_4, \label{slCnstr} \\
 &\alpha^2 +  \gamma + \beta = \lambda_2. \label{slEq}
 \end{align}
 \end{subequations}
 
In the case of $\mathfrak{sl}(2,\Real)$ algebra, we have $\alpha$, $\beta$, $\gamma\,{\in}\,\Real$.
Equations \eqref{slCnstrEq} define the trajectories of a Hamiltonian system, provided $\lambda_4$ is fixed.
Figure 3a shows the phase portrait of the system at $\lambda_4 \,{=}\,4$
in the  $\beta \gamma \alpha$-space (above)
and on the  $\gamma \alpha$-plane (below). 
\begin{figure}[h]
\caption{Phase portraits of the 1-gap Hamiltonian system 
associated with the sinh-Gordon hierarchy} \label{f:G1PP3Dsl}
\medskip
3a. $\mathfrak{sl}(2,\Real)$, $\lambda_4\,{=}\,4$ $\phantom{mmmmmmm}$   
3b. $\imath \mathfrak{sl}(2,\Real)$, $\lambda_4\,{=}\,4$
$\phantom{mmmmmmmm}$ \\
\parbox[c]{0.3\textwidth}{\includegraphics[width=0.29\textwidth]{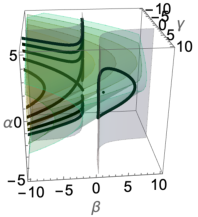}} $\phantom{mmmm}$ 
\parbox[c]{0.45\textwidth}{\includegraphics[width=0.43\textwidth]{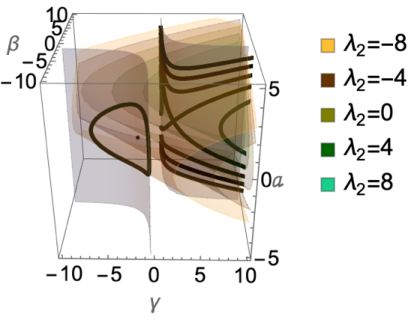}} $\phantom{m}$ \\
$\phantom{mmmmmmmmm}$ \\
\parbox[c]{0.3\textwidth}{\includegraphics[width=0.3\textwidth]{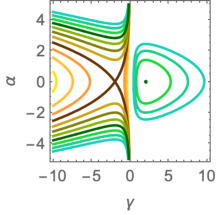}} $\phantom{mmm}$ 
\parbox[c]{0.56\textwidth}{\includegraphics[width=0.56\textwidth]{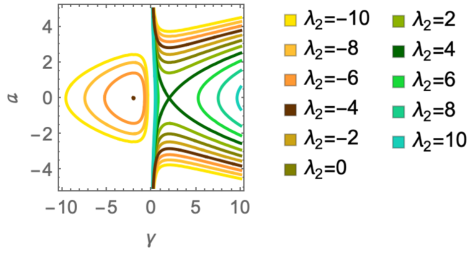}} 
\end{figure}
3D trajectories are plotted using equations \eqref{UniformG1}.

For $\lambda_2\,{<}\,{-}2 \sqrt{\lambda_4}$, there are  two open trajectories
corresponding to $\bm{C} \,{=}\, 0$ and $\tfrac{1}{2} \omega'$.
Indeed, $u\,{=}\,\sfb \rmx$ and $u\,{=}\,\sfb \rmx \,{+}\, \tfrac{1}{2} \omega'$, $\sfb\,{\in}\,\Real$, 
are the two lines  in the Jacobian variety where $\wp_{1,1}$ and $\wp_{1,1,1}$ are real-valued.
At $\lambda_2\,{=}\,{-}2 \sqrt{\lambda_4}$, the vertices of the two trajectories 
touch each other.
For ${-}2 \sqrt{\lambda_4} \,{<}\, \lambda_2\,{<}\, 0$, the relation  
$\tfrac{1}{2} \ImN \omega' \,{=}\,{-} \ImN \omega$ holds, and so there exists
a single trajectory with $\bm{C} \,{=}\, 0\,{\sim}\,\tfrac{1}{2} \omega'$.
This trajectory consists of two branches.
For $0 \,{<}\, \lambda_2 \,{<}\, 2 \sqrt{\lambda_4}$, we have
$\tfrac{1}{2} \ReN \omega \,{=}\, \ReN \omega'$, which implies $0\,{\sim}\,\tfrac{1}{2} \omega$.
In this case, there is also only one trajectory with $\bm{C} \,{=}\, 0$. 
At $\bm{C} \,{=}\, \tfrac{1}{2} \omega'$,
the variables $\gamma$, $\beta$, and $\alpha$ are not real-valued.
At $\lambda_2\,{=}\,2 \sqrt{\lambda_4}$, the level surface \eqref{slEq}
touches the positive wing of the constraint surface \eqref{slCnstr} at a single point,
specifically $(\tfrac{1}{2} \lambda_2,\tfrac{1}{2} \lambda_2,0)$.
This point represents the trajectory obtained with $\bm{C} \,{=}\, \tfrac{1}{2}\omega'$,
alongside the open trajectory with $\bm{C} \,{=}\, 0$.
\begin{table}[h]
\caption{Open and closed trajectories with $\sfb \rmx + \bm{C}$, \\
$\phantom{mmmmma}$$\sfb\in\Real$ in
the case of $\mathfrak{sl}(2,\Real)$ algebra, and\\
$\phantom{mmmmma}$$\sfb \in\imath \Real$ in
the case of $\imath\mathfrak{sl}(2,\Real)$ algebra.} \label{TbSinhG}
\begin{tabular}{llrl}
&Periods&  Trajectories: open & closed \\
$\lambda_2\,{<}\,{-}2 \sqrt{\lambda_4}$
& $\omega \,{\in}\, \Real$, $\omega' \,{\in}\,\imath \Real$ 
& $\sfb\,{\in}\,\Real$, \hfill $\bm{C} = 0$, $\tfrac{1}{2} \omega'$, & none\\
&& $\sfb\,{\in}\,\imath\Real$, \hfill $\bm{C} = 0$,  & $\tfrac{1}{2} \omega$\\
$\lambda_2\,{=}\,{-}2 \sqrt{\lambda_4}$ 
& $\omega \to \infty$, $\omega' \,{\in}\,\imath \Real$ 
& $\sfb\,{\in}\,\Real$, \hfill $\bm{C} = 0$, $\tfrac{1}{2} \omega'$, & none\\
&& $\sfb\,{\in}\,\imath\Real$, \hfill $\bm{C} = 0$,  & $\tfrac{1}{2} \omega$\\
$\left\{ 
\begin{array}{l} {-}2 \sqrt{\lambda_4} \,{<}\, \lambda_2\,{<}\,0\\
\lambda_2\,{=}\, 0 \vphantom{\dfrac{A}{A}} \\ 
0\,{<}\,\lambda_2\,{<}\, 2 \sqrt{\lambda_4}
\end{array} \right. $  
& $\begin{array}{l} \ImN \omega \,{=}\, {-}\frac{1}{2} \ImN \omega',\  \omega' \,{\in}\,\imath \Real\\
\left\{ \begin{array}{l} \ReN \omega \,{=}\, \ReN \omega',\\
 \ImN \omega \,{=}\, {-}\ImN \omega'  
 \end{array} \right.\\
\omega \,{\in}\, \Real,\ \ReN \omega' = \frac{1}{2} \ReN \omega
\end{array} $ 
& $\begin{array}{r} \sfb\,{\in}\,\Real,\quad \bm{C} = 0,\\
\sfb\,{\in}\,\imath\Real,\hfill  \bm{C} = 0,
\end{array} $ & none \\
$ \lambda_2\,{=}\, 2 \sqrt{\lambda_4}\vphantom{A^{\int^A}}$ 
& $\omega \in \Real$, $\omega' \to \imath \infty$  
& $\sfb\,{\in}\,\Real$, \hfill $\bm{C} = 0$,  & $\tfrac{1}{2} \omega'$ \\
&& $\sfb\,{\in}\,\imath\Real$, \hfill $\bm{C} = 0$, $\tfrac{1}{2} \omega$  & none\\
$ \lambda_2\,{>}\, 2 \sqrt{\lambda_4}$ 
& $\omega \in \Real$, $\omega' \,{\in}\, \imath \Real$
& $\sfb\,{\in}\,\Real$, \hfill $\bm{C} = 0$,  & $\tfrac{1}{2} \omega'$ \\
&& $\sfb\,{\in}\,\imath\Real$, \hfill $\bm{C} = 0$, $\tfrac{1}{2} \omega$  & none
\end{tabular}
\end{table}
\noindent
Finally, for $\lambda_2\,{>}\,2 \sqrt{\lambda_4}$, a closed trajectory emerges at the intersection 
with the positive wing of the constraint surface.
This closed trajectory corresponds to $\bm{C} \,{=}\, \tfrac{1}{2}\omega'$,
while an open trajectory is obtained with $\bm{C} \,{=}\, 0$. 
These results are summarized in Table\;\ref{TbSinhG}.
The smooth solutions arising at $\lambda_2\,{>}\,2 \sqrt{\lambda_4}$ 
are illustrated in fig.\;\ref{f:G1slSols}a.
\begin{figure}[h]
\caption{Solution to the sinh-Gordon equation in genus 1} \label{f:G1slSols}
\medskip
4a. $\mathfrak{sl}(2,\Real)$, $\lambda_4\,{=}\,4$ $\phantom{mmmmmmmmm}$   
4b. $\imath \mathfrak{sl}(2,\Real)$, $\lambda_4\,{=}\,4$
$\phantom{mmmmmmmm}$ \\ 
\parbox[c]{0.9\textwidth}{\includegraphics[width=0.9\textwidth]{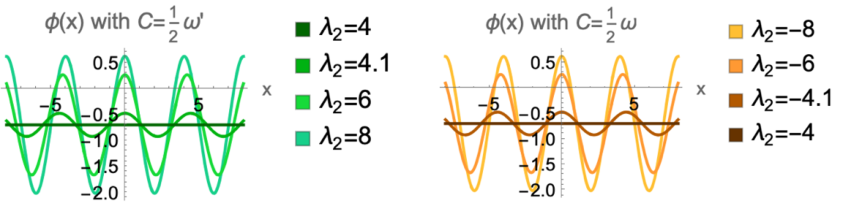}}
\end{figure}

In the case of $\imath \mathfrak{sl}(2,\Real)$ algebra, we define 
$a \,{=}\, {-}\imath \alpha$ (instead of $\alpha \,{\in}\,\imath \Real$), 
allowing \eqref{slEq} to take the form
 \begin{equation}\label{islEq}
- a^2 +  \gamma + \beta = \lambda_2. \tag{\ref{slEq}'}
 \end{equation}
The phase portrait produced by \eqref{slCnstr} and \eqref{islEq}
for $\lambda_4\,{=}\,4$ and various values of $\lambda_2$ are shown on fig\;\ref{f:G1PP3Dsl}b.
See Table\;\ref{TbSinhG} for the analysis of trajectories and fig.\;\ref{f:G1slSols}b
for the smooth solutions that arise when $\lambda_2\,{<}\,{-}2 \sqrt{\lambda_4}$.
Every solution ($\sfb\,{\in}\,\imath \Real$) associated with a spectral curve 
possessing only non-negative branch points $\{e_i\}$ coincides with 
the solution ($\sfb\,{\in}\, \Real$)  associated with the spectral curve
possessing  branch points $\{-e_i\}$, see  Remark\;\ref{R:Invol}.

To summarize,
smooth solutions to the sinh-Gordon equation
exist (i)  for $\lambda_2\,{>}\,2 \sqrt{\lambda_4}$ if $\sfb\,{\in}\,\Real$,
or (ii)  for $\lambda_2\,{<}\,{-}2 \sqrt{\lambda_4}$ if $\sfb\,{\in}\,\imath\Real$.
Case (i) corresponds to  points $(\lambda_2, \lambda_4)$ located in the right part of Region\;I
 in fig.\;\ref{f:G1RootDiag}, where the two non-zero branch points are negative.
Case (ii) corresponds to points $(\lambda_2, \lambda_4)$ in the left part of Region\;I,
where the two non-zero branch points are positive.

\subsection{One-gap sine-Gordon system}
In the cases of $\mathfrak{su}(2)$ and $\imath \mathfrak{su}(2)$ algebras,
instead of $\beta$ and $\gamma$ (where $\beta \,{=}\, \bar{\gamma}$), 
we employ the real-valued coordinates
$\gamma_{\ReN} \,{=}\, \ReN \gamma$ and  $\gamma_{\ImN} \,{=}\, \ImN \gamma$.
If $\mathfrak{g}\,{=}\, \mathfrak{su}(2)$, then $\alpha \,{\in}\,\Real$ 
and \eqref{slCnstrEq} turns into 
\begin{subequations}\label{suCnstrEq}
 \begin{align}
 &\gamma_{\ReN}^2 + \gamma_{\ImN}^2 = \lambda_4, \\
 &\alpha^2 +  2 \gamma_{\ReN}  = \lambda_2. \label{suEq}
 \end{align}
 \end{subequations}
If $\mathfrak{g}\,{=}\, \imath \mathfrak{su}(2)$,  we replace $\alpha$ with $a \,{=}\, {-}\imath \alpha$,
 and \eqref{suEq} takes the form
 \begin{equation}
- a^2 + 2 \gamma_{\ReN}  = \lambda_2. \tag{\ref{suEq}'}
 \end{equation}
 
 \begin{figure}[h]
\caption{Phase portraits of the $1$-gap Hamiltonian system 
at $\lambda_4\,{=}\,4$ associated with the sine-Gordon hierarchy} \label{f:G1PP3Dsu}
\medskip
$\phantom{mm}$
5a. $\mathfrak{su}(2)$ $\phantom{mmmmmmmmmmmm}$   
5b. $\imath \mathfrak{su}(2)$
$\phantom{mmmmmmmmm}$ \\
\parbox[c]{0.35\textwidth}{\includegraphics[width=0.33\textwidth]{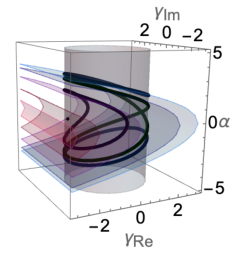}} $\phantom{mmm}$ 
\parbox[c]{0.5\textwidth}{\includegraphics[width=0.44\textwidth]{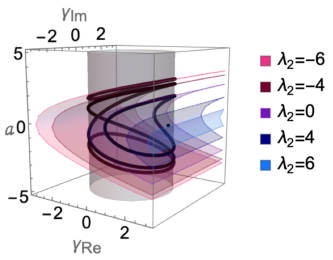}} $\phantom{m}$ \\
\parbox[c]{0.29\textwidth}{\includegraphics[width=0.29\textwidth]{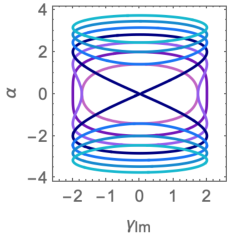}} $\phantom{mmmmm}$ 
\parbox[c]{0.55\textwidth}{\includegraphics[width=0.55\textwidth]{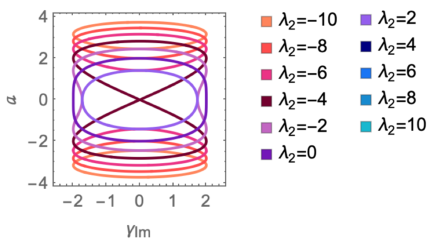}} \\
\parbox[c]{0.31\textwidth}{\includegraphics[width=0.31\textwidth]{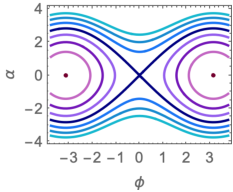}} $\phantom{mmmmm}$ 
\parbox[c]{0.55\textwidth}{\includegraphics[width=0.32\textwidth]{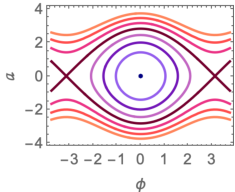}}
\end{figure}
Phase portraits of the Hamiltonian systems for $\lambda_4\,{=}\,4$ associated with 
the $\mathfrak{su}(2)$ and $\imath \mathfrak{su}(2)$ algebras
are shown in fig.\;\ref{f:G1PP3Dsu}. Specifically,
the figures display the phase portraits 
in the  $\gamma_{\ReN} \gamma_{\ImN} \alpha$-space (top),
the $\gamma_{\ImN} \alpha$-plane (middle),
and the $\phi \alpha$-plane (bottom). Here,  $\phi$  denotes the unknown variable
of the sine-Gordon equation, see subsection\;\ref{ss:SineSinhParam} for more details. 
The Hamiltonian expressed in terms of $\phi$ and  $\alpha$ 
(the $\mathfrak{su}(2)$ case) takes the form
$$ h_0 = \sfb^2 \big(\alpha^2 \pm 2 \sqrt{\lambda_4} \cos \phi\big),$$
or, in terms of $\phi$ and  $a$ (the $\imath\mathfrak{su}(2)$ case), the form
$$ h_0 = \sfb^2 \big({-}a^2 \pm 2 \sqrt{\lambda_4} \cos \phi\big).$$

The 3D trajectories are plotted using the parametric equations for 
$\gamma_{\ReN}$, $\gamma_{\ImN}$, and $\alpha$ (or $a$)
obtained from \eqref{UniformG1}. When $\lambda_2 \,{=}\, 2 \sqrt{\lambda_4}$,
we apply the formulas \eqref{G1Deg1}, since $e_1\,{=}\, e_2 \,{=}\,{-}\tfrac{1}{2}\lambda_2$, $\lambda_2 \,{>}\,0$.
Thus, the separatrix trajectory is parameterized by
\begin{gather}\label{NormUniformG1Deg1}
\gamma =  \frac{ -\tfrac{1}{2}\lambda_2}{\tan^{2}\big((\tfrac{1}{2}\lambda_2)^{1/2}  (\sfb \rmx + C)\big)},\quad 
\alpha =  \frac{\pm(2\lambda_2)^{1/2}}{\cos \big((2\lambda_2)^{1/2}  (\sfb \rmx + C ) \big)},\quad
C = \tfrac{1}{4} \pi,\, \tfrac{3}{4} \pi.
\end{gather}

In the case of $\mathfrak{su}(2)$ algebra, trajectories
exist if $\lambda_2\,{\geqslant}\,{-}2 \sqrt{\lambda_4}$, and  all of them are closed.
At $\lambda_2\,{=}{-}\,2 \sqrt{\lambda_4}$, the trajectory collapses into  the single point $(\tfrac{1}{2}\lambda_2,0,0)$.
\begin{table}[h]
\caption{Trajectories (all closed) with $\sfb \rmx + \bm{C}$.} \label{TbSineG}
\begin{tabular}{lll}
&  $\sfb\,{\in}\,\Real$ ($\imath \mathfrak{su}(2)$ case) & $\sfb\,{\in}\,\imath\Real$ ($\mathfrak{su}(2)$ case) \\
$\lambda_2\,{<}\,{-}2 \sqrt{\lambda_4}$
& $\bm{C} = \tfrac{1}{4} \omega'$, $\tfrac{3}{4} \omega'$, &  none \\
$\lambda_2\,{=}\,{-}2 \sqrt{\lambda_4}$ 
& $\bm{C} = \tfrac{1}{4} \omega'$, $\tfrac{3}{4} \omega'$
& $\bm{C} = \infty$\\
$ {-}2 \sqrt{\lambda_4} \,{<}\, \lambda_2\,{\leqslant}\,0$  
& $\bm{C} = \tfrac{1}{2} \omega \sim \tfrac{1}{4} \omega' $ 
& $\bm{C} = \tfrac{1}{2} \omega \sim \tfrac{1}{4} \omega'$ \\
$0\,{\leqslant}\,\lambda_2\,{<}\, 2 \sqrt{\lambda_4}$   &
$\bm{C} = \tfrac{1}{2} \omega' \sim  \tfrac{1}{4} \omega$ 
& $\bm{C} = \tfrac{1}{2} \omega' \sim  \tfrac{1}{4} \omega$ \\
$ \lambda_2\,{=}\, 2 \sqrt{\lambda_4}\vphantom{A^{\int^A}}$ 
&  $\bm{C} = \imath \infty$, &  $\bm{C} = \tfrac{1}{4} \omega$, $\tfrac{3}{4} \omega$ \\
$ \lambda_2\,{>}\, 2 \sqrt{\lambda_4}$ 
& none & $\bm{C} = \tfrac{1}{4} \omega$, $\tfrac{3}{4} \omega$ 
\end{tabular}
\end{table}
For $|\lambda_2|\,{<}\,2 \sqrt{\lambda_4}$, the trajectory consists of a single loop.
This is obtained with $\bm{C}\,{=}\,\tfrac{1}{2}\ReN \omega$ if 
${-}\,2 \sqrt{\lambda_4} \,{<}\,\lambda_2 \,{\leqslant}\,0$, or with $\bm{C}\,{=}\,\tfrac{1}{4} \omega$
(or $\bm{C}\,{=}\,\tfrac{3}{4} \omega$)
if  $0\,{<}\,\lambda_2 \,{\leqslant}\,2 \sqrt{\lambda_4}$.
When $\lambda_2 \,{=}\, 2 \sqrt{\lambda_4}$, we obtain 
the separatrix trajectory, where the single loop splits into two. 
For $\lambda_2\,{\geqslant}\,2 \sqrt{\lambda_4}$, two different trajectories are
obtained with $\bm{C}\,{=}\,\tfrac{1}{4}\omega$
and $\bm{C}\,{=}\,\tfrac{3}{4}\omega$. These results are summarized in Table\;\ref{TbSineG}.
\begin{figure}[h]
\caption{Solution to the sine-Gordon equation in genus 1} \label{f:G1suSols}
\medskip
6a. $\mathfrak{su}(2)$, $\lambda_4\,{=}\,4$: Waves and 
parallelograms of periods $\phantom{mmmmmmmmmmmmm}$    \\ 
\parbox[c]{0.9\textwidth}{\includegraphics[width=0.9\textwidth]{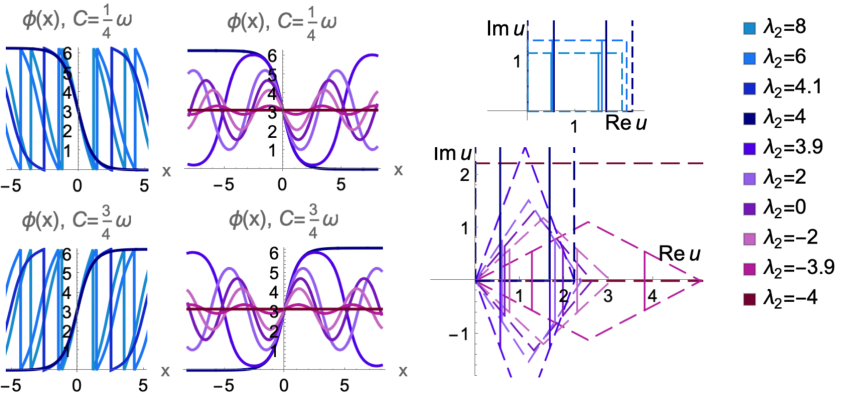}}\\
6b. $\imath \mathfrak{su}(2)$, $\lambda_4\,{=}\,4$: Waves and 
parallelograms of periods $\phantom{mmmmmmmmmmmma}$  \\
\parbox[c]{0.91\textwidth}{\includegraphics[width=0.91\textwidth]{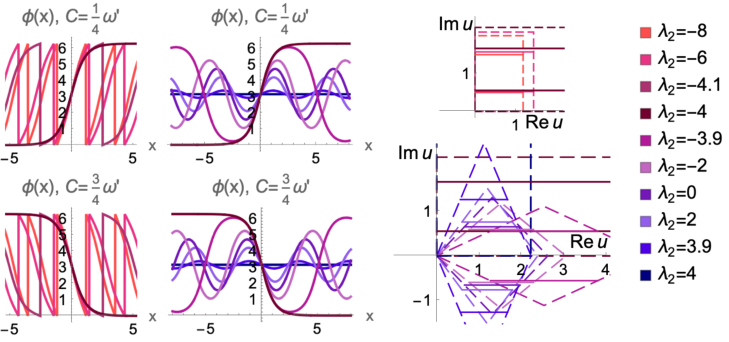}}
\end{figure}

In the case of $\mathfrak{g}\,{=}\, \imath \mathfrak{su}(2)$, trajectories exist
if $\lambda_2\,{\leqslant}\,2 \sqrt{\lambda_4}$.  All such trajectories are closed, see
 Table\;\ref{TbSineG} for more details.

One-gap solutions to the sine-Gordon equation are presented in fig.\;\ref{f:G1suSols} for
the values $\sfb\,{=}\,1/2$ and $\sfb\,{=}\,\imath/2$.
The equivalence  between solutions in the $\mathfrak{su}(2)$  and  $\imath \mathfrak{su}(2)$ cases 
mentioned in Remark\;\ref{R:Invol} is clearly evident.

Solutions corresponding to the separatrix trajectory take the form
\begin{equation}
\begin{split}
\phi(\rmx) &= \imath \log \Big({-}\tanh^{-2}\big(({-}\tfrac{1}{2}\lambda_2)^{1/2} (\sfb \rmx + C)\big)\Big)\\
&= \arctan \Big( 
\frac{\pm 4 \sinh \big(\sqrt{{-}2\lambda_2}\,\sfb \rmx\big)}{3-\cosh \big(2\sqrt{{-}2\lambda_2}\,\sfb \rmx\big)} \Big),
\quad \lambda_2<0,\ \ \sfb\in \Real.
\end{split}
\end{equation}
These are known as kink and anti-kink solutions. 

Therefore, smooth solutions to the sine-Gordon equation
exist if (i) $\lambda_2\,{\leqslant}\,2 \sqrt{\lambda_4}$ when $\sfb\,{\in}\,\Real$,
or (ii) $\lambda_2\,{\geqslant}\,{-}2 \sqrt{\lambda_4}$ when $\sfb\,{\in}\,\imath\Real$.
Case (i) represents  points $(\lambda_2, \lambda_4)$ located within Region\;II and the left part of Region\;I
(with positive non-zero branch points)  in fig.\;\ref{f:G1RootDiag}.
Case (ii) represents  points $(\lambda_2, \lambda_4)$ located within Region\;II and the right part of Region\;I
(with negative non-zero branch points).

\section{Finite-gap solutions in genus $2$}
Two-gap Hamiltonian systems within the sine(sinh)-Gordon hierarchy  arise in $\M_2$ with
dynamic variables $\{\gamma_{-1}$, $\beta_{-1}$, $\alpha_0$, $\gamma_{1}$, $\beta_{1}$, $\alpha_2\}$. 
Each system is restricted to a submanifold defined by two constraints
\begin{gather}\label{CnstrG2}
\begin{split}
&\gamma_{-1} \beta_{-1}  = \rmr_{-1},\\
&\alpha_0^2 + \gamma_{-1} \beta_{1} + \beta_{-1} \gamma_{1} = \rmr_0,
\end{split}
\end{gather}
with fixed values for $\rmr_{-1}$ and $\rmr_{0}$.
We call this submanifold an orbit.
The dynamics of the system are governed by the Hamiltonians
\begin{gather}\label{HamG2}
\begin{split}
&h_1 = 2 \alpha_0 \alpha_2 + \gamma_{1} \beta_{1} + \sfb (\beta_{-1} + \gamma_{-1}),\\
&h_2 = \alpha_2^2 + \sfb (\beta_{1} + \gamma_{1}).
\end{split}
\end{gather}
By solving  constraints \eqref{CnstrG2}  for $\beta_{-1}$ and $\beta_{1}$ 
and eliminating these variables from \eqref{HamG2}, we obtain 
Hamiltonians expressed in terms of 
$\{\gamma_{-1}$, $\alpha_0$, $\gamma_{1}$, $\alpha_2\}$:
\begin{align*}
&h_1 = 2 \alpha_0 \alpha_2  + \sfb \gamma_{-1} + (\rmr_0 - \alpha_0^2) \frac{\gamma_1}{\gamma_{-1}}
 + \rmr_{-1} \Big( \frac{\sfb}{\gamma_{-1}}  -  \frac{\gamma_1^2}{\gamma_{-1}^2} \Big),\\
&h_2 =  \alpha_2^2 + \sfb \gamma_{1} + (\rmr_0 - \alpha_0^2) \frac{\sfb}{\gamma_{-1}}
- \rmr_{-1} \frac{\sfb \gamma_1}{\gamma_{-1}^2} .
\end{align*}

\subsection{Hamiltonian system in separated variables}
After the substitution 
\begin{align*}
&\gamma_{-1} = \sfb z_1 z_2,& &\gamma_1 = - \sfb (z_1+z_2),& \\
&\alpha_0 = - \frac{w_1 z_2^2 - w_2 z_1^2}{z_1 z_2 (z_1 z_2)},&
&\alpha_2 = \frac{w_1 z_2 - w_2 z_1}{z_1 z_2 (z_1 z_2)},&
\end{align*}
induced by \eqref{PointsDef1}, the Hamiltonians take the form
\begin{align*}
&h_1 = - \frac{z_2}{z_1^3(z_1-z_2)}\big(w_1^2 - \sfb^2 z_1^5 - \rmr_0 z_1^2 - \rmr_{-1} z_1 \big) \\
&\qquad + \frac{z_1}{z_2^3(z_1-z_2)}\big(w_2^2 - \sfb^2 z_2^5 - \rmr_0 z_2^2 - \rmr_{-1} z_2 \big),\\
&h_2 =  \frac{1}{z_1^3(z_1-z_2)}\big(w_1^2 - \sfb^2 z_1^5 - \rmr_0 z_1^2 - \rmr_{-1} z_1 \big) \\
&\qquad - \frac{1}{z_2^3(z_1-z_2)}\big(w_2^2 - \sfb^2 z_2^5 - \rmr_0 z_2^2 - \rmr_{-1} z_2 \big).
\end{align*}
This implies that each pair $(z_i,w_i)$ is  a point of the spectral curve
\begin{gather*}
- w_i^2 + \sfb^2 z_i^5 + \rmh_2 z_i^4 + \rmh_1 z_i^3 + \rmr_0 z_i^2 + \rmr_{-1} z_i = 0,\qquad i=1,\,2.
\end{gather*}
Each pair consists of a canonical coordinate $z_i$ and the corresponding momentum $w_i$.
The evolution of such a pair represents motion on the spectral curve.

\subsection{Hamiltonian system in normalized coordinates}
It is convenient to replace the spectral curve
\begin{gather*}
- w^2 + \sfb^2 z^5 + \rmh_2 z^4 + \rmh_1 z^3 + \rmr_0 z^2 + \rmr_{-1} z = 0
\end{gather*}
with its canonical form
\begin{gather}\label{CcanonG2}
- y^2 + x^5 + \lambda_2 x^4 + \lambda_4 x^3 + \lambda_6 x^2 + \lambda_8 x = 0,
\end{gather}
using  substitution \eqref{SectrCToCanonC}. 
We then transform the dynamic variables into their normalized counterparts: 
$\tilde{\gamma}_{2i-1}\,{=}\, \sfb^{-1} \gamma_{2i-1}$, 
$\tilde{\beta}_{2i-1}\,{=}\, \sfb^{-1} \beta_{2i-1}$, $\tilde{\alpha}_{2i}\,{=}\, \sfb^{-1} \alpha_{2i}$.
Then,  constraints \eqref{CnstrG2} and  Hamiltonians \eqref{HamG2} 
take the form
\begin{align}\label{CnstrNormG2}
\begin{split}
&\tilde{\gamma}_{-1} \tilde{\beta}_{-1}  =  \lambda_8,\\
&\tilde{\alpha}_0^2 + \tilde{\gamma}_{-1} \tilde{\beta}_{1} + \tilde{\beta}_{-1} \tilde{\gamma}_{1} = \lambda_6,
\end{split}\\
\begin{split}\label{HamNormG2}
&\lambda_4 = 2 \tilde{\alpha}_0 \tilde{\alpha}_2 + \tilde{\gamma}_{1} \tilde{\beta}_{1} 
+ \tilde{\beta}_{-1} + \tilde{\gamma}_{-1},\\
&\lambda_2 = \tilde{\alpha}_2^2 + \tilde{\beta}_{1} + \tilde{\gamma}_{1},
\end{split}
\end{align}
and, therefore, depend only on  the parameters $\lambda$ of the curve \eqref{CcanonG2}.
An orbit supporting  a two-gap Hamiltonian system  is defined by fixing the parameters 
$ \lambda_8$ and~$\lambda_6$. 

As an example, we consider the orbit with $\lambda_8\,{=}\,16$ and
$\lambda_6\,{=}\,{-}240$. See fig.\;\ref{f:G2RootDiag} 
for regions on the $\lambda_2 \lambda_4$-plane; the regions 
represent different combinations of real and complex conjugate branch points, and thus
different dynamical regimes.
\begin{figure}[h]
\caption{Regions on  $\lambda_2 \lambda_4$-plane with $\lambda_8\,{=}\,16$,
$\lambda_6\,{=}\,{-}240$.} \label{f:G2RootDiag}
\medskip
\parbox[c]{0.38\textwidth}{\includegraphics[width=0.38\textwidth]{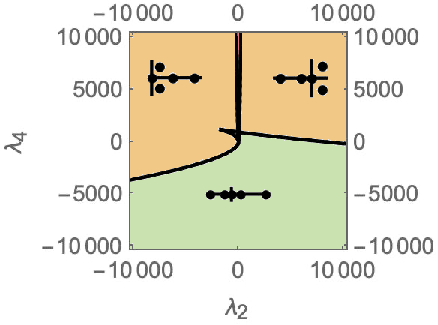}} 
\parbox[c]{0.6\textwidth}{\includegraphics[width=0.6\textwidth]{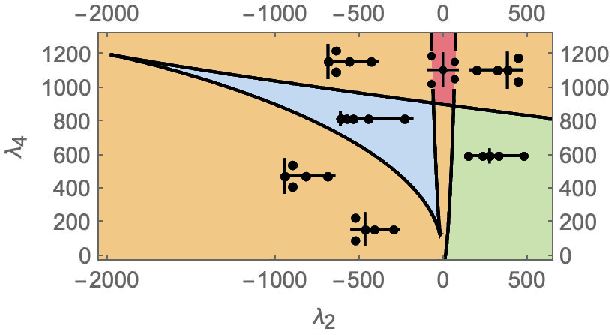}} 
\end{figure}

The discriminant of the curve \eqref{CcanonG2} is
\begin{equation}
\Delta = \lambda_8^2 \big( (72 \lambda_8 \lambda_4 - 27 \lambda_8 \lambda_2^2
- 27\lambda_6^2 + 9 \lambda_6 \lambda_4 \lambda_2 - 2 \lambda_4^3)^2 
- 4 (12 \lambda_8 - 3 \lambda_6 \lambda_2 + \lambda_4^2)^3 \big).
\end{equation}
For $\Delta \,{\neq}\,0$, the curve has genus $2$.
When $\Delta \,{=}\,0$, the genus decreases, and the curve degenerates 
into an elliptic or rational curve. 
In particular, if $\lambda_8\,{=}\,0$, the curve has genus at most $1$.
The black contours in fig.\;\ref{f:G2RootDiag} represent solutions of $\Delta \,{=}\,0$
for a particular choice of $\lambda_8$ and $\lambda_6$.

\subsection{Two-gap sinh-Gordon system}\label{ss:SinhGSols}
Taking into account the reality conditions, we know, that smooth solutions to the sinh-Gordon equation
exist if all branch points are either real and negative or all real and positive. 
Thus, in the diagram in fig.\;\ref{f:G2RootDiag} 
we choose pairs $(\lambda_2,\lambda_4)$ from the blue region.
Due to the involution mentioned in Remark\;\ref{R:Invol},
a complementary region representing curves with all negative branch points exists
in the orbit fixed by $\lambda_8\,{=}\,16$, $\lambda_6\,{=}\, 240$. 

Equations \eqref{CnstrNormG2}, \eqref{HamNormG2} accommodate the case $\mathfrak{sl}(2,\Real)$,
and the normalized coordinates are parametrized as 
\begin{gather}\label{ParamNormG2}
\begin{aligned}
&\tilde{\gamma}_{-1} = - \wp_{1,3}(u),& 
&\tilde{\beta}_{-1} = \frac{- \lambda_8}{\wp_{1,3}(u)},& 
&\tilde{\alpha}_0 = - \frac{1}{2}  \frac{\wp_{1,3,3}(u)}{\wp_{1,3}(u)},& \\
&\tilde{\gamma}_{1} = - \wp_{1,1}(u),&
&\tilde{\beta}_{1} = \frac{\wp_{3,3}(u)}{\wp_{1,3}(u)},& 
&\tilde{\alpha}_2 = - \frac{1}{2}  \frac{\wp_{1,1,3}(u)}{\wp_{1,3}(u)},&
\end{aligned}
\end{gather}
where $u\,{=}\, \sfb \big(\rmx,{\pm} \sfb^{-2} \lambda_{8}^{-1/2}\rmt \big)\,{+}\,\bm{C}$
and $\bm{C}$ represents half-period shifts. There are four half-period shifts $\bm{C}$
with distinct values of $\ImN \bm{C}$; they produce distinct 
trajectories of a 2-gap Hamiltonian system. The variables $\tilde{\gamma}_{-1}$
$\tilde{\gamma}_{1}$, $\tilde{\beta}_{-1}$, $\tilde{\beta}_{1}$, $\tilde{\alpha}_{0}$,
and $\tilde{\alpha}_{2}$ are real-valued along these trajectories.
All trajectories are open, except for
the trajectory defined by $\bm{C} \,{=}\, \bm{K} \,{\equiv}\, \tfrac{1}{2}\omega_2 + \tfrac{1}{2}\omega'_1
+ \tfrac{1}{2}\omega'_2$, which is closed if and only if
all non-zero branch points are negative.
There is no pair $(\lambda_2,\lambda_4)$ that represents a curve with this property on
the chosen orbit ($\lambda_8\,{=}\,16$, $\lambda_6\,{=}\,{-}240$).

In the case of $\imath \mathfrak{sl}(2,\Real)$ algebra,  
$\tilde{\alpha}_{0}$, $\tilde{\alpha}_{2} \,{\in}\,\imath\Real$
and $\tilde{\gamma}_{-1}$
$\tilde{\gamma}_{1}$, $\tilde{\beta}_{-1}$, $\tilde{\beta}_{1}\,{\in}\,\Real$.
In the phase space of a Hamiltonian system we replace $\tilde{\alpha}_{2i}$ with 
real-valued coordinates $\tilde{a}_{2i}\,{=}\,{-}\imath\tilde{\alpha}_{2i}$
and use the same formulas \eqref{ParamNormG2} to parametrize trajectories.
Four half-period shifts $\bm{C}$
with distinct values of $\ReN \bm{C}$ produce distinct trajectories. 
According to Theorem\;\ref{T:RealCondSinhG},
the trajectory with $\bm{C} \,{=}\, \bm{K}$ is closed if and only if 
all non-zero branch points are real and positive.
For the orbit fixed by $\lambda_8\,{=}\,16$ and $\lambda_6\,{=}\,{-}240$,
points $(\lambda_2,\lambda_4)$ in the blue region of the diagram in fig.\;\ref{f:G2RootDiag}
represent such curves.
We compute this smooth finite-gap solution to the sinh-Gordon equation using the formula
\begin{equation}\label{sinhGSolG2} 
\phi(\rmx,\rmt) =  \log \Big({-}\lambda_{8}^{-1/2}
\wp_{1,2N-1} \Big(\sfb \big(\rmx,
{-}\sfb^{-2} \lambda_{8}^{-1/2}\rmt  \big)^t  + \bm{K}\Big) \Big).
\end{equation}
The result is illustrated in fig.\;\ref{f:G2slSols}.
\begin{figure}[h]
\caption{Solution to the sinh-Gordon equation in genus 2: \\
$\lambda_8\,{=}\,16$, $\lambda_6\,{=}\,{-}240$, $\lambda_4\,{=}\,640$, $\lambda_2\,{=}\,{-}320$, $\sfb\,{=}\,\imath/2$;\\
$e_1 = 0$,   $e_2 \approx 0.0852$, $e_3 \approx 0.3830$, $e_4 \approx 1.5421$,  $e_5 \approx 317.99$.} 
\label{f:G2slSols}
\parbox[c]{0.24\textwidth}{\includegraphics[width=0.24\textwidth]{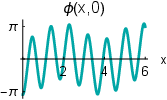}}
\parbox[c]{0.24\textwidth}{\includegraphics[width=0.24\textwidth]{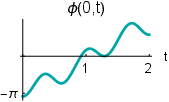}}
\parbox[c]{0.35\textwidth}{\includegraphics[width=0.4\textwidth]{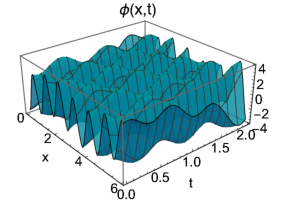}}
\end{figure}

\begin{figure}[h]
\caption{Solutions to the sine-Gordon equation in genus 2:\\
two complex conjugate pairs of branch points}\label{f:G2suSolsCC}
\medskip
\begin{tabular}{lcc}
$\ \ \lambda_8=16$, $\lambda_6=-240$ &$\sfb\,{=}\, 1/2$ & $\sfb\,{=}\,\imath/2$ \\
\parbox[c]{0.3\textwidth}{$\begin{array}{l}
\lambda_4 = 1180,\ 
\lambda_2 = 64 \\
e_1 \approx -32.101-12.7467\imath\\
e_2 \approx -32.101+12.7467\imath\\
e_3 = 0\\
e_4 \approx 0.1010 - 0.0568 \imath \\
e_5 \approx 0.1010 + 0.0568 \imath
\end{array}$} &
\parbox[c]{0.33\textwidth}{\includegraphics[width=0.33\textwidth]{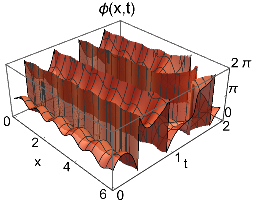}} &
\parbox[c]{0.33\textwidth}{\includegraphics[width=0.33\textwidth]{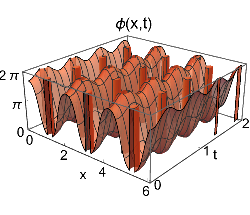}}\\
\hline
\parbox[c]{0.3\textwidth}{$\begin{array}{l}
\lambda_4 = 1180,\ 
\lambda_2 = 0 \\
e_1 \approx -0.1017 - 34.3514 \imath \\
e_2 \approx -0.1017 + 34.3514 \imath \\
e_3 = 0\\
e_4 \approx 0.1017 - 0.0567 \imath\\
e_5 \approx 0.1017 + 0.0567 \imath
\end{array}$} &
\parbox[c]{0.33\textwidth}{\includegraphics[width=0.33\textwidth]{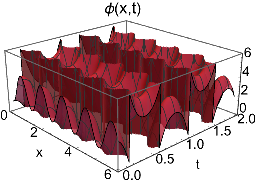}} &
\parbox[c]{0.33\textwidth}{\includegraphics[width=0.33\textwidth]{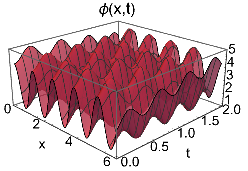}} \\
\hline
\parbox[c]{0.3\textwidth}{$\begin{array}{l}
\lambda_4 = 1180,\ 
\lambda_2 = -64 \\
e_1 = 0 \\
e_2 \approx 0.1025 - 0.0567 \imath\\
e_3 \approx 0.1025 + 0.0567 \imath\\
e_4 \approx 31.8975 - 12.225 \imath\\
e_5  \approx 31.8975 + 12.225 \imath 
\end{array}$} &
\parbox[c]{0.33\textwidth}{\includegraphics[width=0.33\textwidth]{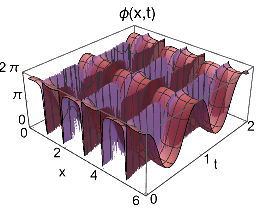}} &
\parbox[c]{0.33\textwidth}{\includegraphics[width=0.32\textwidth]{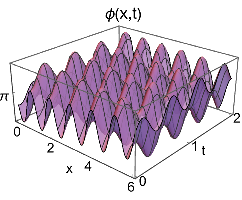}}
\end{tabular}
\end{figure}

\subsection{Two-gap sine-Gordon system}
The most interesting and physically meaningful solutions to the sine-Gordon equation
arise if the corresponding spectral curve possesses the maximal number of complex conjugate pairs of 
branch points. On the orbit with $\lambda_8\,{=}\,16$ and $\lambda_6\,{=}\,{-}240$,
such curves have values of $(\lambda_2, \lambda_4)$ within the red region.
Two solutions to the sine-Gordon equation are be obtained by
\begin{equation}\label{sinGSolG2} 
\phi(\rmx,\rmt) =  -\imath \log \Big({\pm} \lambda_{8}^{-1/2}
\wp_{1,2N-1} \Big(\sfb \big(\rmx,
{\pm}\sfb^{-2} \lambda_{8}^{-1/2}\rmt  \big)^t  + \bm{C}\Big) \Big),
\end{equation}
where $\bm{C}\,{=}\, \bm{K}$.
We use `$+$' for $\sfb\,{\in}\,\Real$ (the case of $\imath \mathfrak{su}(2)$ algebra), 
and `$-$' for $\sfb\,{\in}\,\imath \Real$ (the case of $\mathfrak{su}(2)$ algebra).
Examples of these solutions are illustrated in fig.\;\ref{f:G2suSolsCC}.

Solutions to the sine-Gordon equation also exist if the corresponding curve
possesses more than one real branch point.
For the orbit fixed by $\lambda_8\,{=}\,16$ and $\lambda_6\,{=}\,{-}240$, 
we consider curves with $(\lambda_2,\lambda_4)$ from the orange region
(see fig.\;\ref{f:G2suSolsCR}) and  the blue region (see fig.\;\ref{f:G2suSolsRR}).
In  formula  \eqref{sinGSolG2}, we take
$\sfb\,{\in}\,\Real$ if non-zero branch points are positive, 
and $\sfb\,{\in}\,\imath\Real$ if they are negative. 

\begin{figure}[h]
\caption{Solutions to the sine-Gordon equation in genus 2: \\
a complex conjugate pair and three real branch points}\label{f:G2suSolsCR}
\medskip
\begin{tabular}{llcc}
\parbox[c]{0.17\textwidth}{$\begin{array}{l}
\lambda_4 = 280\\
\lambda_2 = -1480 \\
\sfb\,{=}\,1/2 \\
\bm{C} = \tfrac{1}{4} \omega'_1
\end{array}$} &
\parbox[c]{0.27\textwidth}{$\begin{array}{l}
e_1  =0 \\
e_2 \approx 0.0594 - 0.3877 \imath \\
e_3 \approx 0.0594 + 0.3877 \imath\\
e_4 \approx 0.07029  \\
e_5\approx  1479.81
\end{array}$} &
\parbox[c]{0.24\textwidth}{\includegraphics[width=0.24\textwidth]{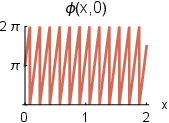}} &
\parbox[c]{0.24\textwidth}{\includegraphics[width=0.24\textwidth]{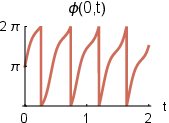}} \\
\hline
\parbox[c]{0.13\textwidth}{$\begin{array}{l}
\lambda_4 = 280\\
\lambda_2 = -320 \\
\sfb\,{=}\,1/2 \\
\bm{C} = \tfrac{1}{4} \omega'_1
\end{array}$} &
\parbox[c]{0.25\textwidth}{$\begin{array}{l}
e_1  =0 \\
e_2 \approx 0.0723 \\
e_3 \approx 0.4014 - 0.7299 \imath\\
e_4 \approx 0.4014 + 0.7299 \imath\\
e_5 \approx 319.125
\end{array}$} &
\parbox[c]{0.24\textwidth}{\includegraphics[width=0.24\textwidth]{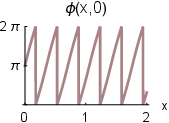}} &
\parbox[c]{0.24\textwidth}{\includegraphics[width=0.24\textwidth]{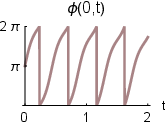}}\\
\hline
\parbox[c]{0.13\textwidth}{$\begin{array}{l}
\lambda_4 = 1180\\
\lambda_2 = -320 \\
\sfb\,{=}\,1/2 \\
\bm{C} = \tfrac{1}{4} \omega'_1
\end{array}$} &
\parbox[c]{0.25\textwidth}{$\begin{array}{l}
e_1  =0 \\
e_2 \approx 0.1058 - 0.0565 \imath \\
e_3 \approx 0.1058 + 0.0565 \imath\\
e_4 \approx 3.5169  \\
e_5\approx 316.271
\end{array}$} &
\parbox[c]{0.24\textwidth}{\includegraphics[width=0.24\textwidth]{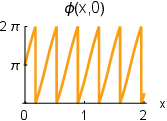}} &
\parbox[c]{0.24\textwidth}{\includegraphics[width=0.24\textwidth]{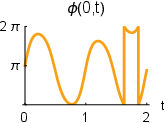}} \\
\hline
\parbox[c]{0.13\textwidth}{$\begin{array}{l}
\lambda_4 = 1180\\
\lambda_2 = 320 \\
\sfb\,{=}\,\imath/2\\
\bm{C} \,{=}\, \tfrac{1}{4} \omega_1 \,{+}\, \tfrac{1}{2} \omega_2
\end{array}$} &
\parbox[c]{0.25\textwidth}{$\begin{array}{l}
e_1 \approx -316.267 \\
e_2 \approx -3.9299 \\
e_3 = 0\\
e_4 \approx 0.0982-0.0568 \imath \\
e_5 \approx 0.0982+0.0568 \imath
\end{array}$} &
\parbox[c]{0.24\textwidth}{\includegraphics[width=0.24\textwidth]{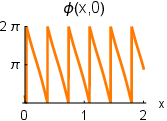}} &
\parbox[c]{0.24\textwidth}{\includegraphics[width=0.24\textwidth]{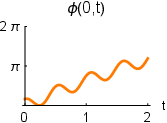}}
\end{tabular}
\end{figure}

In the case of three real branch points, there exist two solutions corresponding to 
quarter-period shifts $\bm{C}$ with multiples of $\tfrac{1}{4}$ and $\frac{3}{4}$.
One solution is increasing, while the other is decreasing
and can be obtained from the first via the transformation $\phi\,{\mapsto}\,{-}\phi$. 
Note that $\phi$ represents an angle and therefore
has the range $[0,2\pi)$. One of these solutions 
is shown in fig.\;\ref{f:G2suSolsCR} for each case.

In the case of all five real branch points, there exist four solutions corresponding to
four quarter-period shifts $\bm{C}$ (see fig.\;\ref{f:G2suSolsRR}).
\begin{figure}[h]
\caption{Solutions to the sine-Gordon equation in genus 2: \\
$\lambda_8\,{=}\,16$, $\lambda_6\,{=}\,{-}240$, $\lambda_4\,{=}\,640$, $\lambda_2\,{=}\,{-}320$,
$\sfb\,{=}\,1/2$;\\
$e_1 = 0$,   $e_2 \approx 0.0852$, $e_3 \approx 0.3830$, $e_4 \approx 1.5421$,  $e_5 \approx 317.99$.}\label{f:G2suSolsRR}
\medskip
6a. $\bm{C}=\tfrac{1}{4}\omega'_1$ $\phantom{mmmmmmmmmmmmmt}$ 
6b. $\bm{C}=\tfrac{3}{4}\omega'_1$  $\phantom{mmmmmmmmmaa}$  \\ 
\parbox[c]{0.24\textwidth}{\includegraphics[width=0.24\textwidth]{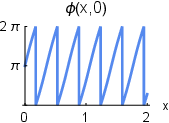}}
\parbox[c]{0.24\textwidth}{\includegraphics[width=0.24\textwidth]{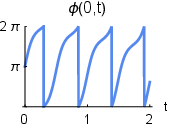}}
\parbox[c]{0.24\textwidth}{\includegraphics[width=0.24\textwidth]{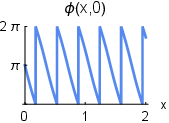}}
\parbox[c]{0.24\textwidth}{\includegraphics[width=0.24\textwidth]{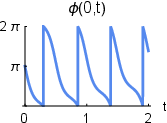}}\\
6c. $\bm{C}=\tfrac{1}{4}\omega'_1+\tfrac{1}{2}\omega'_2\vphantom{\displaystyle A^{\int^A}}$ $\phantom{mmmmmmmmmm}$ 
6d. $\bm{C}=\tfrac{3}{4}\omega'_1+\tfrac{1}{2}\omega'_2$  $\phantom{mmmmmma}$  \\ 
\parbox[c]{0.24\textwidth}{\includegraphics[width=0.24\textwidth]{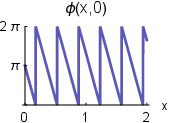}}
\parbox[c]{0.24\textwidth}{\includegraphics[width=0.24\textwidth]{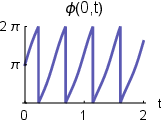}}
\parbox[c]{0.24\textwidth}{\includegraphics[width=0.24\textwidth]{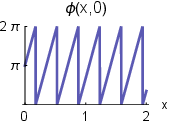}}
\parbox[c]{0.24\textwidth}{\includegraphics[width=0.24\textwidth]{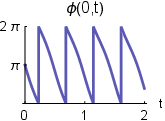}}
\end{figure}

\appendix

\section{Hyperelliptic addition law}\label{A:AddLaw}
The addition laws on hyperelliptic curves from \cite{bl2005} is employed.

The addition operation on the Jacobian variety of a plane algebraic curve is defined by the polynomial 
 function $\mathcal{R}_{3g}$ of weight $3g$. On a hyperelliptic curve $\mathcal{V}$ of genus $g$, 
 it has the following form:
\begin{gather}\label{NuIntro}
\begin{split}
&\mathcal{R}_{3g}(x,y) \equiv y \nu_y(x) + \nu_x(x), \\
&\nu_y(x) \equiv \sum_{i=1}^{\kFr} \nu_{2i-2} x^{\kFr-i},\quad 
\nu_x(x) \equiv \sum_{i=1}^{3\kFr-1} \nu_{2i-1} x^{3\kFr - 1 - i},\quad \text{if } g = 2\kFr -1,\ \ \text{or} \\
&\nu_y(x) \equiv \sum_{i=1}^{\kFr}  \nu_{2i-1} x^{\kFr-i},\quad
\nu_x(x) \equiv \sum_{i=1}^{3\kFr+1} \nu_{2i-2} x^{3\kFr+1-i},\quad \text{if } g = 2\kFr,
\end{split}
\end{gather}
where $\nu_0\,{=}\,1$ and  $\kFr \in \Natural$.
The function $\mathcal{R}_{3g}$ has $2g$ arbitrary coefficients; therefore, 
it is uniquely determined by specifying $2g$ points in its divisor of zeros. 
Let $2g$ points on $\mathcal{V}$, containing no pairs in involution, 
form two non-special divisors $D_\text{I}$ and $D_\text{II}$.  
Let $D_\text{III}$ be such a divisor that $(\mathcal{R}_{3g})_0 \,{=}\, D_\text{I} \,{+}\, D_\text{II} \,{+}\, D_\text{III}$.
This implies  $u_\text{I}\,{+}\,u_\text{II}\,{+}\,u_\text{III}\,{=}\,0$, where
$\mathcal{A}(D_\text{I}) \,{=}\, u_\text{I}$, $\mathcal{A}(D_\text{II}) \,{=}\, u_\text{II}$,
and $\mathcal{A}(D_\text{III}) \,{=}\, u_\text{III}$.

Equations  \eqref{EnC22g1}  enable us to reduce $\mathcal{R}_{3g}$ to a degree $g-1$ polynomial in $x$ 
with coefficients expressed in terms of $\nu_k$ and basis $\wp$-functions, namely 
$\wp_{1,2i-1}$ and $\wp_{1,1,2i-1}$, $i\,{=}\,1$, \ldots, $g$.
Since the pre-images of $u_\text{I}$, $u_\text{II}$, $u_\text{III}$  form 
the divisor of zeros of $\mathcal{R}_{3g}$, we have
\begin{gather}\label{AddLawEqs}
\begin{pmatrix}
1_g & \Qp(u_\text{I}) \\
1_g & \Qp(u_\text{II}) \\
1_g & \Qp(u_\text{III}) 
\end{pmatrix} 
\begin{pmatrix}  \bar{\nu} \\ \nu \end{pmatrix} = -
\begin{pmatrix}  \mathbf{q}(u_\text{I}) \\ \mathbf{q}(u_\text{II}) \\ \mathbf{q}(u_\text{III}) \end{pmatrix},
\end{gather}
where $\bar{\nu} \,{=}\, (\nu_{g+2}, \nu_{g+4},  \dots, \nu_{3g} )^t$, 
 $\nu \,{=}\, (\nu_g, \dots, \nu_2, \nu_1)^t$, and $\Qp$, $\mathbf{q}$ are defined as in \cite[Sect.\,6.2]{BerKdV2024}.

Let $u_\text{I}$ and $u_\text{II}$ be given; then $u_\text{III} \,{=}\, {-} u_\text{I} \,{-}\, u_\text{II}$. 
Thus, the equations for $u_\text{I}$ and $u_\text{II}$ in \eqref{AddLawEqs} are used to find $\nu$ and 
$\bar{\nu}$ by means of Cramer's rule. Values  $\wp_{1,2i-1}(u_\text{III})$  are obtained from the
equality
$$\nu_y^2 f(x,- \nu_x / \nu_y) = 
\mathcal{R}_{2g}(x,y;u_\text{I}) \mathcal{R}_{2g}(x,y;u_\text{II}) \mathcal{R}_{2g}(x,y;u_\text{III}).$$
Explicit expressions for $\wp_{1,2i-1}(u_\text{III})$ 
can be found in \cite[Sect.\,6.2]{BerKdV2024}.
Next, values of $\wp_{1,1,2i-1}(u_\text{III})$ are computed from the equations for $u_\text{III}$ in \eqref{AddLawEqs}.
A more general approach to the addition laws can be found in \cite{bl2005},
with hyperelliptic addition laws as an illustration.

For our purposes, we need an expression for the $\wp_{1,2g-1}$-function of the form
\begin{equation}\label{wp12gm1}
\wp_{1,2g-1}(u_\text{III}) = (-1)^{g-1} \frac{\nu_{3g}^2 - \lambda_{4g+2}  \nu_{g-1}^2 }
{\wp_{1,2g-1}(u_\text{I})\wp_{1,2g-1}(u_\text{II})}.
\end{equation}
Note that $\lambda_{4g+2} \,{\equiv}\, 0$ in the  sine(sinh)-Gordon hierarchy.

\section{Proof of Lemma~\ref{P:EtaEta} }\label{A1}  

\textbf{Genus 1.}
Let $\mathcal{V}$ be an elliptic curve of the form \eqref{V22g1Eq} with branch points located at
 $e_2$, $\bar{e}_2$, and $0$. Let $\wp$ be the doubly periodic function associated with this curve;
 therefore,  $\wp_{1,1}(u) \equiv \wp(u)$, $\wp_{1,1,1}(u) \equiv \wp'(u)$ satisfy the cubic relation
 $$ \forall u\in \Complex\backslash\{0\} \qquad  
\tfrac{1}{4} \wp_{1,1,1}(u)^2 = \wp_{1,1}(u)^3 + \lambda_2 \wp_{1,1}(u)^2 + \lambda_4 \wp_{1,1}(u).$$
Let $\nu_3$ be defined by \eqref{eta3gGen} with
\begin{gather}\label{NDG1}
|\Dp| =  \wp_{1,1}(s) - e_2,\qquad
|\Np| = -\tfrac{1}{2} e_2 \wp_{1,1,1}(s), \quad  s\in\Real.
\end{gather}
Substituting \eqref{NDG1} into \eqref{NNDDRel} and considering that 
$\lambda_4 = e_2 \bar{e}_2$ and $\lambda_2 = - (e_2+\bar{e}_2)$,
we obtain
\begin{multline}\label{nuRelG1}
\tfrac{1}{4} e_2 \bar{e}_2 \wp_{1,1,1}(s)^2
- \lambda_{4} \wp_{1,1}(s) (\wp_{1,1}(s) - e_2)(\wp_{1,1}(s) - \bar{e}_2) \\
= (e_2 \bar{e}_2 - \lambda_{4}) \wp_{1,1}(s)^3 
+ \big(e_2 \bar{e}_2 \lambda_2 +  \lambda_{4}(e_2+\bar{e}_2) \big) \wp_{1,1}(s)^2 \equiv 0,
\end{multline}
by applying the cubic relation.

\bigskip
\textbf{Genus 2.}
Let a genus two curve $\mathcal{V}$, defined as in \eqref{V22g1Eq}, have 
 branch points  located at $\{e_2$, $\bar{e}_2$, $e_4$, $\bar{e}_4$, $0\}$.
The functions  $\wp_{1,2i-1}$, $\wp_{1,1,2i-1}$, $i=1$, $2$,  associated with this curve
satisfy  the fundamental cubic relations: 
  $\forall u\in \Jac(\mathcal{V}) \backslash \Sigma$
\begin{align}
&\tfrac{1}{4} \wp_{1,1,1}(u)^2 = 
\wp_{1,1}(u)^2 (\wp_{1,1}(u) +\lambda_2) + \wp_{1,1}(u) \wp_{1,3}(u)  \notag \\
&\phantom{mmmmmmmmmmmmma} 
+ \wp_{3,3}(u) + \lambda_4 \wp_{1,1}(u) + \lambda_6, \notag\\
&\tfrac{1}{4} \wp_{1,1,1}(u) \wp_{1,1,3}(u) =
\wp_{1,1}(u) \wp_{1,3}(u) (\wp_{1,1}(u) +\lambda_2) \label{FundRelsG2} \\
&\phantom{mmmm}  + \tfrac{1}{2} \big(\wp_{1,3}(u)^2  -  \wp_{1,1}(u) \wp_{3,3}(u)
+ \lambda_4 \wp_{1,3}(u) + \lambda_8\big), \notag\\
&\tfrac{1}{4} \wp_{1,1,3}(u)^2 
= \wp_{1,3}(u)^2 (\wp_{1,1}(u) +\lambda_2)  -  \wp_{1,3}(u) \wp_{3,3}(u). \notag
\end{align}

Recalling that $\nu_6$ is defined by \eqref{eta3gGen} with 
\begin{align*}
&|\Dp| =   \tfrac{1}{2} \wp_{1,1,1}(s) (\wp_{1,3}(s) + e_2 e_4)
- \tfrac{1}{2} \wp_{1,1,3}(s) (\wp_{1,1}(s) - e_2 - e_4),&\\
&|\Np| = - \tfrac{1}{2} \wp_{1,1,1}(s) e_2 e_4 \wp_{1,3}(s)  (\wp_{1,1}(s) - e_2 - e_4)\\
&\qquad\qquad\quad + \tfrac{1}{2}  \wp_{1,1,3}(s) e_2 e_4 \big( \wp_{1,1}(s)(\wp_{1,1}(s) - e_2 - e_4)+
\wp_{1,3}(s) + e_2 e_4 \big), &
\end{align*}
and substituting these into \eqref{NNDDRel}, 
we obtain
\begin{multline}
 \lambda_{8}^{-1} |\Np| |\bar{\Np}| + \wp_{1,3} |\Dp| |\bar{\Dp}|  
 = \tfrac{1}{4}  \wp_{1,1,1}^2 
\big( 2 \wp_{1,3} \rma_{1,3} - \wp_{1,1} \rma_{3,3} \big)  \\
- \tfrac{1}{2}  \wp_{1,1,1} \wp_{1,1,3} 
\big( \wp_{1,3} \rma_{1,1}  +  \rma_{3,3}  \big) 
+ \tfrac{1}{4}  \wp_{1,1,3}^2 
\big(  \wp_{1,1} \rma_{1,1}  + 2  \rma_{1,3} \big) = 0,
\end{multline}
where $\rma_{i,j}(s) \equiv \tfrac{1}{4} \wp_{1,1,2i-1}(s) \wp_{1,1,2j-1}(s)$
and the argument $s$  is omitted for brevity.

\bigskip
\textbf{Genus 3.}
Let $\mathcal{V}$ be a genus three curve  of the form  \eqref{V22g1Eq}
with branch points $e_2$, $\bar{e}_2$, $e_4$,  $\bar{e}_4$, $e_6$, $\bar{e}_6$, and $0$.
Let the $\wp$-functions be associated with this curve. Then, the fundamental cubic relations hold:
  $\forall u\in \Jac(\mathcal{V}) \backslash \Sigma$
\begin{align}
&\tfrac{1}{4} \wp_{1,1,1}(u)^2 = 
\wp_{1,1}(u)^2 (\wp_{1,1}(u) +\lambda_2)  + \wp_{1,1}(u) \wp_{1,3}(u)  \notag \\
&\phantom{mmmmmmmmmmmmm} 
+ \wp_{3,3}(u)  -  \wp_{1,5}(u) + \lambda_4 \wp_{1,1}(u) + \lambda_6, \notag\\
&\tfrac{1}{4} \wp_{1,1,1}(u) \wp_{1,1,3}(u) =
\wp_{1,1}(u) \wp_{1,3}(u) (\wp_{1,1}(u) +\lambda_2) + \wp_{3,5}(u)  \notag \\
&\phantom{mmmm}  
+ \wp_{1,1}(u) \wp_{1,5}(u)   + \tfrac{1}{2} \big( \wp_{1,3}(u)^2  
- \wp_{1,1}(u) \wp_{3,3}(u)
+ \lambda_4 \wp_{1,3}(u) + \lambda_8\big), \notag\\
&\tfrac{1}{4} \wp_{1,1,1}(u) \wp_{1,1,5}(u) =
\wp_{1,1}(u) \wp_{1,5}(u)  (\wp_{1,1}(u) +\lambda_2) \notag \\
&\phantom{mmmm} 
+ \tfrac{1}{2} \big( \wp_{1,3}(u)  \wp_{1,5}(u) -  \wp_{1,1}(u) \wp_{3,5}(u) - \wp_{5,5}(u)
+ \lambda_4 \wp_{1,5}(u)\big), \label{FundRelsG3}  \\
&\tfrac{1}{4} \wp_{1,1,3}(u)^2 =  \wp_{1,3}(u)^2   (\wp_{1,1}(u) +\lambda_2)
+ 2 \wp_{1,3}(u) \wp_{1,5}(u) \notag \\
&\phantom{mmmmmmmmmmmmmm}  
-  \wp_{1,3}(u) \wp_{3,3}(u) + \wp_{5,5}(u) + \lambda_{10}, \notag \\
&\tfrac{1}{4} \wp_{1,1,3}(u) \wp_{1,1,5}(u) =
\wp_{1,3}(u) \wp_{1,5}(u) (\wp_{1,1}(u) +\lambda_2) \notag\\
&\phantom{mmmm}  + \wp_{1,5}(u)^2 + \tfrac{1}{2} \big( {-}\wp_{1,5}(u) \wp_{3,3}(u)
-  \wp_{1,3}(u) \wp_{3,5}(u) + \lambda_{12}\big), \notag \\
&\tfrac{1}{4} \wp_{1,1,5}(u)^2 =  \wp_{1,5}(u)^2   (\wp_{1,1}(u) +\lambda_2)
- \wp_{1,5}(u) \wp_{3,5}(u). \notag
\end{align}

In the definition \eqref{eta3gGen} of
$\nu_9$ we have
\begin{align*}
&\Pp(s) = \begin{pmatrix} 
\wp_{1,1}(s) & \wp_{1,1}(s)^2 + \wp_{1,3}(s) \\
\wp_{1,3}(s) & \wp_{1,1}(s)\wp_{1,3}(s)  + \wp_{1,5}(s) \\
\wp_{1,5}(s) & \wp_{1,1}(s) \wp_{1,5}(s) 
\end{pmatrix},\quad 
 \Rp(s) = \begin{pmatrix}  -\tfrac{1}{2} \wp_{1,1,1}(s)\\
 -\tfrac{1}{2} \wp_{1,1,3}(s)  \\  -\tfrac{1}{2} \wp_{1,1,5}(s)
 \end{pmatrix} \\
&{-}\bm{q}(s) = \begin{pmatrix} 
\tfrac{1}{2} \wp_{1,1,3}(s) + \tfrac{1}{2} \wp_{1,1,1}(s)  \wp_{1,1}(s) \\
\tfrac{1}{2} \wp_{1,1,5}(s) + \tfrac{1}{2} \wp_{1,1,1}(s)  \wp_{1,3}(s) \\
 \tfrac{1}{2} \wp_{1,1,1}(s)  \wp_{1,5}(s) 
\end{pmatrix}, \quad
\begin{array}{l}
\wp_{1,1}(\bm{C}) = e_2 + e_4 + e_6,\\
\wp_{1,3}(\bm{C}) = - e_2 e_4 - e_2 e_6 - e_4 e_6, \\
\wp_{1,5}(\bm{C}) = e_2 e_4 e_6,
\end{array} 
\end{align*}
and $\wp_{1,1,2i-1}(\bm{C}) \,{=}\, 0$ for $i\,{=}\,1$, $2$, $3$.
Introducing
$\rma_{i,j}(s) \equiv \tfrac{1}{4} \wp_{1,1,2i-1}(s) \wp_{1,1,2j-1}(s)$ allows
us to eliminate all three-index $\wp$-functions in \eqref{NNDDRel}
and express the resulting formula in terms of $\rma_{i,j} \rma_{k,l} - \rma_{i,k} \rma_{j,l}$, 
which vanish identically.
Actually,
\begin{multline}
 \lambda_{12}^{-1} |\Np| |\bar{\Np}| - \wp_{1,5} |\Dp| |\bar{\Dp}|  
 =  C_{1,1;3,3}(\rma_{1,1} \rma_{3,3} - \rma_{1,3}^2 )  \\
+ C_{1,1;3,5}(\rma_{1,1} \rma_{3,5} - \rma_{1,3} \rma_{1,5}) 
+ C_{1,1;5,5}(\rma_{1,1} \rma_{5,5} - \rma_{1,5}^2 ) \\
+ C_{1,3;3,5}(\rma_{1,3} \rma_{3,5} - \rma_{1,5} \rma_{3,3}) 
+ C_{1,3;5,5}(\rma_{1,3} \rma_{5,5} - \rma_{1,5} \rma_{3,5} )\\
+ C_{3,3;5,5}(\rma_{3,3} \rma_{5,5} - \rma_{3,5}^2 ) \equiv 0,
\end{multline}
where
\begin{align*}
C_{1,1;3,3} =&\; \wp_{1,5}(s)^3   (2\wp_{1,1}(s) +\lambda_2) \\
&+ \wp_{1,5}(s)^2 \big( \wp_{1,3}(\bm{C})  \wp_{1,3}(\bar{\bm{C}}) 
- (\wp_{1,5}(\bm{C}) + \wp_{1,5}(\bar{\bm{C}}))\wp_{1,1}(u)
 + \lambda_8  \big),\\
C_{1,1;3,5} = &\;{-}2 \wp_{1,5}(s)^3 - 2\wp_{1,5}(s)^2  \big( \wp_{1,3}(s) (2\wp_{1,1}(s) +\lambda_2) 
+ \wp_{1,5}(\bm{C}) + \wp_{1,5}(\bar{\bm{C}})\\
&+  (\wp_{1,3}(\bm{C}) + \wp_{1,3}(\bar{\bm{C}}) )  \wp_{1,1}(s)
 - \lambda_6 \big)
 + \wp_{1,5}(s) \wp_{1,3}(s) \big(\wp_{1,3}(\bm{C})  \wp_{1,3}(\bar{\bm{C}})  \\
& +(\wp_{1,5}(\bm{C}) + \wp_{1,5}(\bar{\bm{C}}) ) \wp_{1,1}(s)  - \lambda_8 \big) - 2 \lambda_{12} \wp_{1,5}(s),\\
C_{1,1;5,5} = &\; \wp_{1,5}(s)^2 \big( 2\wp_{1,3}(s) - \wp_{1,3}(\bm{C})  - \wp_{1,3}(\bar{\bm{C}})\big)
+ \wp_{1,5}(s) \big( \wp_{1,3}(s)^2 (2\wp_{1,1}(s) +\lambda_2) \\
& + \wp_{1,3}(s) \big({-} (\wp_{1,3}(\bm{C}) + \wp_{1,3}(\bar{\bm{C}}))\wp_{1,1}(s)
+ \lambda_6 \big) + \lambda_{10} \big),\\
C_{1,3;3,5} = &\; \wp_{1,5}(s)^2 \big(2 \wp_{1,3}(s)+
 \wp_{1,1}(s) (4\wp_{1,1}(s) + 3\lambda_2) + \wp_{1,1}(\bm{C})  \wp_{1,1}(\bar{\bm{C}}) + \lambda_4\big) \\
 &+ \wp_{1,5}(s) \big({-} (\wp_{1,5}(\bm{C}) + \wp_{1,5}(\bar{\bm{C}})) ( \wp_{1,1}(s)^2 + \wp_{1,3}(s) ) \\
& + (\wp_{1,3}(\bm{C})  \wp_{1,3}(\bar{\bm{C}})  + \lambda_8) \wp_{1,1}(s) \big),\\
C_{1,3;5,5} =&\; {-} \wp_{1,5}(s)^2 (2\wp_{1,1}(s) +\lambda_2) 
- \wp_{1,5}(s) \wp_{1,3}(s) \big(2 \wp_{1,3}(s) \\
&+ \wp_{1,1}(s) (4\wp_{1,1}(s) + 3\lambda_2) + 2 \lambda_4\big) 
 - \wp_{1,5}(s) \big((\wp_{1,3}(\bm{C}) + \wp_{1,3}(\bar{\bm{C}}))\wp_{1,1}(s)^2 \\
& + \lambda_6 \wp_{1,1}(s) + \wp_{1,3}(\bm{C}) \wp_{1,3}(\bar{\bm{C}}) + \lambda_8\big),\\
C_{3,3;5,5} = &\; \wp_{1,5}(s) \wp_{1,3}(s) (2\wp_{1,1}(s) +\lambda_2) 
+ \wp_{1,5}(s) \big(2 \wp_{1,1}(s)^3 + 2 \lambda_2 \wp_{1,1}(s)^2 \\
& + (\wp_{1,1}(\bm{C})  \wp_{1,1}(\bar{\bm{C}}) + \lambda_4 ) \wp_{1,1}(s) 
+ \wp_{1,5}(\bm{C}) + \wp_{1,5}(\bar{\bm{C}})  + \lambda_6 \big).
\end{align*}
The identities $\rma_{i,j} \rma_{k,l} - \rma_{i,k} \rma_{j,l} = 0$ expressed in terms of $2$-index $\wp$-functions
define the Kummer variety of $\mathcal{V}$.

 This completes the proof of Lemma~\ref{P:EtaEta}.


\end{document}